\documentclass{article}
\usepackage{amsfonts,amsmath,amsthm}
\usepackage[margin=1in]{geometry}
\usepackage{subfig}
\usepackage{bm}
\usepackage{graphicx}
\usepackage{xcolor}

\newtheorem{lemma}{Lemma} 
\newtheorem{proposition}{Proposition} 
\newtheorem{theorem}{Theorem} 
\newtheorem{remark}{Remark}
\newtheorem{definition}{Definition}

\DeclareMathOperator{\minmod}{minmod}

\title{Moment Methods for the 3D Radiative Transfer Equation Based on $\varphi$-Divergences}
\author{M.R.A.~Abdelmalik\footnote{Department of Mechanical Engineering, Eindhoven University of Technology, Groene Loper 3, 5612 AE Eindhoven, Netherlands}, 
Z.~Cai\thanks{Department of Mathematics, National University of Singapore, 10 Lower Kent Ridge Road, Singapore 119076.}~~and
T.~Pichard\footnote{CMAP, CNRS, École polytechnique, Institut polytechnique de Paris, 91120, Palaiseau, France}}
\date{\today}

\begin{document}

\maketitle

\begin{abstract}
    The method of moments is widely used for the reduction of kinetic equations into fluid models. It consists in extracting the moments of the kinetic equation with respect to a velocity variable, but the resulting system is a priori underdetermined and requires a closure relation. In this paper, we adapt the $\varphi$-divergence based closure, recently developed for rarefied gases i.e. with a velocity variable describing $\mathbb{R}^d$, to the radiative transfer equation where velocity describes the unit sphere $\mathbb{S}^2$. This closure is analyzed and a numerical method to compute it is provided. Eventually, it provides the main desirable properties to the resulting system of moments: Similarily to the entropy minimizing closure ($M_N$), it dissipates an entropy and captures exactly the equilibrium distribution. However, contrarily to $M_N$, it remains computationnally tractable even at high order and it relies on an exact quadrature formula which preserves exactly symmetry properties, i.e. it does not trigger ray effects. The purely anisotropic regimes (beams) are not captured exactly but they can be approached as close as desired and the closures remains again tractable in this limit.  
\end{abstract}


\section{Introduction}
\label{sec:intro}
This paper aims at constructing and analysing a new moment closure relation on the sphere. It is used for solving the so-called radiative transfer equation (RTE) that is a kinetic equation on the radiative intensity commonly used to describe the propagation of radiations through a medium. This equation corresponds to a transport in all direction but at a fixed velocity norm, i.e. the speed of light. It is supplemented with an absorption-emission term that is linear contrarily to Boltzmann collision term. Recent numerical developments for solving transport equation with fixed velocity includes applications for instance in astrophysics (see e.g.~\cite{astro1,astro2}), radiation therapy (see e.g.~\cite{pichard_thesis,kuepper_thesis,page_thesis}), hot plasma physics (see e.g.~\cite{duclous_thesis,dubroca_feugeas}) or neutron transport. 

The techniques that aim at solving directly such a kinetic equation are either computationnaly expensive, this is the case for the statistical methods (Direct Simulation Monte Carlo; DSMC) which converge slowly or with the deterministic direct methods (discrete ordinate methods; $S_N$) which require a large data storage, or they rely on approximation that are valid only in certain regimes, e.g. those based on ray modelling (pencil beam or ray tracing) or on diffusion theory. 

Due to the linearity, the method of moments is therefore particularly suitable for reducing the RTE. It corresponds to a Petrov-Galerkin approximation, i.e. the equation is integrated against polynomial basis functions and the unknown is approximated by a function of a linear combination of those basis functions. This technique is widely used on kinetic equations for applications in rarefied gases.    
Among the most popular of these moment technique, we may list 
\begin{itemize}
\item The discrete ordinate $S_N$ methods can be interpreted as a moment method using as an approximation function a sum of Diracs at fixed locations. In order to obtain a desired accuracy in all regimes, one generally requires a very large number of locations, i.e. a very large number of degree of freedom. Furthermore, this technique does not preserve potential symmetries in the solution (typically the isotropic solution is not captured) and create artifacts called ray effect.
\item The polynomial approximation $P_N$ consists in choosing the approximation function in the span of the basis function, this is a common Galerkin method. This technique leads to a linear system of moment equations and it is probably the simplest moment method to construct for this application. Also, it does preserve the symmetry in the solution (and the isotropic regime), but it approximate poorly the purely anisotropic regimes, i.e. the beams, which are commonly used in many applications. The only alternative to obtain a decent approximation in this regime consists in using a very large number of moments again, destroying the gain in computational costs.   
\item The entropy-based moment methods $M_N$ (\cite{Minerbo,dubroca_feugeas}) consists in choosing, among the admissible solutions, the one that minimizes a certain entropy. This construction is analogue to the one used in rarefied gases (\cite{levermore,junk,schneider,hauck_lev_tits}) and offers several desirable mathematical properties to the approximation, such as entropy dissipation and a symmetric hyperbolic structure. However, its use requires solving a very large number of optimization problems which do not have an analytical solution. Even if the exact $M_N$ approximation does preserve symmetry in the solution, computing it numerically generally requires an approximation (see typically \cite{Hauck,Hauck2,Hauck3} for efficient techniques) which eventually violates symmetry. Such a closure can also be computed analytically at order 1 (\cite{dubroca_feugeas}) or can be approximated up to order 2 (\cite{levermore_eddington,pichard_M2,Li_B2,groth_M2}), but high order multi-D models remain inaccessible. 
\item Other alternatives were developed recently and show promising results, but they mainly are 1D methods that can hardly be extended to multi-D problems (\cite{Monreal_thesis,Schneider_KN,PiN}).  
\end{itemize}
    
Among the properties that are looked for when constructing numerical methods, we can list low computational costs, the capture of purely anisotropic and isotropic regime, and especially exact symmetry preservation, and entropy dissipation. To the author's knowledge, no technique possesses all of these properties yet. 
In this direction, we design in this paper a technique based on the $\varphi$-divergence method of moments (\cite{Abdelmalik_thesis,abdelmalik}) that dissipates a given entropy, requires a fairly low computational cost to solve, preserves exactly the symmetry in the solution. It captures exactly the purely isotropic and it approximates correctly purely anisotropic regimes. 

In the next section, we recall some properties of the RTE and its solution. In the following, the method of moment is formulated in a new framework and the a moment closure is constructed. Section~\ref{sec:num_meth} presents the numerical method to compute this closure and to solve the moment system. It is followed by some numerical examples. The last section gathers concluding comments.

\section{Radiative transfer equation}
\label{sec:RTE}
\begin{subequations}
  \label{eq:RTE_pb}
  In this section, we recall some properties of the radiative transfer equation (RTE) and its solution from the litterature. Those properties are also studied in the next section after the moment extraction. The RTE yields
  \begin{equation} 
    \partial_t I + \Omega \cdot \nabla_x I = LI := \sigma\left(\frac{1}{4\pi}\int_{\mathbb{S}^2} I \,\mathrm{d}\Omega - I\right) 
    \label{eq:RTE}
  \end{equation}
  where $\Omega\in \mathbb{S}^2$, $x\in \mathbb{R}^3$ and $t\in ]0,T[$. The unknown $I$ is the radiative intensity which corresponds to an energy distribution function in the phase space $\mathbb{R}^3 \times \mathbb{S}^2$.
The space-time are non-dimensionalized in such a manner that the speed of propagation is unity, while it is generally set to $c$ the celerity of light for physical interpretations. 
      
  Since we consider only interactions of the radiations with the background, the source term is linear, and is chosen to be a linear Boltzmann operator with a cross-section $\sigma > 0$ assumed constant for simplicity. This equivals to a relaxation toward the isotropic distributions. 
  
  The system~\eqref{eq:RTE} is supplemented with an initial 
  \begin{gather}
    I(t=0) = I_0. \label{eq:RTE_IC}
  \end{gather}
 The boundary value problem at the kinetic level is well-documented (see e.g.~\cite{Dautray-Lions}), but its extension at the moment level remains an open problem, and this study is therefore postponed to future work. 
\end{subequations}

One first observes that the problem~\eqref{eq:RTE} that we want to solve is linear. Then we recall the following result which can be found e.g.~\cite{Dautray-Lions}.
\begin{proposition} \label{prop:WP_kin}
  Suppose that $\tau>0$ and that the initial and boundary conditions 
  \[ I_0 \in L^p(\mathbb{R}^3\times \mathbb{S}^2) 
  \quad\text{ for } \quad 1 \le p < \infty\]
  are non-negative. Then there exists a unique $I \in C([0,T]; \ L^p(\mathbb{R}^3 \times \mathbb{S}^2))$ satisfying~\eqref{eq:RTE_pb} in a weak sense. Furthermore, this solution is non-negative.
\end{proposition}
The proof suggested in~\cite{Dautray-Lions} relies on the semi-group theory with a lifting to include non-homogeneous boundary conditions. In the next section, we focus on a variational approach that is better suited for the method of moments. 

We specify properties satisfied by the solutions to~\eqref{eq:RTE} that are considered in the next section for the moment closure: 
\begin{itemize} 
\item \textbf{Energy conservation}: The radiative intensity corresponding to the 0-th order moment of $I$ is preserved in time 
\[ \partial_t \left(\int_{\mathbb S^2} I(\Omega) \,\mathrm{\Omega}\right) + \operatorname{div}_x  \left(\int_{\mathbb S^2} \Omega I(\Omega) \,\mathrm{d}\Omega\right) = 0. \]
This is the only collision invariant in the present case. 
\item \textbf{Entropy dissipation}: Considering any convex scalar function $\eta$, one formally computes
\begin{align}
\label{eq:entropy}
\partial_t \left( \int_{\mathbb{S}^2} \eta(I(\Omega)) \,\mathrm{d}\Omega \right) + \operatorname{div}_x \left( \int_{\mathbb{S}^2} \Omega \eta(I(\Omega)) \,\mathrm{d}\Omega\right) &= \frac{\sigma}{4\pi} \int_{\mathbb{S}^2\times\mathbb{S}^2} [\eta'(I(\Omega')) -  \eta'(I(\Omega))] I(\Omega) \,\mathrm{d}\Omega' \,\mathrm{d}\Omega,
\end{align}
where the right-hand-side is non-positive because, using the convexity of $\eta$, we have
\begin{align*}
\int_{\mathbb{S}^2\times\mathbb{S}^2} [\eta'(I(\Omega')) -  \eta'(I(\Omega))] I(\Omega) \,\mathrm{d}\Omega' \,\mathrm{d}\Omega &= -\int_{\mathbb{S}^2\times\mathbb{S}^2} [\eta'(I(\Omega')) -  \eta'(I(\Omega))] I(\Omega') \,\mathrm{d}\Omega' \,\mathrm{d}\Omega\\ 
&= -\frac{1}{2}\int_{\mathbb{S}^2\times\mathbb{S}^2} [\eta'(I(\Omega')) -  \eta'(I(\Omega))] (I(\Omega') - I(\Omega)) \,\mathrm{d}\Omega' \,\mathrm{d}\Omega \le 0.
\end{align*}

Furthermore, the right-hand-side is zero for all convex $\eta$ if and only if the distribution $I$ is isotropic, or equivalently if the right-hand-side of~\eqref{eq:RTE} is zero. In the following, we denote 
\[C = \left\{ I \,:\,I(\Omega) = \frac{1}{4\pi} \int_{\mathbb{S}^2} I(\Omega') \,\mathrm{d}\Omega' \right\},\]
the set of isotropic distributions. Especially, such solutions need to be captured or well-approximated by the numerical methods for~\eqref{eq:RTE_pb}. Other types of distributions are also of interest in many applications, typically the Dirac distributions $I(\Omega) = \delta(\Omega-\Omega_0)$ represent beams and are often considered for physical applications.  

\item \textbf{Rotational invariance}: Considering an orthogonal matrix $\mathcal O \in \mathbb{R}^3$ and a solution $I$ to~\eqref{eq:RTE}, then 
\[ (\partial_t I)(\mathcal O \Omega) = \partial_t ( I(\mathcal O \Omega)), \qquad (\Omega \cdot \nabla_x I)(\mathcal O \Omega) = (\mathcal O \Omega) \cdot \nabla_x (I(\mathcal O \Omega)), \qquad (LI)(\mathcal O \Omega) = L (I(\mathcal O \Omega)). \]
The translational invariance is debatable since one may fix Galilean (more commonly used) or Lorentzian (more physically relevant) invariance here (see e.g.~\cite{M1_relativist}). We do not focus on this issue and only consider rotational invariance in the following. 
\end{itemize}

To summarize, the list of the main properties of the RTE that we want to study at the moment level are: linearity of~\eqref{eq:RTE_pb}, well-posedness, positivity of the solution, energy conservation, entropy dissipation, isotropic and Dirac regime of the solution, rotational invariance. 

\section{Moment System Hierarchies}
In this section we derive a hierarchy of closed systems of moment equations from (\ref{eq:RTE}) subject to (\ref{eq:RTE_IC}). We base our derivation on a Galerkin approximation for the velocity variable of (\ref{eq:RTE}) and (\ref{eq:RTE_IC}) in re-normalized form. Such a re-normalization mapping is chosen to retain the entropy inequality (\ref{eq:entropy}) for some chosen entropy function $\eta(\cdot)$. In this work, we focus our attention on entropy functions that correspond to $\varphi-$divergences \cite{csiszar}, i.e.
\begin{equation}
    \eta(I) = \mathcal E(\Omega) \varphi\left(\frac{I(\Omega)}{\mathcal E(\Omega)}\right)
\end{equation}
where $\mathcal E(\Omega)$ is some known prior distribution independent of $t$ and $x$, and $\varphi(\cdot)$ is some convex function. Our choices of functions $\varphi(\cdot)$ are motivated by a sequence of corresponding moment closure models that bridge $P_N$ closures (see e.g.~\cite{pomraning_book,mihalas_book}) and $M_N$ closures (\cite{Minerbo,dubroca_feugeas}).

To derive the $\varphi-$divergence--based moment systems we first formulate (\ref{eq:RTE}) subject to (\ref{eq:RTE_IC}) weakly in Section~\ref{sec:weak} and proceed to derive the closed moment systems in Section~\ref{sec:phidiv}.

\subsection{Weak Formulation} \label{sec:weak}
To approximate weak-solutions of the radiative transfer equation, we consider the weak-form of (\ref{eq:RTE}) subject to (\ref{eq:RTE_IC}) in the $\Omega$ variable. Following Proposition~\ref{prop:WP_kin}, we seek $I \in D(L) = L^1(\mathbb{S}^2)$ satisfying for all scalar-valued test function $m$ from $V = L^\infty(\mathbb{S}^2)$
\begin{subequations}\label{eq:weak_form}
\begin{align}
\int_{\mathbb S^2} m(\Omega) \left[ 
\big( \partial_t + \Omega \cdot \nabla_{x} \big) I(t, x, \Omega) - L\big(I(t, x, \cdot)\big)(\Omega) 
\right] \, \text d \Omega &= 0,
\quad \forall t \in ]0,T[,\ \forall x \in \mathbb{R}^3,\label{eq:momrte} 
\end{align}
subject to
\begin{align}
\int_{\mathbb S^2} m(\Omega) \left[ 
 I(0,x,\Omega) - I_0(x,\Omega)
 \right] \, \text d \Omega &= 0,  
\quad \forall x \in \mathbb{R}^3. \label{eq:momic}
\end{align}
\end{subequations}
The system in (\ref{eq:weak_form}) may be conceived of as a system of equations, in the sense that carrying out the $\Omega$-integrals in (\ref{eq:weak_form}) yields a partial differential equation with corresponding initial-boundary value data for each $m \in V$. In the case we restrict the test functions to the basis functions of a finite-dimensional subspace $M$ of $V$. If $M$ is chosen as a polynomial space, then (\ref{eq:weak_form}) corresponds to a system (constituting a finite number) of partial differential moment equations. However, such a moment-system is not closed since $D(L)$ is infinite-dimensional. We aim to derive closed moment models that retain the salient properties of the (\ref{eq:RTE}), namely, rotational symmetry, conservation of energy, and entropy dissipation. The derivation of such a closed system of partial differential moment equations involves characterizing the test space $M$ and the specification of a moment-closure relation.

In order to derive conditions on $M$ we examine the weak formulation of the scattering operator,
\begin{equation}
q : \left\{\begin{array}{ll} D(L) \times V &\to \mathbb R, \\ (I, m) &\mapsto \displaystyle\int_{\mathbb S^2} m(\Omega) L(I)(\Omega) \, \text d \Omega. \end{array}\right.
\end{equation}
The conservation properties of $L$ amount to
\begin{equation}\label{eq:conservation-weak-collision}
q(I, m) = 0, \quad \forall I \in C,
\end{equation}
where $C$ is the set of isotropic functions over $\mathbb{S}^2$, and the rotational symmetry of $L$ is equivalent to
\begin{equation}\label{eq:galilean-invariance-weak-collision}
q(\mathcal T_{\mathcal O} I, \mathcal T_{\mathcal O} m) = q(I, m),
\end{equation}
for $I \in D(L)$, $m \in V$,
and where $(\mathcal T_{\mathcal O} f)(\Omega) = f({\mathcal O} \Omega)$ for orthogonal matrices $\mathcal O \in \mathbb R^{3 \times 3}$.

Restricting the test space to $M$ means that we restrict the domain of $q$ to $D(L) \times M$.
This gives rise to a restricted collision operator,
\begin{equation}
\pi L: \left\{ \begin{array}{ll} D(L) &\to M', \\ I &\mapsto \big(m\mapsto q(I, m) \big). \end{array}\right.
\end{equation}
Thus, the choice of $M'$ determines whether $\pi L$ inherits the salient properties of $L$.
First we note that for $\pi L$ to conserve the total energy, $M$ has to contain the space of all isotropic functions:
\begin{equation}\label{eq:col_invar}
C \subset M.
\end{equation}
Also, to retain rotational symmetry~\eqref{eq:galilean-invariance-weak-collision} it is necessary that any rotation by $\mathcal O \in \mathbb R^{3 \times 3}$ of a test function in $M$ is again an element of $M$,
\begin{equation}\label{eq:rot_invar}
\mathcal T_{\mathcal{O}} M \subset M.
\end{equation}
In the remainder of this work we consider spaces $M$ that are spanned by polynomials up to a certain degree $N \geq 1$. Such spaces satisfy conditions (\ref{eq:col_invar}) and (\ref{eq:rot_invar}). In the next section, we will consider the closure relation that preserves the dissipation of entropy.

\subsection{$\varphi-$divergence--based closure}\label{sec:phidiv}
To close the moment-system (\ref{eq:weak_form}) in such a way that the entropy dissipation property is retained, one may select a specific entropy function $\eta(\cdot)$ (which can be any smooth convex function in our case) and consider the renormalization mapping
\begin{equation}\label{eq:entmap}
I(t, x,\Omega) = \beta(g(t,x,\Omega)) := (\eta')^{-1}(g(t,x,\Omega)).
\end{equation}
The we can study the following Galerkin approximation of $g \in M$ in the moment-system (\ref{eq:weak_form}) \cite{levermore,dreyer1987}:
find $g \in M$ satisfying for all $m \in M$ 
\begin{subequations}\label{eq:ent_moments}
\begin{align}
\int_{\mathbb S^2} m(\Omega) \left[
\big( \partial_t + \Omega \cdot \nabla_{x} \big) \beta(g(t, x, \Omega)) - L\big(\beta(g(t, x, \cdot))\big)(\Omega) 
\right] \, \text d \Omega = 0, \quad \forall x \in \mathbb{R}^3;\label{eq:ent_mom_rte}
\end{align}
subject to
\begin{align}
\int_{\mathbb S^2} m(\Omega)\left[ 
\beta(g(0,x,\Omega)) - I_0(x,\Omega)
\right] \, \text d \Omega &= 0,  \quad \forall x \in \mathbb{R}^3.
\end{align}
\end{subequations}
Note that here we need $M \subset \beta^{-1}(D(L))$ in order that \eqref{eq:ent_mom_rte} is well defined.

In the remainder of this work we consider a sequence of $\varphi-$divergence based entropy functions (\cite{abdelmalik})
\begin{displaymath}\label{eq:phi_ent}
\eta_K(I) = \beta_K^{-1}(I) = K I \left( \frac{K}{K+1} I^{1/K} - 1\right), \quad K=1,3,5,\ldots
\end{displaymath}
associated with the corresponding sequence of renormalization maps
\begin{equation}
I = \beta_K(g) = \left( 1 + \frac{g}{K} \right)^K , \quad K=1,3,5,\ldots \label{eq:_def_betaK}
\end{equation}
For simplicity, given $N$ and $K$, we will name our new approach as the "$\beta_{N,K}$ model".

The convexity requirement, posed in (\ref{eq:entropy}), for such an entropy function follows from the positivity of its second derivative:
\[\eta_K''(I) = I^{\frac{1-K}{K}}\geq0 \quad \forall I \in \mathbb R.\]

To elucidate the relationship between the the $\varphi-$divergence based closures in (\ref{eq:_def_betaK}) corresponds to a Galerkin formulation of the $M_N$ closure (\cite{levermore,dubroca_feugeas}) which encompasses the $P_N$ closures (see e.g.~\cite{spectral_meth1,spectral_meth2}). We note that setting $K=1$ and approximating $g$ in a space of polynomials of degree $N$ leads to
\begin{equation}\label{eq:P}
\beta_1 (g) = 1+g
\end{equation}
which corresponds to $P_N$ closure and in the limit $K\rightarrow\infty$
\begin{equation}\label{eq:M}
\beta_{\infty} = \exp(g),
\end{equation}
which corresponds to $M_N$ with Boltzmann-Shannon entropy ($\eta(I) = I \log I - I$). All choices of $K$ guarantee the conservation of energy, rotational invariance and the dissipation of its own entropy $\eta_K$. However, they behave differently in the sense that the $M_N$ with Boltzmann entropy ($\beta_{N,\infty}$ method) can enforce the positivity of the intensity function, while the $P_N$ method ($\beta_{N,1}$ method) can be viewed as a polynomial spectral method on the sphere, may provide non-positive results. Nevertheless, due to its inherent linearity, it is easier to implement and more widely used since the moment inversion problem is trivial.

For intermediate choices of $1<K<\infty$, we expect to get an ``intermediate" model in which the intensity function may have enforced positivity on a wider part of $\mathbb S^2$ than the $P_N$ representation and the moment inversion problem to be less stiff than the $M_N$ one. 

Choosing a basis $\boldsymbol{m}$ of $M$ and the renormalization map~\eqref{eq:_def_betaK}, system \eqref{eq:ent_moments} rewrites in a conservative hyperbolic form 
\begin{subequations}
\label{eq:betaK_system}
    \begin{equation}
        \partial_t \boldsymbol{\rho} + \operatorname{div}_x \boldsymbol{F}(\boldsymbol{\rho}) = \boldsymbol{L} (\boldsymbol{\rho})
    \end{equation}
    where $\boldsymbol{\rho}$, $\boldsymbol{F}(\boldsymbol{\rho})$ and $\boldsymbol{L} (\boldsymbol{\rho})$ are related through a distribution of the form~\eqref{eq:_def_betaK} by the formula
    \begin{align}
        \label{eq:mom_inv}
        \boldsymbol{\rho} &= \int_{\mathbb{S}^2} \boldsymbol{m}(\Omega) \left(1 + \frac{\boldsymbol{\lambda}(\boldsymbol{\rho}) \cdot \boldsymbol{m}(\Omega)}{K} \right)^K \,\mathrm{d}\Omega, \\
        \boldsymbol{F}(\boldsymbol{\rho}) &= \int_{\mathbb{S}^2} \Omega \boldsymbol{m}(\Omega) \left(1 + \frac{\boldsymbol{\lambda}(\boldsymbol{\rho}) \cdot \boldsymbol{m}(\Omega)}{K} \right)^K \,\mathrm{d}\Omega, \label{eq:flux_betaK}\\
        \boldsymbol{L}(\boldsymbol{\rho}) &= \sigma \left( \frac{\rho_0}{4\pi} \int_{\mathbb{S}^2} \boldsymbol{m}(\Omega) d\Omega - \boldsymbol{\rho}\right),
    \end{align}
    and $\rho_0$ is the zeroth moment of the intensity function. In the definition of the flux function, the operator $\boldsymbol{\lambda}(\boldsymbol{\rho})$ can be regarded as the ``moment inversion operator'', which is implicitly defined by \eqref{eq:mom_inv}. In the next section, we will show that this operator is well defined.
\end{subequations}

\begin{remark}
We remark that the renormalization mappings in (\ref{eq:_def_betaK}) can alternatively be conceived of as an optimizer for the following optimization problem:
\begin{displaymath}
\begin{aligned}
& \operatorname*{argmin}_{I} \eta(I) \\
& \text{subject to } \langle \boldsymbol{m} I \rangle = \boldsymbol{\rho}.
\end{aligned}
\end{displaymath}
The solution to this optimization problem is chosen as $I^*$, which is fully determined by $\boldsymbol{\rho}$, and thus leads to the moment closure. When $\eta$ is convex, the optimal $I$ that minimizes $\eta$ as the form
\begin{displaymath}
I^*(\boldsymbol{\rho}; \Omega) = \beta(\boldsymbol{\lambda}(\boldsymbol{\rho}) \cdot \boldsymbol{m}),  
\end{displaymath}
where $\beta(\cdot)$ is the inverse function of $\eta'(\cdot)$, and $\boldsymbol{\lambda}$ is the Lagrange multiplier to be determined by the moment constraints:
\begin{equation} \label{eq:moment_constraints}
\int_{\mathbb{S}^2} \boldsymbol{m} \beta(\boldsymbol{\lambda}(\boldsymbol{\rho}) \cdot \boldsymbol{m}) \,\mathrm{d}\Omega = \boldsymbol{\rho}.
\end{equation}
\end{remark}
In the remainder of this section we establish the well-posedness properties of the moment system $\varphi-$divergence-based moment systems (\ref{eq:ent_moments}).

\subsection{Well-posedness of the moment inversion problem}
In this subsection, we show that the problem~\eqref{eq:mom_inv} defines a diffeomorphism from the set of moment vectors of distributions of the form~\eqref{eq:_def_betaK}, also named the realizability domain, and the associated set of Lagrange multipliers. Especially, we show that those two sets are both $\mathbb{R}^r$ where $r = (N+1)^2$ is the number of the number of considered moments. 

Let $\{m_1(\Omega), \cdots, m_r(\Omega)\}$ be a basis of the polynomial space up to degree $N$, and define $\boldsymbol{m}(\Omega) = \left(m_1(\Omega), \cdots, m_r(\Omega) \right)^{\top}$ and
\begin{displaymath}
  \boldsymbol{M}(\boldsymbol{\lambda}) = \int_{\mathbb{S}^2} \boldsymbol{m}(\Omega) \left( 1 + \frac{\boldsymbol{\lambda} \cdot \boldsymbol{m}(\Omega)}{K} \right)^K \mathrm{d}\Omega.
\end{displaymath}
Here we always assume that $K$ is odd and positive. 
The moment inversion problem can be formulated as
\begin{displaymath}
\boldsymbol{M}(\boldsymbol{\lambda}) = \boldsymbol{\rho}
\end{displaymath}
for a given $\boldsymbol{\rho} \in \mathbb{R}^r$.

\begin{lemma} \label{lem:coercivity}
Let $\{\boldsymbol{\lambda}_k\} \subset \mathbb{R}^r$ be a sequence of vectors such that $\|\boldsymbol{\lambda}_k\| \rightarrow +\infty$. Then the vectors $\boldsymbol{\rho}_k = \boldsymbol{M}(\boldsymbol{\lambda}_k)$
satisfy $\|\boldsymbol{\rho}_k\| \rightarrow +\infty$.
\end{lemma}

\begin{proof}
Let $g'_k(\Omega) = 1 + K^{-1} \boldsymbol{\lambda}_k \cdot \boldsymbol{m}_k(\Omega)$. It is clear that $\|g_k'\|_{L^2(\mathbb{S}^2)} \rightarrow +\infty$, and for every $k$, we can find a vector $\boldsymbol{\lambda}_k'$ such that $\boldsymbol{\lambda}_k' \cdot \boldsymbol{m}(\Omega) = g'_k(\Omega)$. Thus,
\begin{displaymath}
\boldsymbol{\lambda}_k' \cdot \boldsymbol{\rho}_k = \int_{\mathbb{S}^2} [ g'_k(\Omega)]^{K+1} \,\mathrm{d}\Omega = \|g_k'\|_{L^{K+1}(\mathbb{S}^2)}^{K+1},
\end{displaymath}
where we have used the fact that $K$ is odd. Since the linear space for all polynomials up to degree $N$ has a finite dimension, there exists constants $C_1$ and $C_2$ depending only on $N$ and $K$ such that
\begin{displaymath}
\|\boldsymbol{\lambda}_k'\| \leq C_1 \|\boldsymbol{\lambda}_k' \cdot \boldsymbol{m}\|_{L^2(\mathbb{S}^2)} = C_1 \|g_k'\|_{L^2(\mathbb{S}^2)}, \qquad \|g_k'\|_{L^2(\mathbb{S}^2)} \leq C_2 \|g_k'\|_{L^{K+1}(\mathbb{S}^2)}.
\end{displaymath}
Consequently,
\begin{displaymath}
\frac{\boldsymbol{\lambda}_k'}{\|\boldsymbol{\lambda}_k'\|} \cdot \boldsymbol{\rho}_k = \frac{\|g_k'\|_{L^{K+1}(\mathbb{S}^2)}^{K+1}}{\|\boldsymbol{\lambda}_k'\|} \geq \frac{\|g_k'\|_{L^2(\mathbb{S}^2)}^{K}}{C_1 C_2^{K+1}} \rightarrow +\infty.
\end{displaymath}
Meanwhile, we have
\begin{displaymath}
\frac{\boldsymbol{\lambda}_k'}{\|\boldsymbol{\lambda}_k'\|} \cdot \boldsymbol{\rho}_k \leq\|\boldsymbol{\rho}_k\|.
\end{displaymath}
Therefore, $\|\boldsymbol{\rho}_k\| \rightarrow +\infty$.
\end{proof}

\begin{proposition}
For any $\boldsymbol{\rho}^* \in \mathbb{R}^{r}$, there exists a unique vector $\boldsymbol{\lambda}^* \in \mathbb{R}^{r}$ such that
$\boldsymbol{\rho}^* = \boldsymbol{M}(\boldsymbol{\lambda^*})$.
\end{proposition}

\begin{proof}
We first show the uniqueness. Suppose there are vectors $\boldsymbol{\lambda}_1^*$ and $\boldsymbol{\lambda}_2^*$ satisfying
\begin{displaymath}
\boldsymbol{\rho}^* = \int_{\mathbb{S}^2} \boldsymbol{m}(\Omega) \left( 1 + \frac{\boldsymbol{\lambda}^*_1 \cdot \boldsymbol{m}(\Omega)}{K} \right)^{K} \,\mathrm{d}\Omega, \qquad
\boldsymbol{\rho}^* = \int_{\mathbb{S}^2} \boldsymbol{m}(\Omega) \left( 1 + \frac{\boldsymbol{\lambda}^*_2 \cdot \boldsymbol{m}(\Omega)}{K} \right)^{K} \,\mathrm{d}\Omega.
\end{displaymath}
Taking the difference of these two equations and applying the mean value theorem yield
\begin{displaymath}
\left[ \int_{\mathbb{S}^2} \boldsymbol{m}(\Omega) [\boldsymbol{m}(\Omega)]^{\top} \left( 1 + \frac{[\xi \boldsymbol{\lambda}_1^* + (1-\xi) \boldsymbol{\lambda}_2^*] \cdot \boldsymbol{m}(\Omega)}{K} \right)^{K-1} \,\mathrm{d}\Omega \right] (\boldsymbol{\lambda}_2^* - \boldsymbol{\lambda}_1^*) = 0,
\end{displaymath}
where $\xi \in (0,1)$. Since $K$ is odd and the polynomials in $\boldsymbol{m}(\Omega)$ are linearly independent, the matrix inside the square brackets is symmetric positive definite. Therefore, $\boldsymbol{\lambda}_2^* - \boldsymbol{\lambda}_1^*$ is zero, showing the uniqueness of the solution.

We now show the existence. By the uniqueness proven in the previous paragraph, the map $\boldsymbol{M}: \mathbb{R}^r \rightarrow \mathbb{R}^r$ is injective. It is clear that $\boldsymbol{M}$ is a continuous map. Therefore, by the domain invariance theorem \cite{Brouwer1912}, the range of $\boldsymbol{M}$ (denoted by $\mathcal{R}$ hereafter) is an open set. If $\mathcal{R} \neq \mathbb{R}^r$, then we can find $\boldsymbol{\rho}_{\infty} \in \partial \mathcal{R}$ and a sequence $\{\boldsymbol{\rho}_k\} \subset \mathcal{R}$ such that $\boldsymbol{\rho}_k \rightarrow \boldsymbol{\rho}_{\infty}$. For every $k$, since $\boldsymbol{\rho}_k \in \mathcal{R}$, we can find $\boldsymbol{\lambda}_k \in \mathbb{R}^r$ such that $\boldsymbol{M}(\boldsymbol{\lambda}_k) = \boldsymbol{\rho}_k$. We now consider two cases:
\begin{enumerate}
\item If the sequence $\{\boldsymbol{\lambda}_k\}$ is bounded, it has a convergent subsequence $\{\boldsymbol{\lambda}_{k_j}\}$. Assume that its limit is $\boldsymbol{\lambda}_{\infty}$. Then by the continuity of $\boldsymbol{M}$, we have $\boldsymbol{M}(\boldsymbol{\lambda}_{\infty}) = \boldsymbol{\rho}_{\infty}$. This contradicts our assumption $\boldsymbol{\rho}_{\infty} \not\in \mathcal{R}$.
\item If the sequence $\{\boldsymbol{\lambda}_k\}$ is unbounded, it has a subsequence $\{\boldsymbol{\lambda}_{k_j}\}$ such that $\|\boldsymbol{\lambda}_{k_j}\| \rightarrow +\infty$. By Lemma~\ref{lem:coercivity}, we have $\|\boldsymbol{\rho}_{k_j}\| \rightarrow +\infty$. This contradicts our assumption $\boldsymbol{\rho}_k \rightarrow \boldsymbol{\rho}_{\infty}$.
\end{enumerate}
Therefore, the range of $\boldsymbol{\rho}$ must be $\mathcal{R}$, which completes the proof of existence.
\end{proof}

\subsection{Symmetric dissipative hyperbolicity}
Since the map between $\boldsymbol{\lambda}$ and $\boldsymbol{\rho}$ is invertible, the moment equations \eqref{eq:betaK_system} can also be formulated as equations of $\boldsymbol{\lambda}$. This can be directly observed from \eqref{eq:betaK_system}, and the results have the following form:
\begin{subequations}\label{eq:moments-galerkin}
\begin{equation}\label{eq:moments-galerkin-pde}
\bm A_0(\bm \lambda) \partial_t \bm \lambda + \bm A_i(\bm \lambda) \partial_{x_i} \bm \lambda = \bm s(\bm \lambda),
\end{equation}
where the matrices $\bm A_0, \dots, \bm A_3$ are given by
\begin{align}
\bm A_0(\bm \lambda) &= \int_{\mathbb S^2} \beta_K'(\bm \lambda \cdot \bm m(\Omega)) \bm m(\Omega) [\bm m(\Omega)]^\top \, \text d \Omega, \\
\bm A_i(\bm \lambda) &= \int_{\mathbb S^2} \Omega_i \beta_K'(\bm \lambda \cdot \bm m(\Omega)) \bm m(\Omega) [\bm m(\Omega)]^\top \, \text d \Omega, \quad i=1,2,3.
\end{align}
\end{subequations}
These equations are equivalent to \eqref{eq:betaK_system} for smooth solutions.

To demonstrate the well-posedness of the Cauchy initial value problem (\ref{eq:moments-galerkin}) in the sense of \cite{Kawashima}, we show that (\ref{eq:moments-galerkin}) conforms to the so called symmetric dissipative hyerpbolic systems defined as follows:
\begin{definition}
A system of $r \in \mathbb N$ first order partial differential equations
\begin{equation}\label{eq:def_symm_diss_hyp}
\partial_t \bm w (\bm u) + \partial_{x_i} \bm F_i(\bm w(\bm u)) = \bm B_0(\bm u) \partial_t \bm u + \bm B_i(\bm u) \partial_{x_i} \bm u = \bm c(\bm w(\bm u))
\end{equation}
posed for functions $u$ with open, convex codomain $\mathcal U \subseteq \mathbb R^r$ and
\begin{equation}
\mathcal M = \big\{ \bm \psi \in \mathbb R^r: \bm \psi \cdot \bm s(\bm w(\bm u)) = 0, \, \forall \bm u \in \mathcal U\big\}
\end{equation}
is called symmetric dissipative if
\begin{enumerate}
  \item $\bm B_0(\bm u)$ is symmetric positive definite;
  \item $\bm B_i(\bm u)$, $i=1,2,3$, are symmetric;
  \item $\bm s(\bm w(\bm u)) = 0$ if and only if $\bm u \in \mathcal M$; and
  \item the linearization $\bm \nabla \bm s(\bm w(\bm u))$ in $\bm u \in \mathcal M$ is symmetric and nonpositive definite, and its null space equals $\mathcal M$.
\end{enumerate}
\end{definition}

By showing that (\ref{eq:moments-galerkin}) conforms to symmetric dissipative hyperbolic systems we imply the linear well-posedness and, moreover, that under suitable conditions on the initial data, local-in-time existence of solutions can be established viz. the following theorem due to \cite{Kawashima}

\begin{theorem}
Suppose the balance laws (\ref{eq:def_symm_diss_hyp}) are symmetric dissipative hyperbolic. Let  $\bar{\bm w} \in \mathcal W$, where $\mathcal W\subset\mathbb R^r$ is open and convex, be a constant state such that $\bm s( \bar{\bm w} ) = 0$. If the initial data $\bm w_0(\bm x)$ satisfy $\bm{w}_0(x) - \bar{\bm w} \in H^s(\mathbb R^r)$ with an integer $s\geq [d/2] + 2$ and take values in a compact subset of $\mathcal W$, then there exists $T > 0$ such that the corresponding Cauchy problem for (\ref{eq:def_symm_diss_hyp}) has a unique solution $\bm w = \bm w(t,\bm x)$ satisfying $\bm w - \bar{\bm w} \in C([0,T];H^s(\mathbb{R}^r))$.
\end{theorem}

For moment equations derived from dissipative kinetic theories, the symmetric dissipativity is a natural structure that preserves fundamental properties such as convergence to equilibrium and entropy dissipation when the solution is near the manifold $\mathcal{M}$. Numerically, this indicates stable simulations for suitably constructed numerical schemes.  For our $\beta_{N,K}$ models, we set
\begin{equation}
\mathcal U = \mathbb R^r
\quad 
\text{and}
\quad
\mathcal M = \{ \bm \psi \in \mathbb R^r : \bm \psi \cdot \bm m \in C\}.
\end{equation}
We note that the of the terms appearing in the symmetric dissipative hyperbolicity definition conform to
\[
\bm w  = \bm \rho, \quad \bm u = \bm \lambda, \quad \bm c = \bm s
\]
\begin{proposition}
System~\eqref{eq:moments-galerkin-pde} is symmetric dissipative.
\end{proposition}
\begin{proof}
One observes that $\mathcal U$ is indeed open and convex and that $\bm \psi \cdot \bm L(\bm \lambda) = 0$ for all $\bm \psi \in \mathcal M$.

To satisfy Conditions $1.$ and $2.$ we note that the symmetry of $\bm A_0$ and $\bm A_i$ is evident, and the positive definiteness of $\bm A_0$ follows from the fact that 
\begin{equation}
\beta_K'(g) = M \left( 1 + \frac{g}{K} \right)^{K - 1}
\end{equation}
is non-negative as $K$ is odd, so for all $\bm b \in \mathbb R^r$,
\begin{equation}
\bm b^\top \bm A_0(\bm \lambda) \bm b
= \int_{\mathbb S^2} \beta_K'(\bm \lambda \cdot \bm m)
\big( \bm b \cdot \bm m \big)^2 \, \text d  \Omega \geq 0
\end{equation}
and equality if and only if the integrand vanishes. 

Condition $3.$ rewrites $\boldsymbol{L} (\bm M(\bm \lambda)) = 0$ if and only if $\bm \lambda \cdot \bm m \in C$. By definition of $\boldsymbol{L}$ and by positivity of $\sigma$, then $\boldsymbol{L}(\bm \rho) = 0$ equivals to requiring that $\frac{\bm\rho}{\rho_0} = 4\pi \int_{\mathbb{S}^2} \bm m$. Then, by uniqueness of the moment inversion, this yields a unique representation of $\frac{\bm \rho}{\rho_0}$. Since the invariant space $C$ is composed only of isotropic functions, then this representation $\beta_K(\bm \lambda \cdot \bm m)$ satisfying $\bm \rho = \int \bm m \beta_K(\bm \lambda \cdot \bm m)$ is isotropic and therefore $\bm \lambda \cdot \bm m \in C$ and $\bm \lambda \in \mathcal M$.

For Condition $4.$, one computes \[ \bm s(\bm \lambda) = \frac{\sigma}{4\pi} \int_{\mathbb{S}^2\times\mathbb{S}^2} (\bm m (\Omega') - \bm m (\Omega)) \beta_K(\bm \lambda \cdot \bm m(\Omega)) \,\mathrm{d}\Omega \,\mathrm{d}\Omega'\]
which provides
\[ \bm \nabla \bm s(\bm \lambda) = -\frac{\sigma}{8\pi} 
 \int_{\mathbb{S}^2\times\mathbb{S}^2} (\bm m (\Omega') - \bm m (\Omega))(\bm m (\Omega') - \bm m (\Omega))^\top \beta_K'(\bm \lambda \cdot \bm m(\Omega)) \,\mathrm{d}\Omega \,\mathrm{d}\Omega', \]
 which is non-positive due to the non-positivity of $\beta_K'$ and its kernel coincides is the set of $\bm V$ such that \[ \int_{\mathbb{S}^2\times\mathbb{S}^2} \left( (\bm m (\Omega') - \bm m (\Omega)) \cdot \bm V \right)^2\beta_K'(\bm \lambda \cdot \bm m(\Omega)) d\Omega d\Omega' = 0. \] 
 Since $\beta_K'(\bm \lambda \cdot \bm m(\Omega))$ is strictly positive then this requires $(\bm m (\Omega') - \bm m (\Omega)) \cdot \bm V$ to be uniformly zero. This is only possible if $\bm m \cdot \bm V$ is isotropic and therefore for $\bm V\in C$. 
\end{proof}
\begin{proposition}
The characteristic speeds of~\eqref{eq:betaK_system} are bounded by 1.
\end{proposition}
\begin{proof}
Using the equivalent form of the moment equations \eqref{eq:moments-galerkin}, the characteristic speeds of the moment system in the direction $n \in \mathbb{S}^2$ are solutions of the following generalized eigenvalue problem:
\begin{displaymath}
\lambda \boldsymbol{A}_0 \boldsymbol{v} = \sum_{i=1}^3 n_i \boldsymbol{A}_i.
\end{displaymath}
By the definitions of the matrices $\boldsymbol{A}_i$, this equation holds only when the following matrix is singular:
\begin{displaymath}
\int_{\mathbb{S}^2} (\lambda - \Omega \cdot n) \boldsymbol{m}(\Omega)[\boldsymbol{m}(\Omega)]^{\top}\left(1 + \frac{\boldsymbol{\lambda}(\boldsymbol{\rho}) \cdot \boldsymbol{m}(\Omega)}{K} \right)^{K-1} \,\mathrm{d}\Omega.
\end{displaymath}
For $\Omega \in \mathbb{S}^2$, it holds that $|\Omega \cdot n| < 1$ almost everywhere. Therefore, if $K$ is odd and $\lambda \geq 1$ ($\lambda \leq -1$), the matrix above is symmetric positive (negative) definite. Thus, we know that all the characteristic speeds of the moment equations lie in the open interval $(-1,1)$.
\end{proof}
This makes it convenient for us to choose suitable time steps in the numerical scheme. Remark that a more accurate computation of the characteristic speed is available through this proof.

\subsection{Properties of $\varphi-$divergence--based moment systems}
Coming back to the properties of interests for moment models, the $\beta_{N,K}$ closure satisfies: 
\begin{itemize}
\item \textbf{Well-defined}: The $\beta_{K}$ reconstruction is well-defined for all vectors $\bm \rho \in \mathcal R = \mathbb R^r$ through the polynomial moment inversion operator~\eqref{eq:_def_betaK}. Even though the original PDE~\eqref{eq:RTE} is linear, the present reconstruction corresponds to a non-linear approximation (except in the case $K=1$).   
\item \textbf{Realizability}: The positivity of the underlying distribution function $\beta_{K}$ is not enforced. This yields a realizability domain $\mathcal R = \mathbb{R}^r$, i.e. a set of $\bm \rho \in\mathbb{R}^r$ that possesses a representation of the form $\beta_K$. Forcing the positivity of this representation is known to reduce the realizability domain to a strict subset of $\mathbb R^m$ (see e.g.~\cite{Kershaw,Lasserre_book,Schmuedgen_book})and to yield a moment inversion problem that becomes singular along the boundary of this restricted realizability domain and ill-conditionned close to it. 
\item \textbf{Well-posedness initial-value problem}: The symmetric dissipative structure of the moment system provides the existence and uniqueness of a solution to the initial-value problem. The case with boundaries is not covered by this theory and its study is left for future work.  
\item \textbf{Convergence to equilibrium}: Similarly, the symmetric dissipative structure also provides the dissipation of the convex entropy $\eta_K$, which minimum coincides with the set of isotropic distribution $C$. 
\item \textbf{Characteristic speed}: The Jacobians $A_0$ and $A_i$ being defined as moments of polynomial functions, the characteristic speeds can be computed accurately and are all bounded by $1$. 
\item \textbf{Approximation of the physical regimes}: Since the set $C$ of isotropic functions is part of the approximation space, such functions are exactly captured by the $\beta_{N,K}$ closure. Considering the purely anisotropic regimes represented by Diracs in $\Omega$, those cannot be represented exactly by a representation of the form $\beta_K$ at fixed $K$, even in the limit $|\bm \lambda| \rightarrow +\infty$. However, the distance (in a certain sense defined in the next sections) of Dirac to the set of functions of this form can be controlled through parameter $K$. Especially, Dirac measures can be retreived in the limit $\lim_{K\rightarrow \infty} \beta_K$. 
\item \textbf{Conservation of energy and rotational invariance}: those also hold with the $\beta_{N,K}$ closure.  
\end{itemize}

Now, comparing this construction with state-of-th-art models: As $K$ increases from $1$ to $+\infty$, our $\beta_{N,K}$ models connect the classical $P_N$ and $M_N$ methods with a sequence of moment systems. When $K = 1$, we have $\beta_K(x) = 1+x$, and thus the ansatz of the intensity function is a polynomial of degree $N$, which coincides with the $P_N$ method. When $K \rightarrow +\infty$, the limit of $\beta_K(x)$ is $\exp(x)$, which agrees with the hypothesis in the $M_N$ method based on the Boltzmann entropy. The entire sequence of models shares many good properties of both $P_N$ and $M_N$ models, e.g. the conservation laws, rotational invariance, and the entropy dissipation. Another classical model, known as the discrete ordinates method, is not covered in this series, but the $\beta_{N,K}$ models may exhibit some behaviors similar to the $S_N$ model when the integrals in the moment inversion problem are computed inexactly with numerical integration. Below we will provide a brief comparison between our models and these classical methods. 

\begin{itemize}
\item \textbf{$\boldsymbol{\beta_{N,K}}$ vs $\boldsymbol{P_N}$:} Both models have an unbounded realizability domain due to their permission of negative parts in the intensity function, which is practically more convenient since there is no need to guarantee the realizability of moments during the computation. For positively realizable moments, due to the higher similarity between the $\beta_{N,K}$ model and the $M_N$ model, the $\beta_{N,K}$ model is more likely to generate an intensity function with a larger range of positive value over $\mathbb S^2$ compared to $P_N$. This can help reduce spurious oscillations when approximating singular intensity functions. However, compared with the $P_N$ method, the non-linearity of the $\beta_{N,K}$ models may lead to less accuracy when approximating smooth intensity functions. We will show such examples in Section \ref{sec:num}.
\item \textbf{$\boldsymbol{\beta_{N,K}}$ vs $\boldsymbol{M_N}$:} Compared with the $M_N$ model with the common choices of entropy (Boltzmann $\eta(I) = I\log I - I$ or Bose-Einstein $\eta(I) = (I+1)\log(I+1) - I\log I$), the $\beta_{N,K}$ models are considerably easier to implement, due to the possibility to compute exact values of the moments. As we will elaborate in the next section, the moment inversion problem will be solved by Newton's method, in which the integral
\begin{displaymath}
\int_{\mathbb{S}^2} \boldsymbol{m}(\Omega)\boldsymbol{m}(\Omega)^{\top} \beta'(\boldsymbol{\lambda} \cdot \boldsymbol{m}(\Omega)) \,\mathrm{d}\Omega
\end{displaymath}
needs to be calculated. In the $\beta_{N,K}$ model, the integrand is a polynomial, so that the exact integral can be obtained by numerical quadrature. However, in the $M_N$ model, the integrand involves an exponential or Planck function, which cannot be exactly integrated numerically. The advantage of the $M_N$ model is mainly theoretical: it uses a physical entropy function, guarantees the positivity, and is able to describe beams exactly. However, even with the theoretical possibility, capturing beams in the numerical scheme of the $M_N$ method, with high order $N$, is highly challenging since the intensity function cannot be expressed as the exponential of a bounded polynomial and the moment inversion problem turns singular.
\item \textbf{$\boldsymbol{\beta_{N,K}}$ vs $\boldsymbol{S_N}$:} The $S_N$ method, also known as the discrete ordinates method, approximates the intensity function by $N$ beams at fixed angles. Compared to our $\beta_{N,K}$ models, the $S_N$ method is easier to implement and can preserve positivity. However, the space of such intensity functions is not rotationally invariant, and therefore the $S_N$ models also fail to preserve rotational invariance, due to which the numerical solutions may exhibit ray effect when simulating isotropic radiative sources \cite{Garrett2013comparison}. The $\beta_{N,K}$ models are advantageous from this aspect. Note that in \cite{Garrett2013comparison}, the numerical results of the $M_N$ model also shows some ray effects due to the numerical integration. In fact, when the numerical integration provides inexact results, the rotational invariance is ruined, and the directions specified by the quadrature nodes are artificially preferred, which is similar to the $S_N$ model. Similar behaviors are also expected in the simulations of the $\beta_{N,K}$ models if the integrals are not computed exactly. Nevertheless, one can always avoid the ray effect by taking sufficient quadrature nodes to ensure that all integrals are exactly computed.
\end{itemize}

\section{Numerical Method}
\label{sec:num_meth}
To better understand the $\beta_{N,K}$ models, we will perform some numerical experiments to test their performances. In this section, we will introduce the numerical methods for our experiments, and the numerical results will be reported in the next section.

\subsection{Moment inversion problem}
\label{sec:numerical_method}
Simulating the radiative transfer problem using the $\beta_{N,K}$ model requires solving the moment inversion problem numerically so that the numerical flux can be calculated. In our implementation, instead of using monomials to define the moments, we adopt ``orthogonal moments'' defined by the spherical harmonics:
\begin{displaymath}
\rho_{lm} := \int_{\mathbb{S}^2} Y_{lm}(\Omega) I(\Omega) \,\mathrm{d}\Omega.
\end{displaymath}
The real spherical harmonics $Y_{lm}(\cdot)$ satisfies the orthogonality:
\begin{displaymath}
\int_{\mathbb{S}^2} Y_{lm}(\Omega) Y_{l'm'}(\Omega) \,\mathrm{d}\Omega = \delta_{ll'} \delta_{mm'}.
\end{displaymath}
Then the moment inversion problem can be stated as follows:
\begin{quote} \it
Given the moments $\rho_{lm}$, $l = 0,1,\cdots,N$, $m = -l, \cdots,l$, find coefficients $\lambda_{lm}$, $l = 0,1,\cdots,N$, $m = -l, \cdots,l$ such that
\begin{displaymath}
\int_{\mathbb{S}^2} Y_{lm}(\Omega) \left( 1 + \frac{1}{K} \sum_{l'=0}^N \sum_{m'=-l'}^{l'} \lambda_{l'm'} Y_{l'm'}(\Omega) \right)^K \,\mathrm{d}\Omega = \rho_{lm}.
\end{displaymath}
\end{quote}

In our implementation, we apply Newton's method to solve the moment inversion problem. Using $\lambda_{lm}^{(n)}$ to denote the coefficients at the $n$th time step, we update the solution from the $n$th step to the $(n+1)$th step by solving the following linear system of $\lambda_{lm}^{(n+1)}$:
\begin{multline*}
  \sum_{l'=0}^N \sum_{m'=-l'}^{l'} \left(\lambda_{l'm'}^{(n+1)} - \lambda_{l'm'}^{(n)} \right) \int_{\mathbb{S}^2} Y_{l'm'}(\Omega) Y_{lm}(\Omega) \left( 1 + \frac{1}{K} \sum_{l''=0}^N \sum_{m''=-l''}^{l''} \lambda_{l''m''}^{(n)} Y_{l''m''}(\Omega) \right)^{K-1} \,\mathrm{d}\Omega = \Delta_{lm}, \\
  l=0,1,\cdots,N, \qquad m = -l, \cdots, l.
\end{multline*}
where
\begin{displaymath}
  \Delta_{lm} = \rho_{lm} - \int_{\mathbb{S}^2} Y_{lm}(\Omega)\left( 1 + \frac{1}{K} \sum_{l''=0}^N \sum_{m''=-l''}^{l''} \lambda_{l''m''}^{(n)} Y_{l''m''}(\Omega) \right)^K \,\mathrm{d}\Omega.
\end{displaymath}
One can observe that the linear system above has a symmetric coefficient matrix, and therefore can be solved by the conjugate gradient method.
The Newton iteration terminates when the $L^2$ difference of the moments is smaller than $10^{-10}$:
\begin{displaymath}
  \sqrt{\sum_{l=0}^N \sum_{m=-l}^l |\Delta_{lm}|^2} < 10^{-10}.
\end{displaymath}
The integrals appearing in the linear system are calculated using the Lebedev quadrature. Regardless of the round-off error, the Lebedev quadrature can guarantee the exactness of the numerical integration with sufficient number of quadrature points. The nodes and weights of the Lebedev quadrature for different algebraic orders of accuracy can be found in a series of papers including \cite{Lebedev1976quadratures, Lebedev1999quadrature}. Here we need the order of accuracy to be at least $N(K+1)$ to guarantee the exactness of the numerical integral. The linear system can be solved by the conjugate gradient method due to the positive definiteness of the coefficient matrix. In the time-dependent problem, the initial values $\alpha_{lm}^{(0)}$ are chosen as the solution at the previous time step; otherwise, we set
\begin{displaymath}
  \lambda_{lm}^{(0)} = \rho_{lm}^{(0)}, \qquad l = -N, \cdots, N, \quad m = -l, \cdots, l
\end{displaymath}
as the start of our iterations.

After solving the coefficients, the computation of the flux function can again be obtained by the Lebedev quadrature. In the $\beta_{N,K}$ model, the flux function for the moment $M_{lm}$ is the following integral:
\begin{displaymath}
\int_{\mathbb{S}^2} \Omega \, Y_{lm}(\Omega) \left( 1 + \frac{1}{K} \sum_{l'=0}^N \sum_{m'=-l'}^{l'} \lambda_{l'm'} Y_{l'm'}(\Omega) \right)^K \,\mathrm{d}\Omega,
\end{displaymath}
which requires the order of accuracy to be $(K+1)N+1$ to calculate exactly. In Table \ref{tab:n_quad_point}, we list the number of quadrature points required in the Lebedev quadrature for some values of $N$ and $K$.
\begin{table}[!ht]
\centering
\caption{Number of quadrature points needed in the $\beta_{N,K}$ model}
\label{tab:n_quad_point}
\begin{tabular}{ccc@{\hspace*{35pt}}ccc@{\hspace*{35pt}}ccc}
\hline
$N$ & $K$ & No. of points & $N$ & $K$ & No. of points &
$N$ & $K$ & No. of points \\
\hline
3 & 1 & 26 & 7 & 1 & 86 & 11 & 1 & 194 \\
3 & 3 & 74 & 7 & 3 & 302 & 11 & 3 & 770 \\
3 & 5 & 146 & 7 & 5 & 770 & 11 & 5 & 1730 \\
3 & 7 & 230 & 7 & 7 & 1202 & 11 & 7 & 2702 \\
\hline
\end{tabular}
\end{table}

\subsection{Spatial and temporal discretization}
In our experiments, the finite volume method is adopted to discretize the moment equations. This work focuses only on two-dimensional problems, so that we can write the moment equations in the following form of balance laws:
\begin{displaymath}
\frac{\partial \boldsymbol{\rho}}{\partial t} + 
\frac{\partial \boldsymbol{F}(\boldsymbol{\rho})}{\partial x} +
\frac{\partial \boldsymbol{G}(\boldsymbol{\rho})}{\partial y} = \boldsymbol{L}(\boldsymbol{\rho}).
\end{displaymath}
The spatial domain is discretized with a uniform grid, and each grid cell is denoted by $[x_{i-1/2}, x_{i+1/2}] \times [y_{j-1/2}, y_{j+1/2}]$. Then according to the finite volume method, the numerical solution $\boldsymbol{\rho}_{i,j}^n$ approximates the average of the solution at the $n$th time step $t_n$:
\begin{displaymath}
\boldsymbol{\rho}_{i,j}^n \approx \frac{1}{\Delta x \,\Delta y} \int_{x_{i-1/2}}^{x_{i+1/2}} \int_{y_{j-1/2}}^{y_{j+1/2}} \boldsymbol{\rho}(x,y,t_n) \,\mathrm{d}y \,\mathrm{d}x.
\end{displaymath}
To update $\boldsymbol{\rho}_{i,j}^n$, we use the numerical scheme below following Heun's method:
\begin{align*}
& \boldsymbol{\rho}_{i,j}^* = \boldsymbol{\rho}_{i,j}^n - \frac{\Delta t}{\Delta x} (\boldsymbol{F}_{i+1/2,j}^n - \boldsymbol{F}_{i-1/2,j}^n)
- \frac{\Delta t}{\Delta y} (\boldsymbol{G}_{i,j+1/2}^n - \boldsymbol{G}_{i,j-1/2}^n) + \Delta t \, \boldsymbol{L}(\boldsymbol{\rho}_{i,j}^n), \\
& \boldsymbol{\rho}_{i,j}^{n+1} = \frac{\boldsymbol{\rho}_{i,j}^n + \boldsymbol{\rho}_{i,j}^*}{2} - \frac{\Delta t}{2\Delta x} (\boldsymbol{F}_{i+1/2,j}^* - \boldsymbol{F}_{i-1/2,j}^*)
- \frac{\Delta t}{2\Delta y} (\boldsymbol{G}_{i,j+1/2}^* - \boldsymbol{G}_{i,j-1/2}^*) + \frac{\Delta t}{2} \, \boldsymbol{L}(\boldsymbol{\rho}_{i,j}^*).
\end{align*}
The numerical fluxes $\boldsymbol{F}_{i+1/2,j}^n$, $\boldsymbol{F}_{i+1/2,j}^*$ and $\boldsymbol{G}_{i,j+1/2}^n$, $\boldsymbol{G}_{i,j+1/2}^*$ are computed based on linear reconstructions:
\begin{align*}
\boldsymbol{F}_{i+1/2,j}^n &=  \boldsymbol{F} \left(\boldsymbol{\rho}_{i,j}^n + \frac{\Delta x}{2} \boldsymbol{\sigma}_{i,j}^n, \boldsymbol{\rho}_{i+1,j}^n - \frac{\Delta x}{2} \boldsymbol{\sigma}_{i+1,j}^n\right), &
\boldsymbol{F}_{i+1/2,j}^* &=  \boldsymbol{F} \left(\boldsymbol{\rho}_{i,j}^* + \frac{\Delta x}{2} \boldsymbol{\sigma}_{i,j}^*, \boldsymbol{\rho}_{i+1,j}^* - \frac{\Delta x}{2} \boldsymbol{\sigma}_{i+1,j}^*\right), \\
\boldsymbol{G}_{i,j+1/2}^n &=  \mathbf{G} \left(\boldsymbol{\rho}_{i,j}^n + \frac{\Delta y}{2} \boldsymbol{\kappa}_{i,j}^n, \boldsymbol{\rho}_{i,j+1}^n - \frac{\Delta y}{2} \boldsymbol{\kappa}_{i,j+1}^n\right), &
\boldsymbol{G}_{i,j+1/2}^* &=  \boldsymbol{G} \left(\boldsymbol{\rho}_{i,j}^* + \frac{\Delta y}{2} \boldsymbol{\kappa}_{i,j}^*, \boldsymbol{\rho}_{i,j+1}^* - \frac{\Delta y}{2} \boldsymbol{\kappa}_{i,j+1}^*\right),
\end{align*}
where the slopes $\boldsymbol{\sigma}_{i,j}^n$, $\boldsymbol{\sigma}_{i,j}^*$ and $\boldsymbol{\kappa}_{i,j}^n$, $\boldsymbol{\kappa}_{i,j}^*$ are obtained from the monotonized central limiter, for example:
\begin{displaymath}
\boldsymbol{\sigma}_{i,j}^n= \minmod\left( \frac{2(\boldsymbol{\rho}_{i+1,j}^n - \boldsymbol{\rho}_{i,j}^n)}{\Delta x}, \, \frac{2(\boldsymbol{\rho}_{i,j}^n - \boldsymbol{\rho}_{i-1,j}^n)}{\Delta x}, \, \frac{\boldsymbol{\rho}_{i+1,j}^n - \boldsymbol{\rho}_{i-1,j}^n}{2\Delta x} \right).
\end{displaymath}
We adopt the Lax-Friedrichs fluxes in our implementation:
\begin{displaymath}
\boldsymbol{F}(\boldsymbol{\rho}, \boldsymbol{\varrho}) = \frac{\boldsymbol{F}(\boldsymbol{\rho}) + \boldsymbol{F}(\boldsymbol{\varrho})}{2} - \frac{\boldsymbol{\varrho} - \boldsymbol{\rho}}{2}, \qquad \boldsymbol{G}(\boldsymbol{\rho}, \boldsymbol{\varrho}) = \frac{\boldsymbol{G}(\boldsymbol{\rho}) + \boldsymbol{G}(\boldsymbol{\varrho})}{2} - \frac{\boldsymbol{\varrho} - \boldsymbol{\rho}}{2}. 
\end{displaymath}
Note that here we have used the fact that the characteristic speeds are less than one to determine the numerical viscosity. Note that the moment inversion problem needs to be solved when calculating the flux functions $\boldsymbol{F}(\cdot)$ and $\boldsymbol{G}(\cdot)$. The time step $\Delta t$ is determined such that
\begin{displaymath}
\Delta t \left( \frac{1}{\Delta x} + \frac{1}{\Delta y} + \sigma \right) < 1.
\end{displaymath}

The method above will have second-order accuracy for smooth solutions. Here we choose this scheme because it is relatively easy to implement. As in other hyperbolic equations, higher-order schemes may have better performances. In applications, one can apply any general high-order schemes such as the WENO method and the discontinuous Galerkin method to the moment equations. In the next section, we will apply the method to some benchmark problems. More details of the numerical method such the cell sizes and the boundary conditions will be specified in each example.

\section{Numerical Results} \label{sec:num}
\subsection{Approximation of some density functions}
In this section, we study the approximation of $I(\Omega)$ using the $\beta_{N,K}$ model. Since one major advantage of the $M_N$ model is to represent beams (Dirac functions) exactly, we will also test the capability of the $\beta_{N,K}$ model in approximating beam-related functions. Three numerical examples will be presented in the following subsections.

\subsubsection{Approximation of a single beam}
We first study the approximation of the Dirac function defined on the sphere:
\begin{displaymath}
I(\Omega) = \delta(\Omega - \Omega_0),
\end{displaymath}
where $\Omega_0$ is a given point on $\mathbb{S}^2$. Due to the rotational invariance of the $\beta_{N,K}$ model, any point $\Omega_0$ is equivalent in this test. Below we choose $\Omega_0 = (0,0,1)^{\top}$ so that the moments of $I(\Omega)$ are
\begin{displaymath}
\rho_{lm} := \int_{\mathbb{S}^2} Y_{lm}(\Omega) I(\Omega) \,\mathrm{d}\Omega = \delta_{m0} \sqrt{\frac{2l+1}{4\pi}}.
\end{displaymath}
The moment inversion is solved by Newton's method as described in Section \ref{sec:numerical_method}.

\begin{figure}[!ht]
\centering
\subfloat[$N=3$, $K=1$]{%
  \includegraphics[width=.33\textwidth]{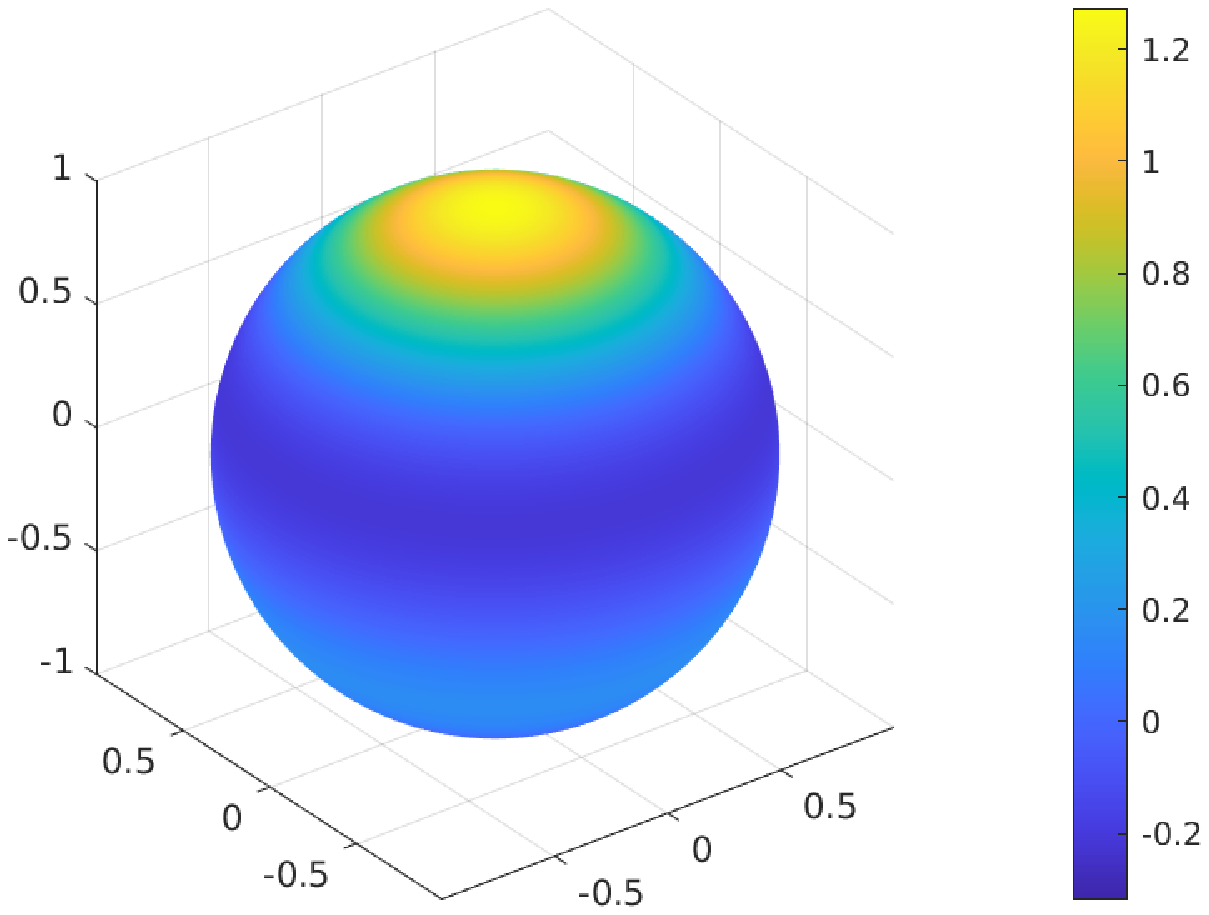}
}
\subfloat[$N=7$, $K=1$]{%
  \includegraphics[width=.33\textwidth]{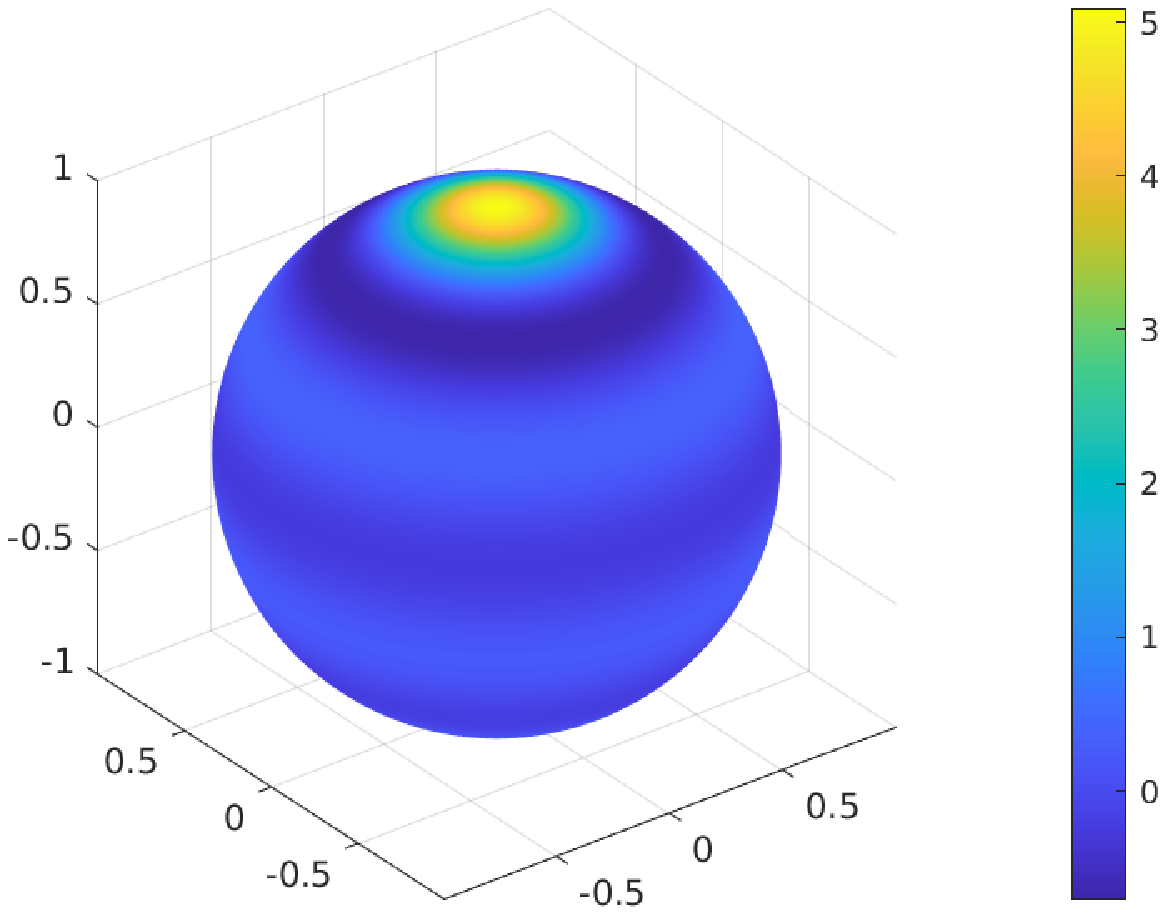}
}
\subfloat[$N=11$, $K=1$]{%
  \includegraphics[width=.33\textwidth]{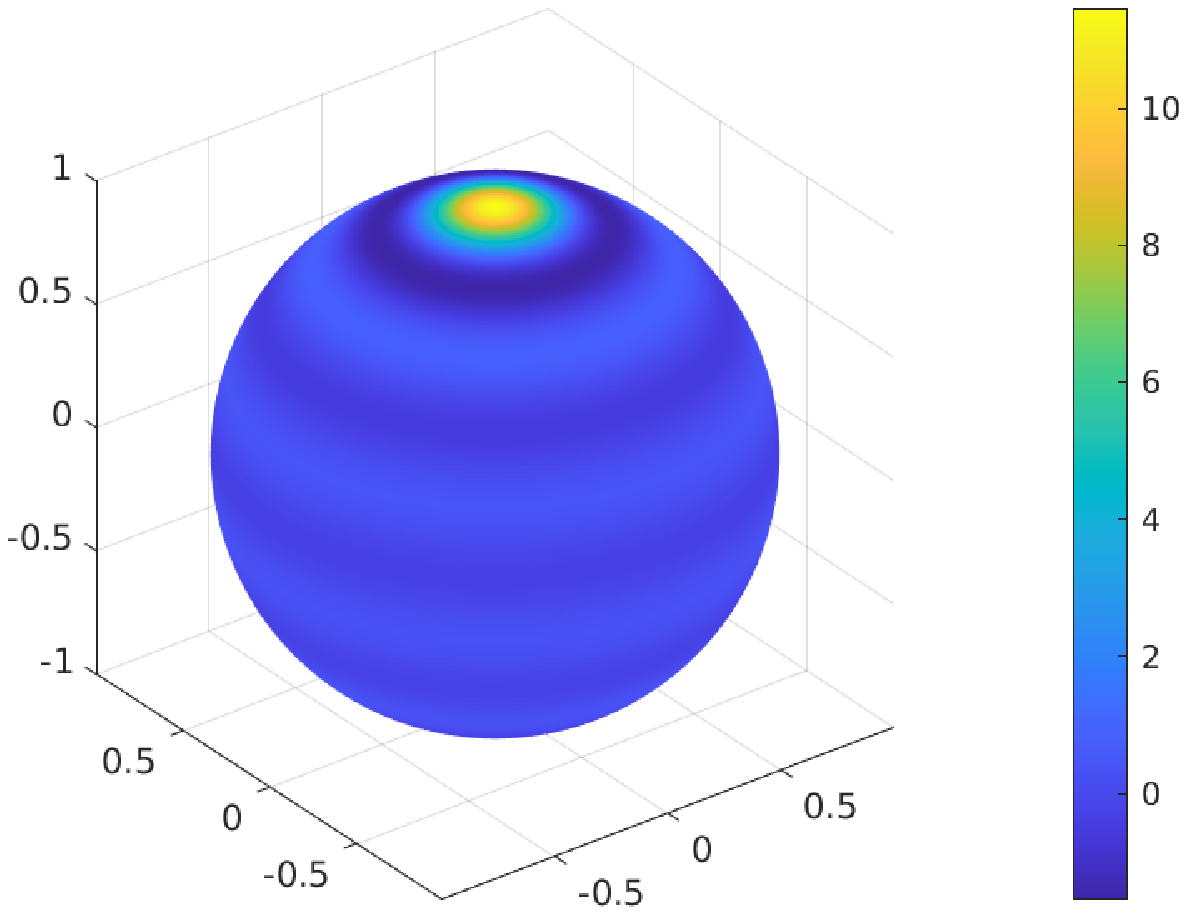}
} \\
\subfloat[$N=3$, $K=3$]{%
  \includegraphics[width=.33\textwidth]{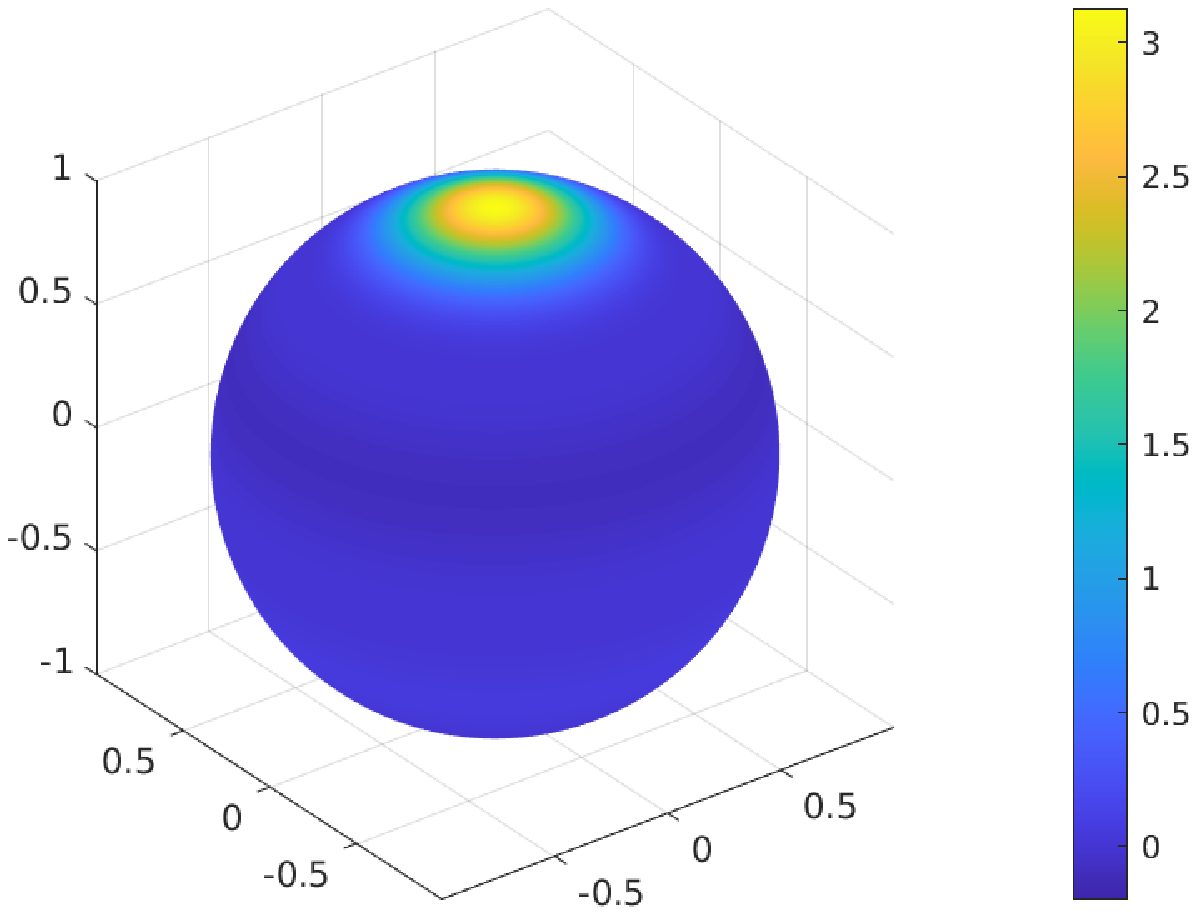}
}
\subfloat[$N=7$, $K=3$]{%
  \includegraphics[width=.33\textwidth]{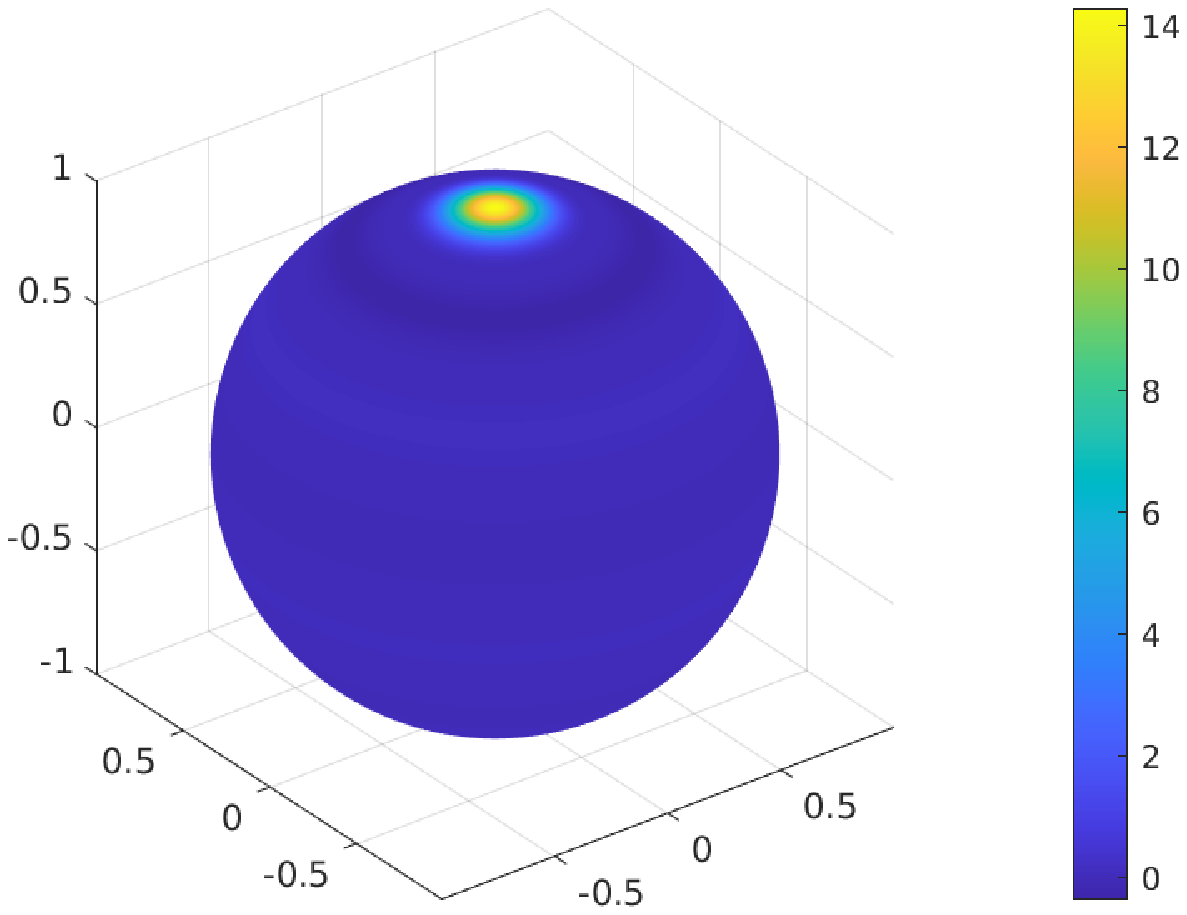}
}
\subfloat[$N=11$, $K=3$]{%
  \includegraphics[width=.33\textwidth]{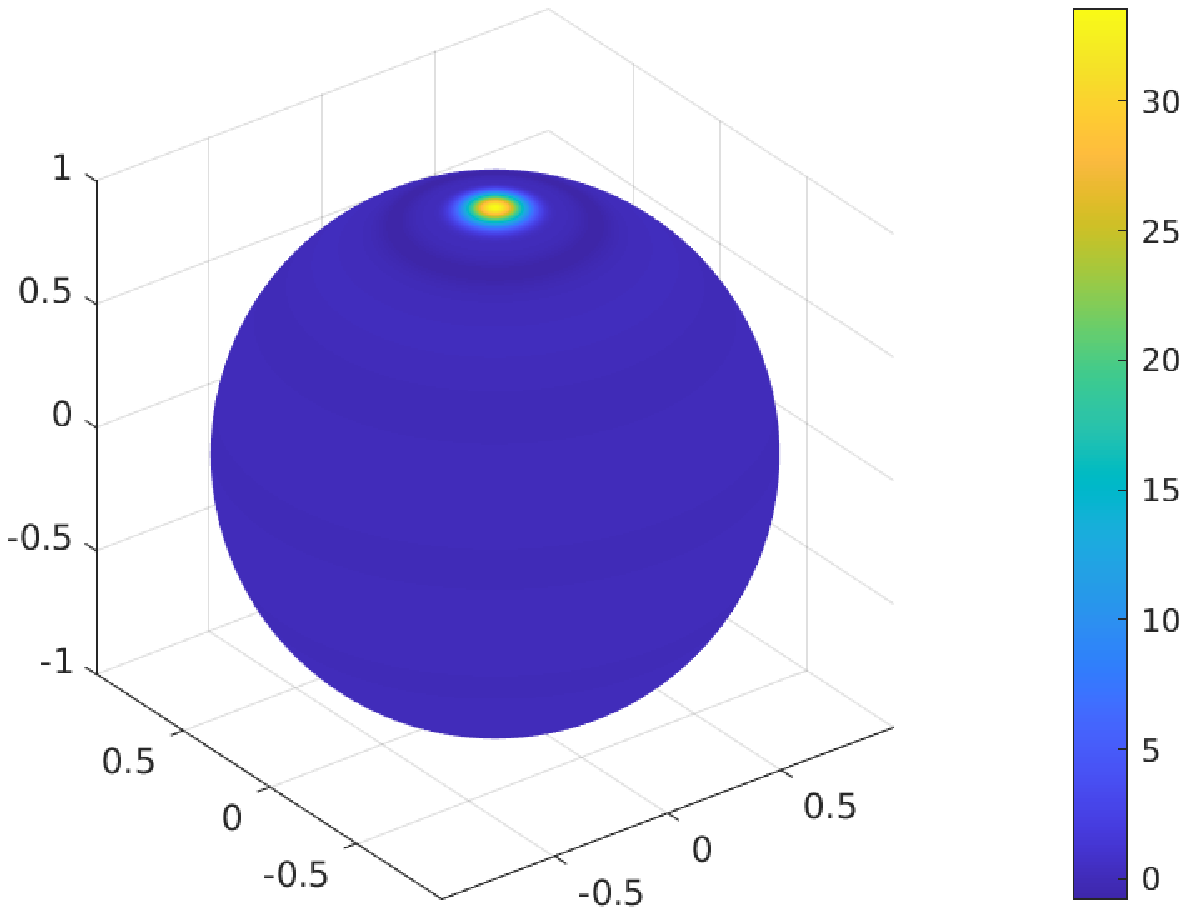}
} \\
\subfloat[$N=3$, $K=5$]{%
  \includegraphics[width=.33\textwidth]{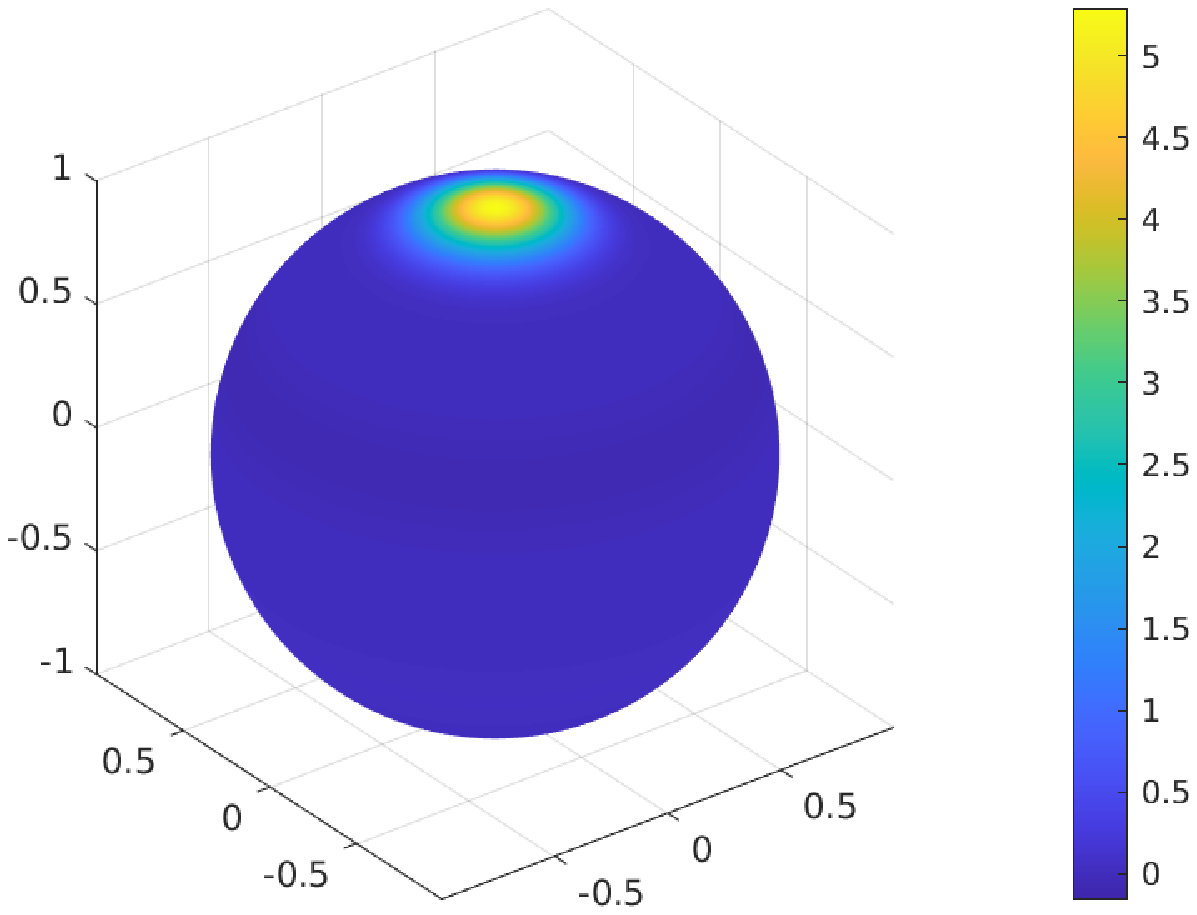}
}
\subfloat[$N=7$, $K=5$]{%
  \includegraphics[width=.33\textwidth]{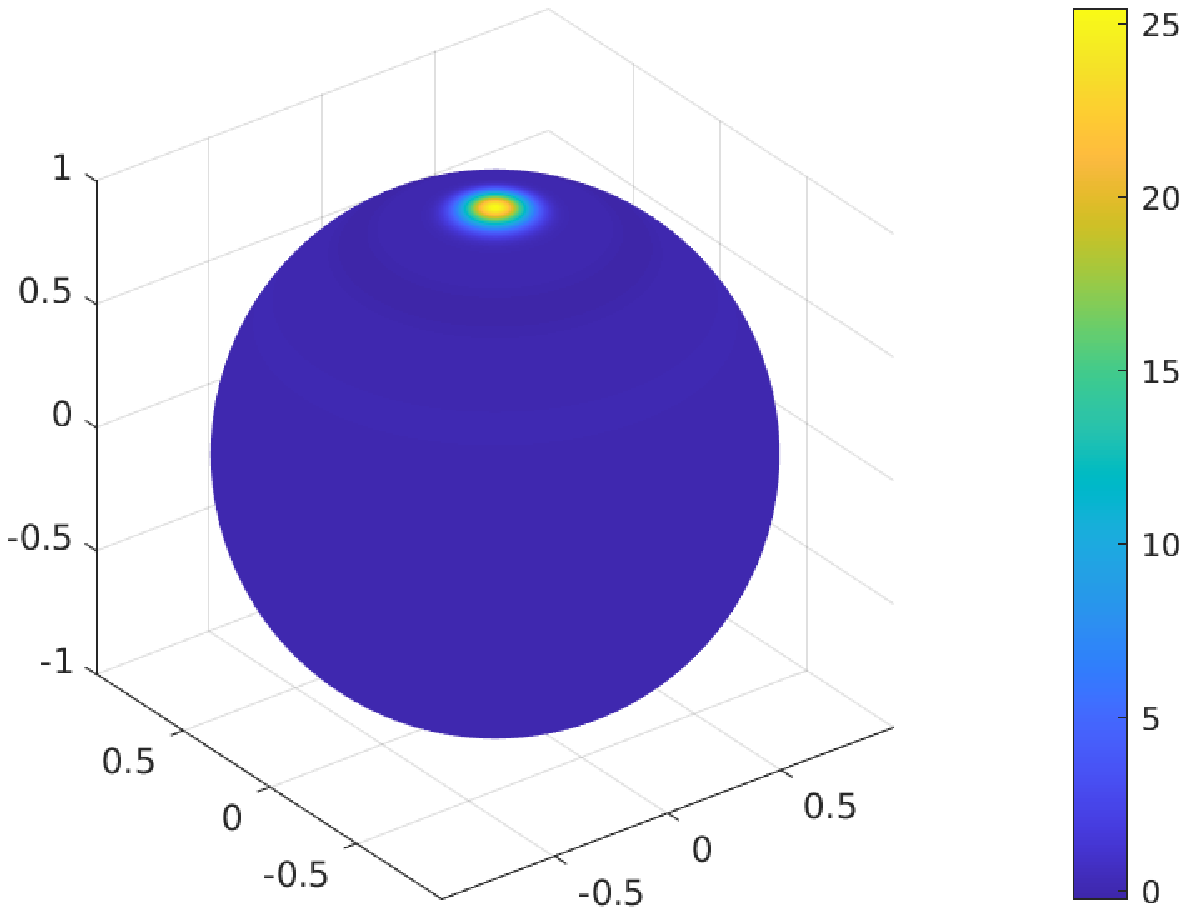}
}
\subfloat[$N=11$, $K=5$]{%
  \includegraphics[width=.33\textwidth]{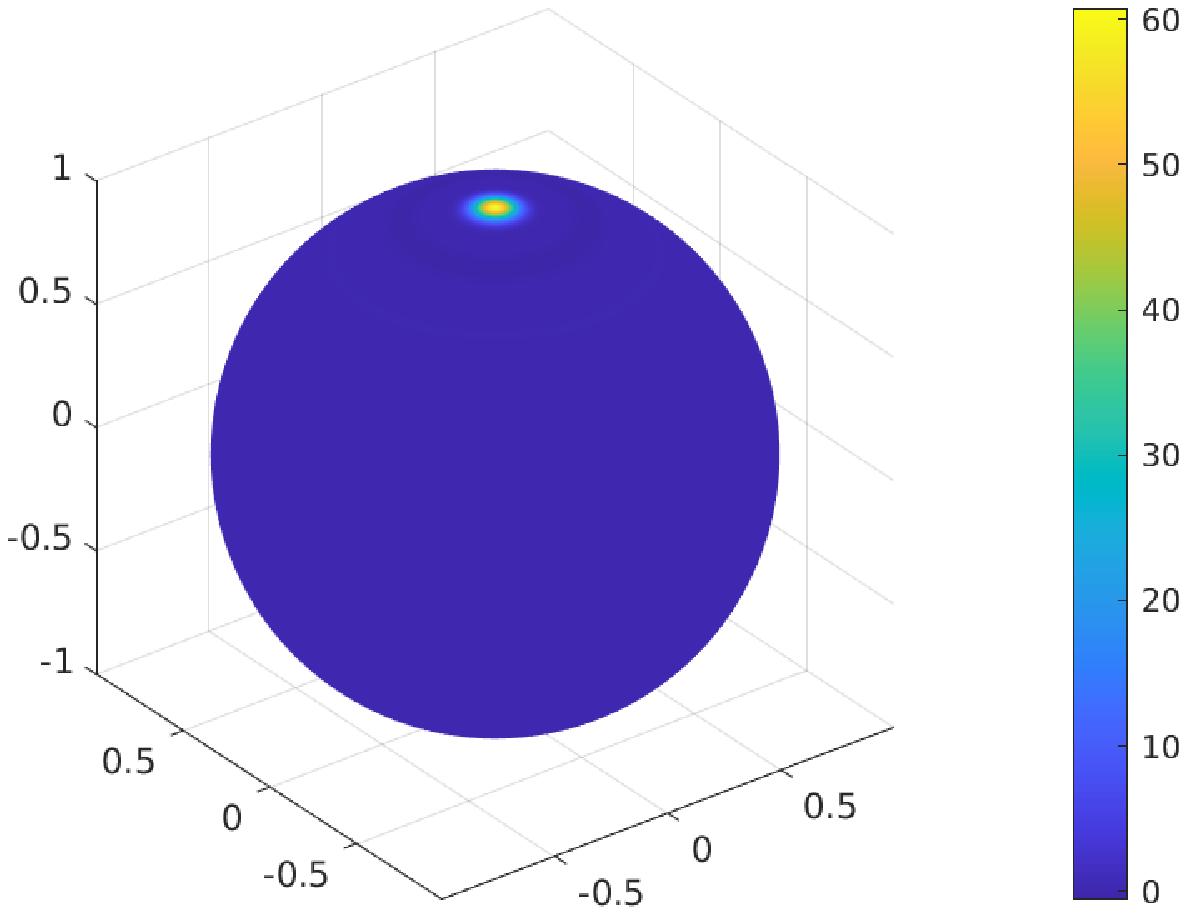}
} \\
\subfloat[$N=3$, $K=7$]{%
  \includegraphics[width=.33\textwidth]{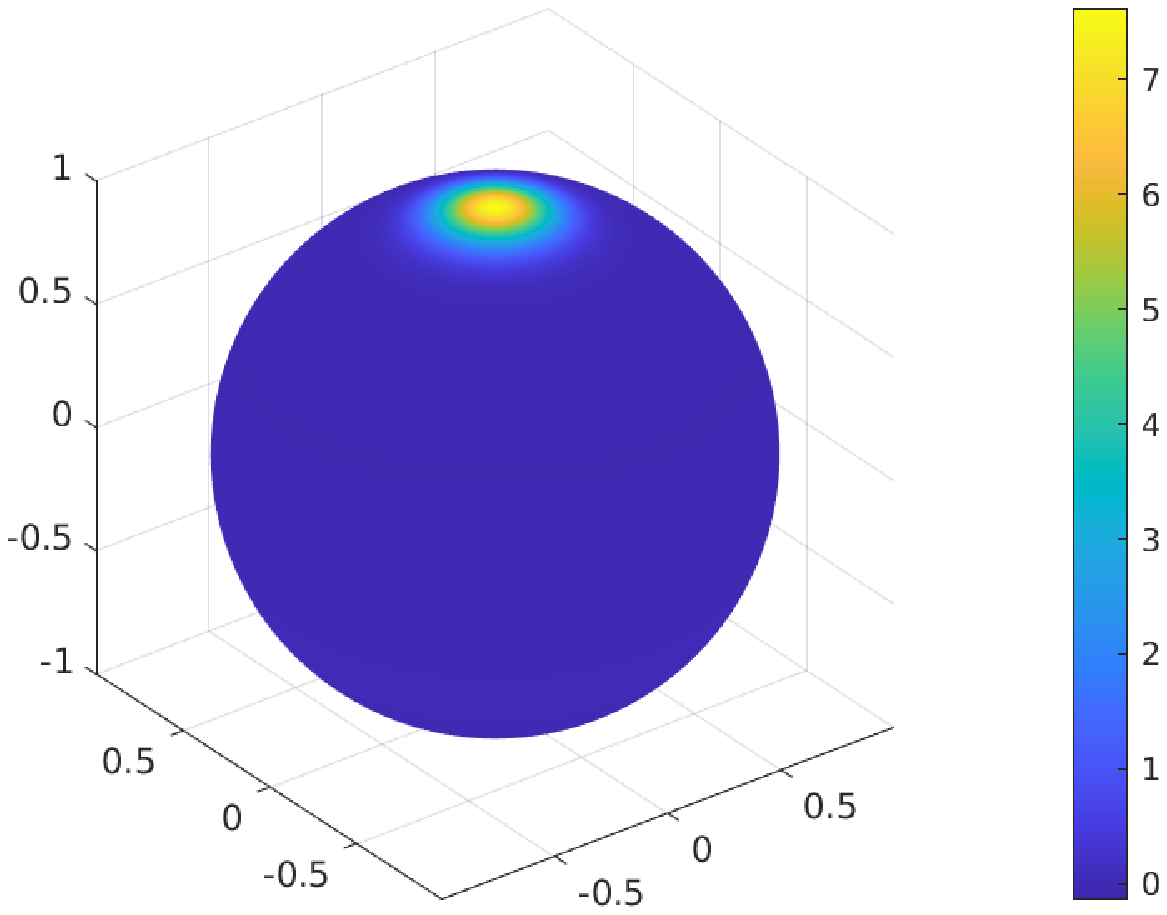}
}
\subfloat[$N=7$, $K=7$]{%
  \includegraphics[width=.33\textwidth]{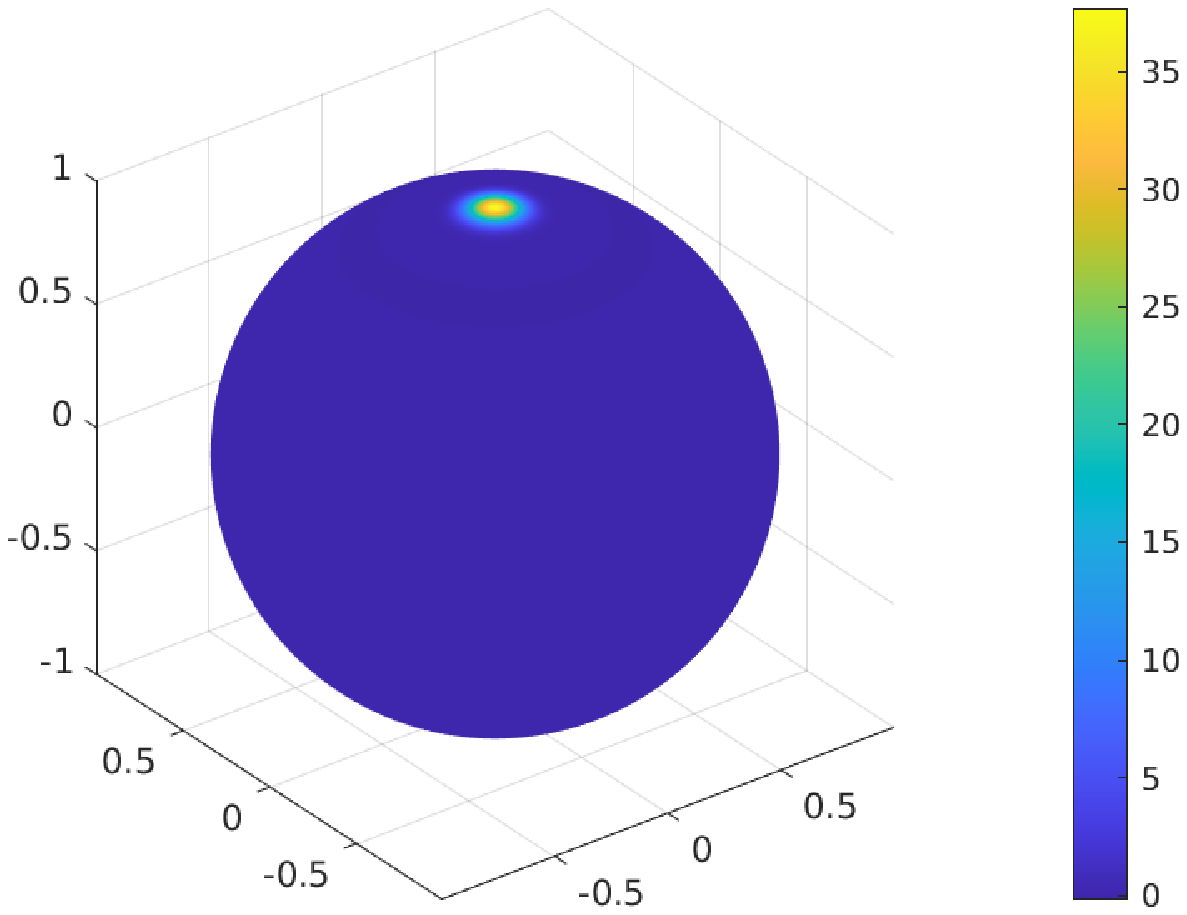}
}
\subfloat[$N=11$, $K=7$]{%
  \includegraphics[width=.33\textwidth]{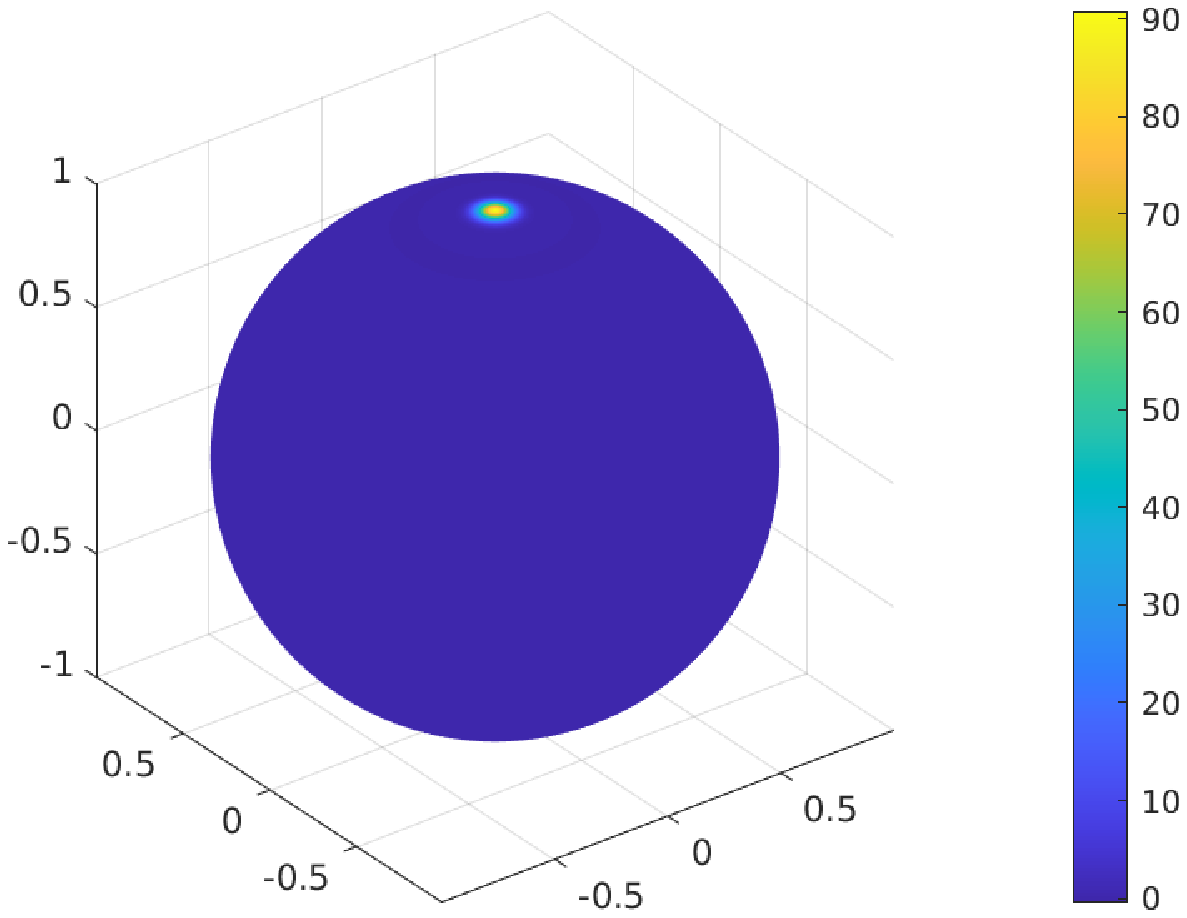}
}
\caption{Approximation of a single Dirac delta function using the $\beta_{N,K}$ model}
\label{fig:single_Dirac}
\end{figure}

The approximations of the Dirac delta function for $N=3,7,11$ and $K=1,3,5,7$ are plotted in Figure \ref{fig:single_Dirac}. Note that when $K=1$, the $\beta_{N,K}$ model is identical to the $P_N$ model. In general, when $N$ or $K$ increases, the intensity gets more concentrated on the north pole of the sphere, indicating better approximation to the Dirac delta function. For $K=1$, a significant negative part of the intensity function can be observed in the plots, and one can clearly see the oscillations of the approximate functions. This is improved remarkably when $K$ increases to $3$. Higher peak value can be attained by further increasing $K$, which implies possible convergence to the Dirac delta function.

To quantify the approximation error, we notice that the Dirac delta function is a member of $H^{\alpha}(\mathbb{S}^2)$ for all $\alpha < -1$, where the $H^{\alpha}$-norm is defined as (see \cite{Barcelo2021Fourier})
\begin{displaymath}
  \|I\|_{H^{\alpha}(\mathbb{S}^2)} = \left(\sum_{l=0}^{+\infty} \sum_{m=-l}^l [1 + l(l+1)]^{\alpha} |\rho_{lm}|^2 \right)^{1/2}.
\end{displaymath}
Here we consider the $H^{-2}$-error between the Dirac delta function and our approximation. According to the general convergence theory for spectral methods, the $H^{-2}$ convergence rate for the approximation of a function in the $H^{\alpha}$ space should be $O(N^{-2-\alpha})$. In this test case, we expect a convergence order that is close to $1$. To verify this prediction, we plot the numerical errors for $K=1,3,5$ in Figure \ref{fig:single_Dirac_error_N}. Due to the numerical difficulty in computing the $H^{-2}$ error, the results are given only up to $N = 13$. All the three lines show a consistent convergence order close to $1$, and the prefactor is smaller for larger values of $K$. This validates our observation from Figure \ref{fig:single_Dirac}. Note that the $\beta_{N,K}$ model has $(N+1)^2$ moments. Hence, the convergence order is $1/2$ with respect to the number of degrees of freedom.
\begin{figure}[!ht]
\centering
\includegraphics[width=.6\textwidth]{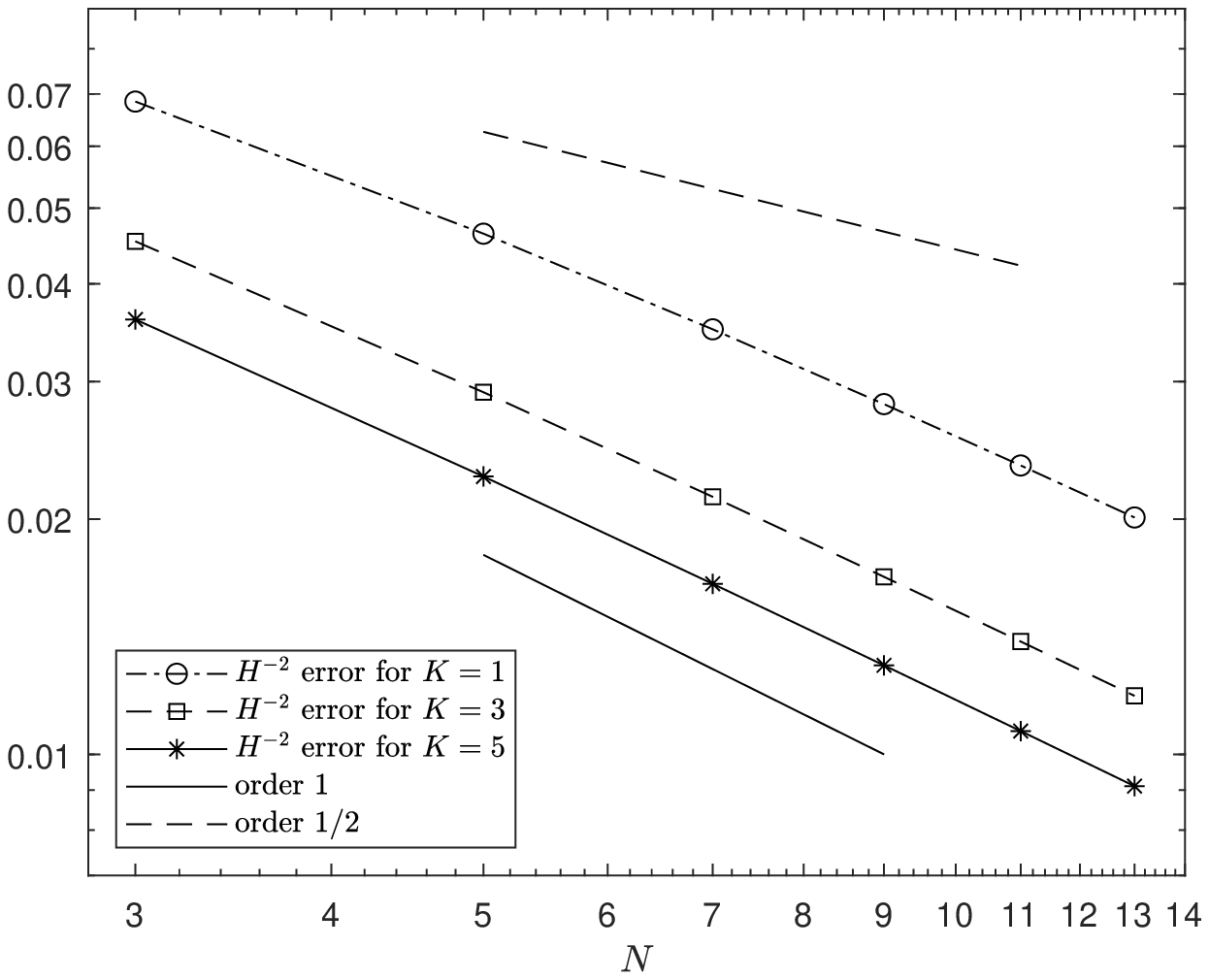}
\caption{Convergence order for the approximation of a single Dirac delta function}
\label{fig:single_Dirac_error_N}
\end{figure}

Recall that the $M_N$ model can represent a single Dirac delta function exactly for all $N \geq 1$, which means that the approximation error coincides with the difference between the $\beta_{N,K}$ model and the $M_N$ model. This fact allows us to use this test case to check the convergence rate of our $\beta_{N,K}$ model towards the $M_N$ model. Now we fix $N$ and compute the $H^{-2}$ error for various values of $K$. The results are given in Figure \ref{fig:single_Dirac_error_K}. It shows that for $N = 1,3,5$, the convergence order with respect $K$ is close to $1/2$. Note that this convergence rate only applies to this particular case, and it may change with the smoothness of the function and the norm used to measure the error.
\begin{figure}[!ht]
\centering
\includegraphics[width=.6\textwidth]{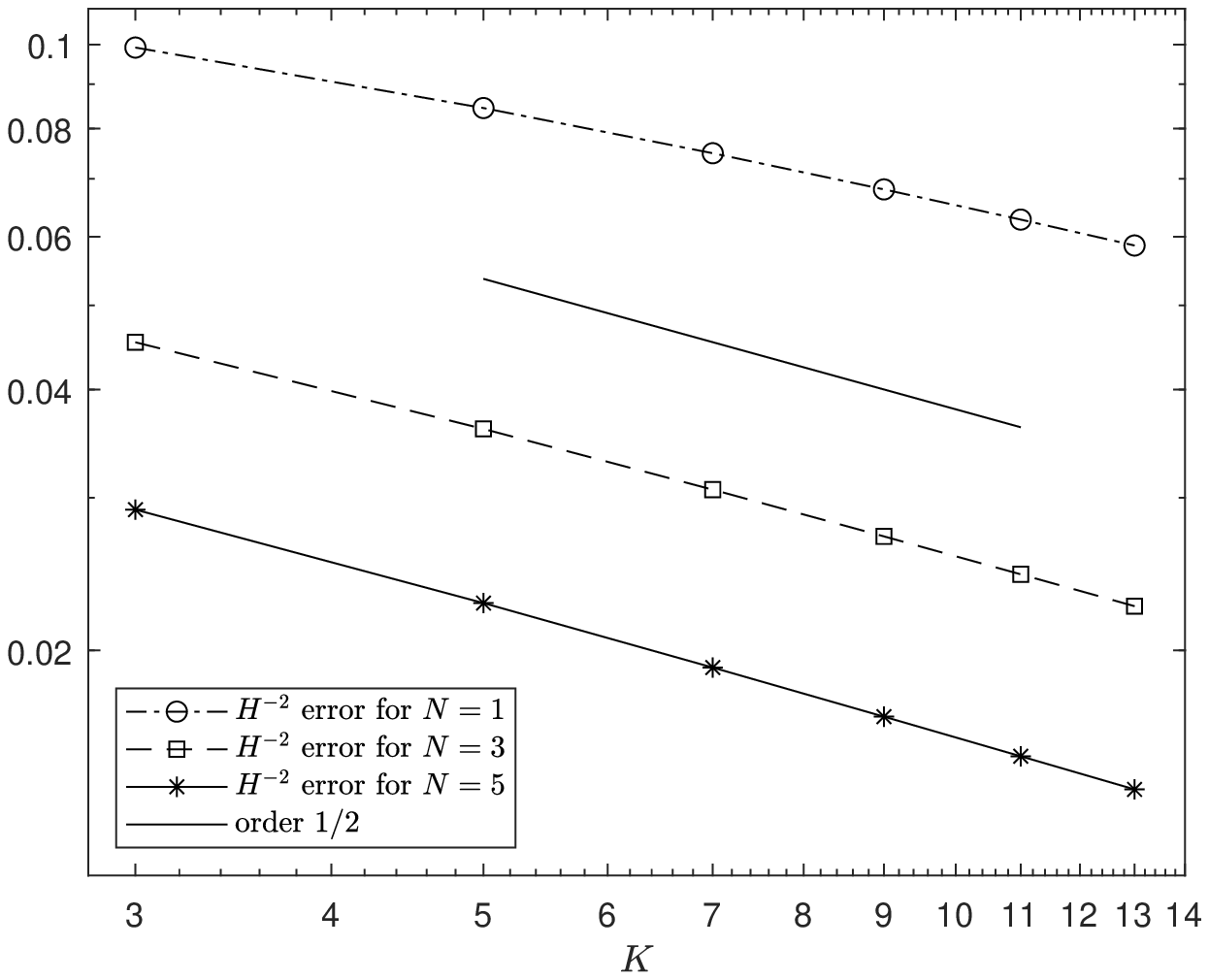}
\caption{Convergence order for the approximation of a single Dirac delta function}
\label{fig:single_Dirac_error_K}
\end{figure}

\subsubsection{Approximation of multiple beams} \label{sec:Dirac}
We now consider intensity functions with multiple beams:
\begin{displaymath}
I(\Omega) = \sum_{k=1}^n \delta(\Omega - \Omega_k).
\end{displaymath}
The numerical scheme to obtain the $\beta_{N,K}$ approximation is the same as the case of a single beam. For conciseness, we will only demonstrate results for $n = 2$ and $n = 3$.

For the two-beam case ($n = 2$), we let
\begin{displaymath}
\Omega_1 = (0,0,1)^{\top}, \qquad \Omega_2 = (0,-1,0)^{\top}.
\end{displaymath}
The numerical results for $N = 1,2,3$ and $K = 5,9$ are plotted in Figure \ref{fig:double_Dirac}. Unlike the case of a single beam, now the $M_1$ model is unable to represent $I(\Omega)$ exactly, so that the $\beta_{1,K}$ model gives a poor approximation of the intensity function. Figures \ref{fig:N1K5} and \ref{fig:N1K9} show that the approximation provides radiations spreading around the angle in the middle of $\Omega_1$ and $\Omega_2$. The $M_2$ model can represent the intensity function exactly, so that the results of the $\beta_{2,K}$ model have two bright spots around $\Omega_1$ and $\Omega_2$. However, even for $K = 9$, the $\beta_{2,K}$ result still shows a significant amount of radiation in the directions connecting $\Omega_1$ and $\Omega_2$. This can be well suppressed by increasing $N$ to $3$.
\begin{figure}[!ht]
\centering
\subfloat[$N=1$, $K=5$]{%
  \label{fig:N1K5}
  \includegraphics[width=.33\textwidth]{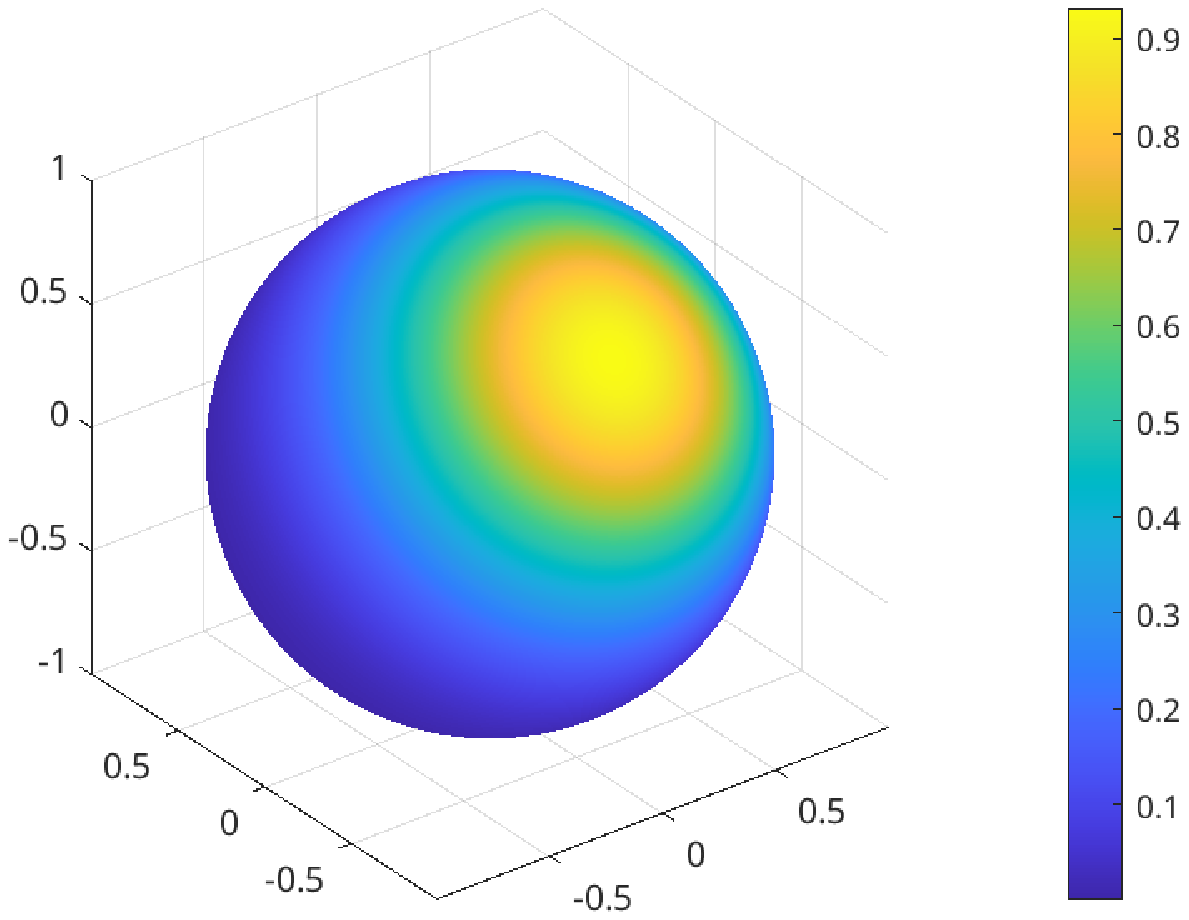}
}
\subfloat[$N=2$, $K=5$]{%
  \includegraphics[width=.33\textwidth]{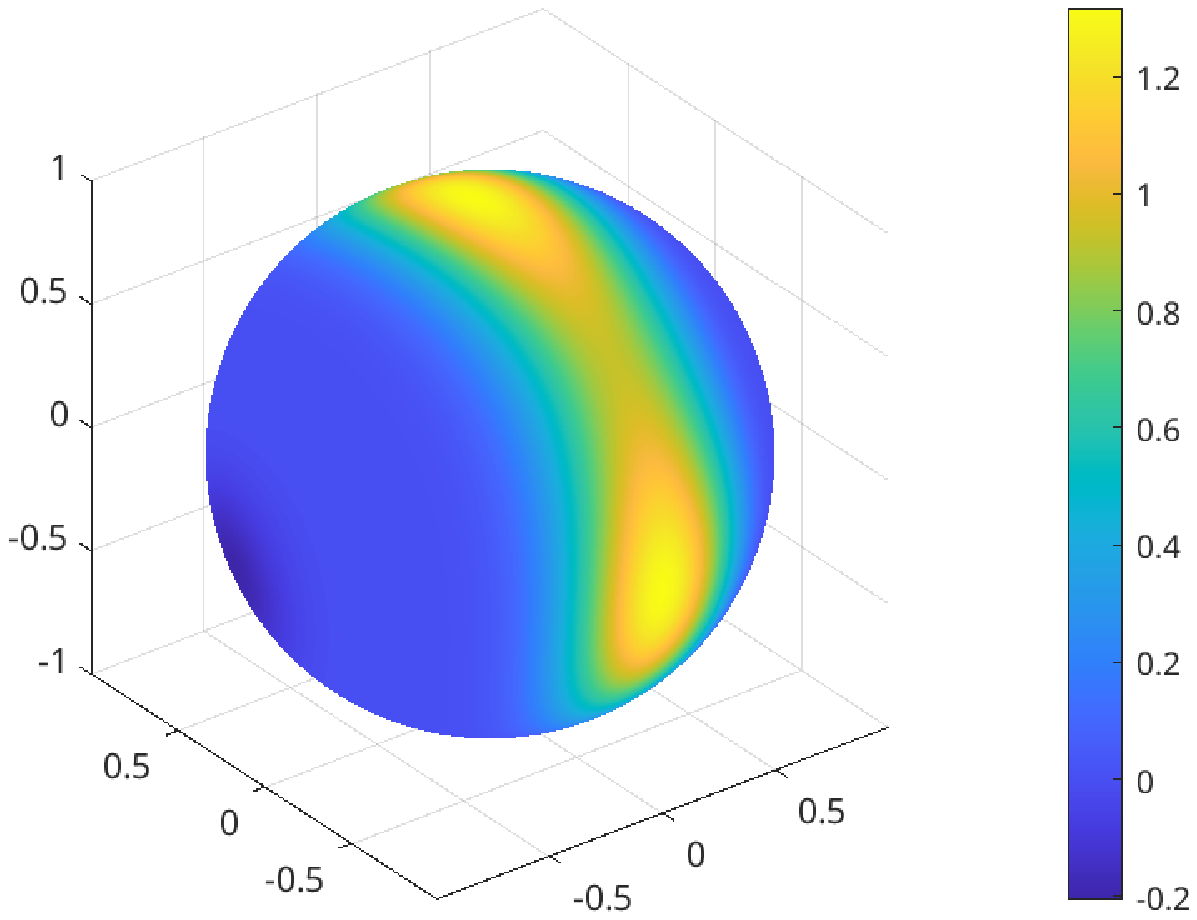}
}
\subfloat[$N=3$, $K=5$]{%
  \includegraphics[width=.33\textwidth]{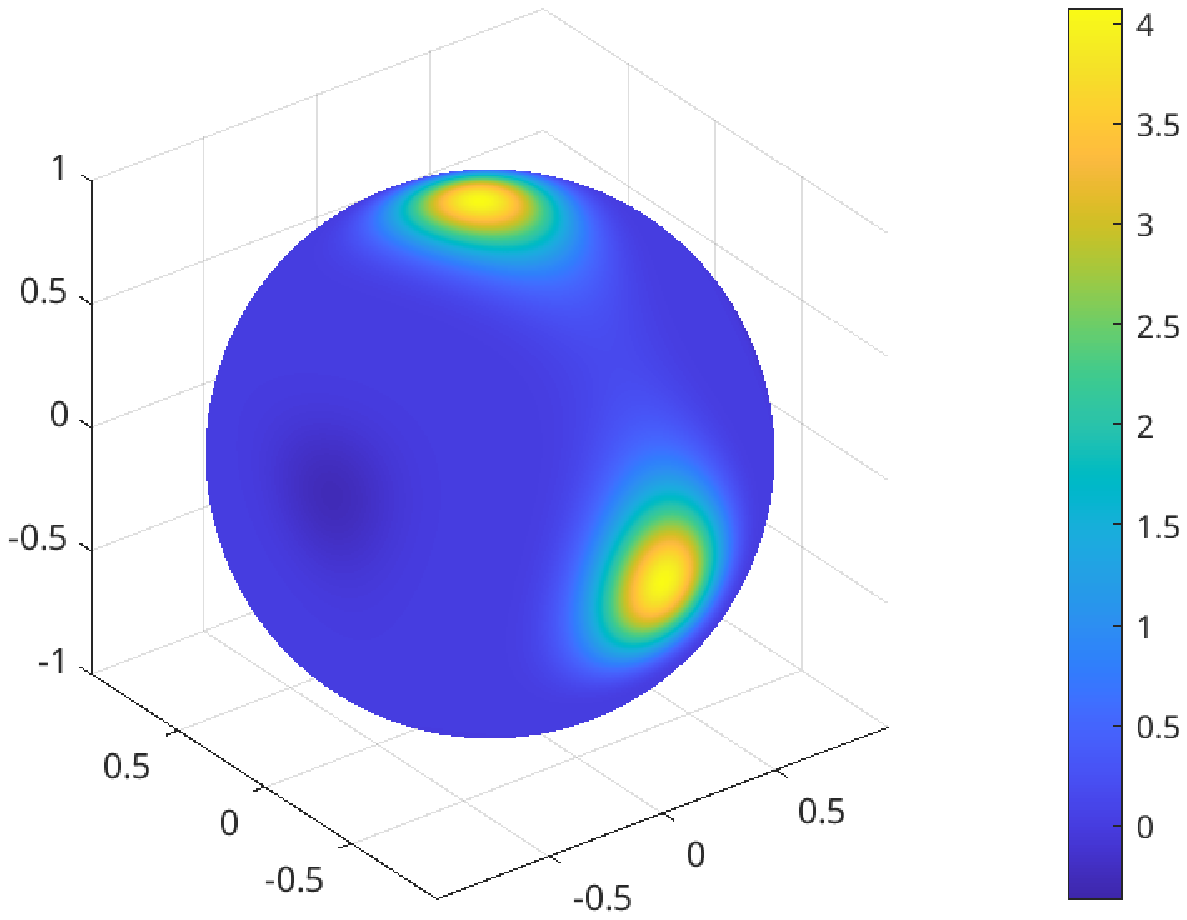}
} \\
\subfloat[$N=1$, $K=9$]{%
  \label{fig:N1K9}
  \includegraphics[width=.33\textwidth]{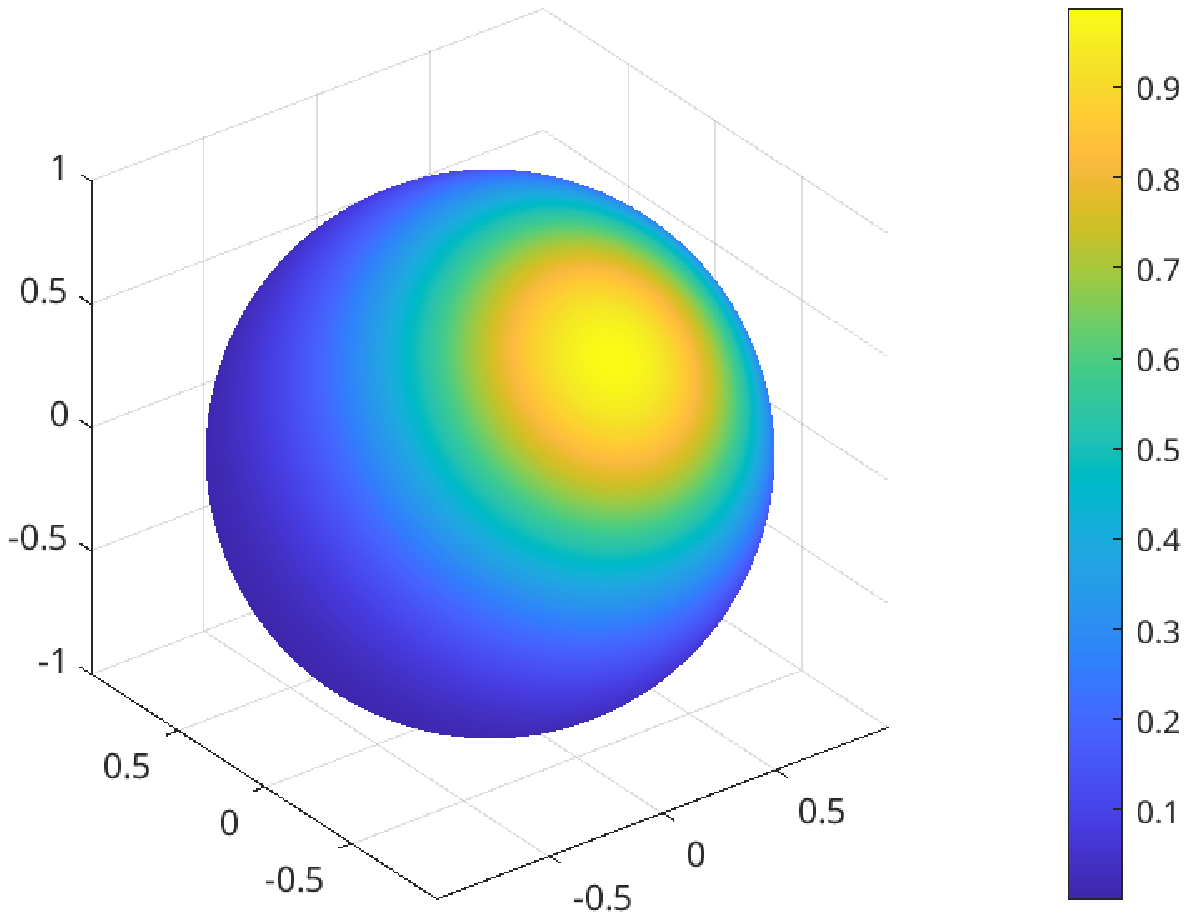}
}
\subfloat[$N=2$, $K=9$]{%
  \includegraphics[width=.33\textwidth]{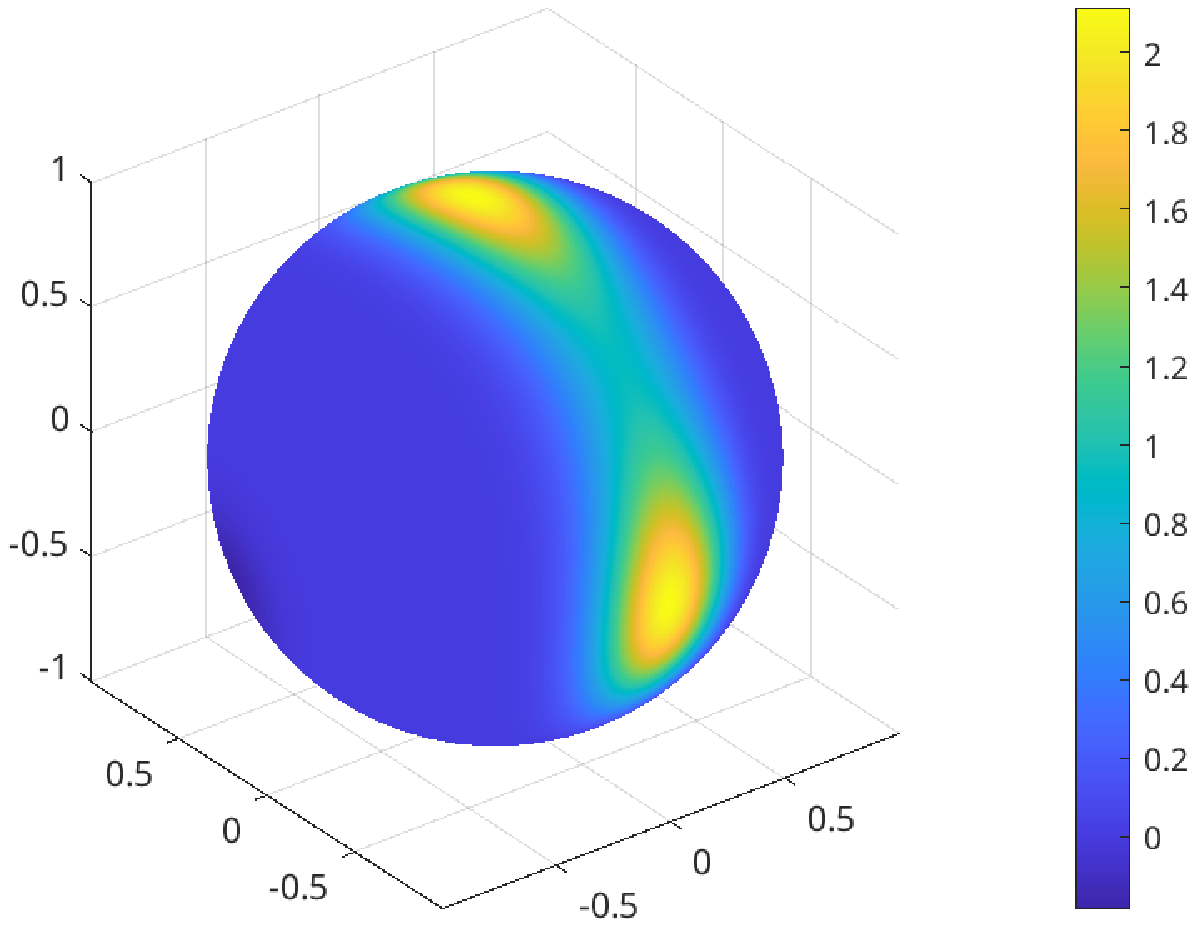}
}
\subfloat[$N=3$, $K=9$]{%
  \includegraphics[width=.33\textwidth]{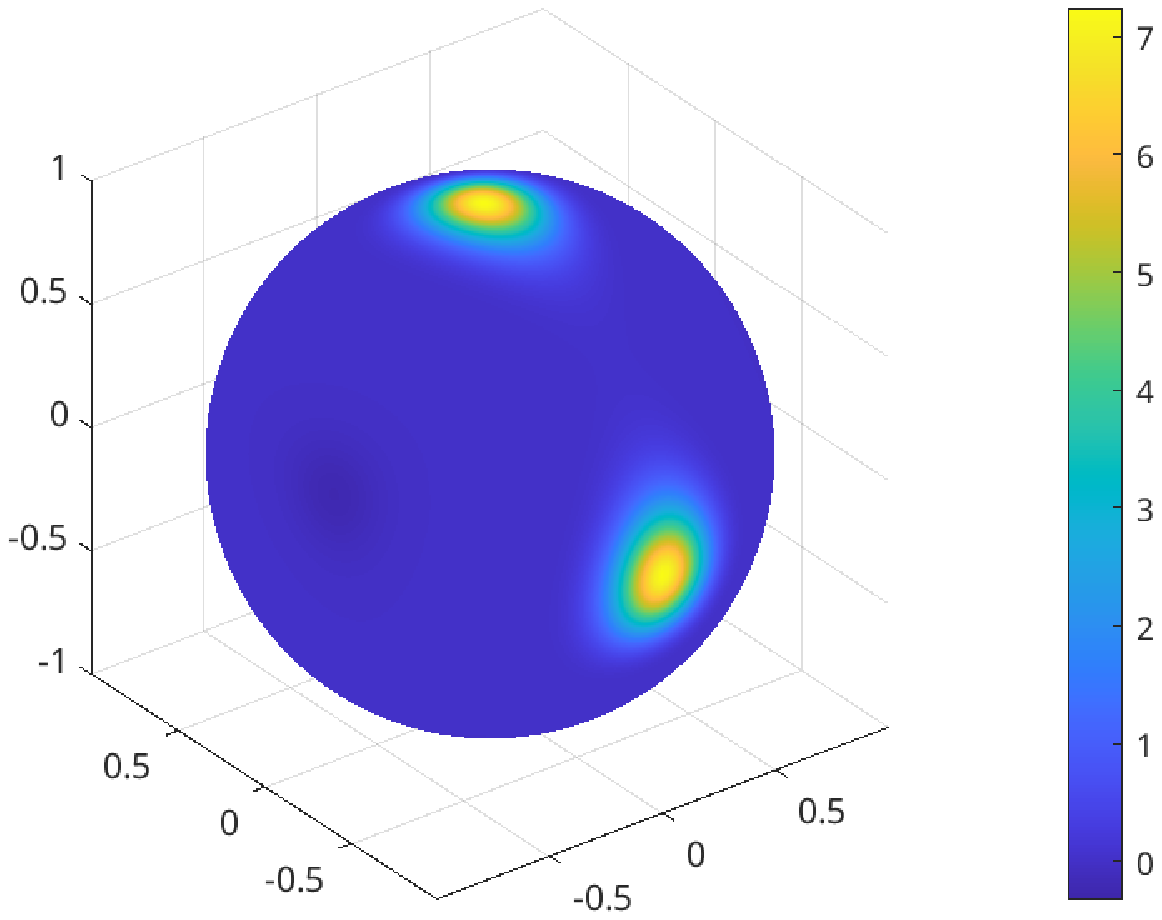}
}
\caption{Approximation to the sum of two Dirac delta functions using the $\beta_{N,K}$ model}
\label{fig:double_Dirac}
\end{figure}

Further increasing $N$ will lead to better approximations. Instead of showing the function plots, we provide the decay of the $H^{-2}$ numerical error. Again, we observe the first-order convergence for all $K = 1,3$ and $5$.
\begin{figure}[!ht]
\centering
\includegraphics[width=.6\textwidth]{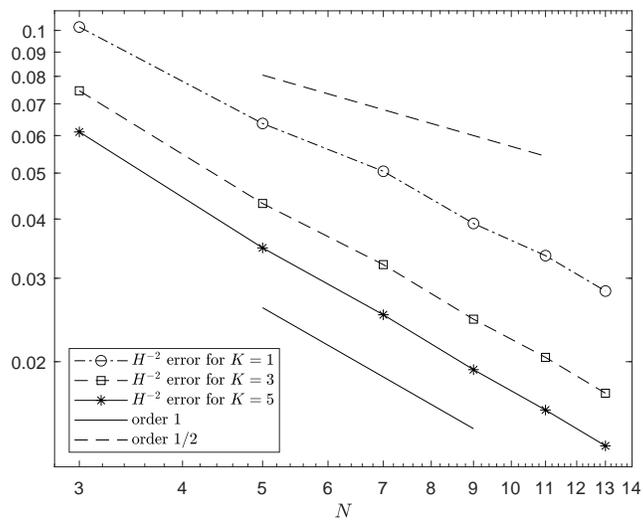}
\caption{Convergence order for the approximation to the sum of two Dirac delta functions}
\label{fig:double_Dirac_error_K}
\end{figure}

The tests for three beams ($n = 3$) show a similar behavior. Here we choose
\begin{displaymath}
\Omega_1 = (0,0,1)^{\top}, \qquad
\Omega_2 = (0,-1,0)^{\top}, \qquad
\Omega_3 = (-1,0,0)^{\top}.
\end{displaymath}
The results are shown in Figure \ref{fig:triple_Dirac}. However, since the $M_2$ model is incapable of representing such a function, the $\beta_{2,K}$ models also fail to produce qualitatively correct results. A sensible approximation requires at least $N = 3$ in the $\beta_{N,K}$ model.
\begin{figure}[!ht]
\centering
\subfloat[$N=1$, $K=5$]{%
  \includegraphics[width=.33\textwidth]{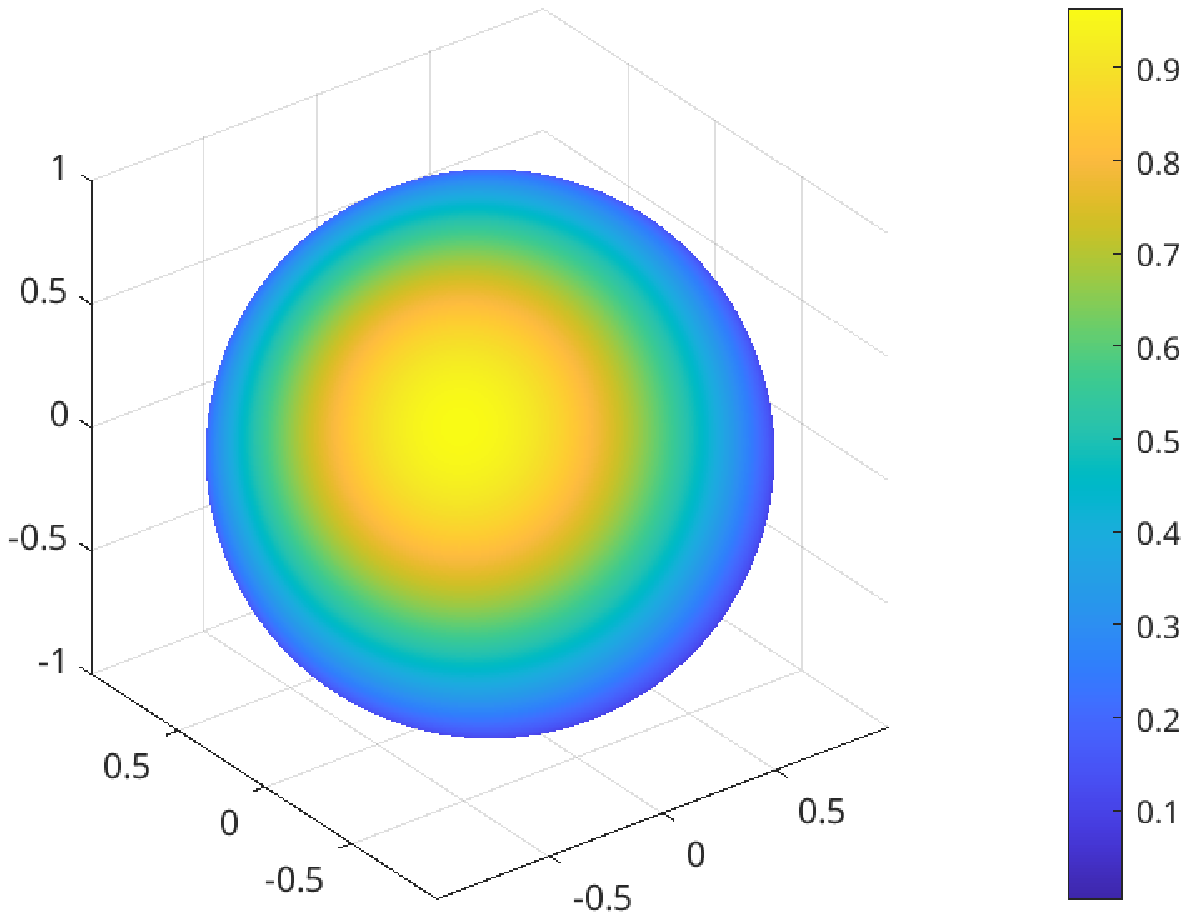}
}
\subfloat[$N=2$, $K=5$]{%
  \includegraphics[width=.33\textwidth]{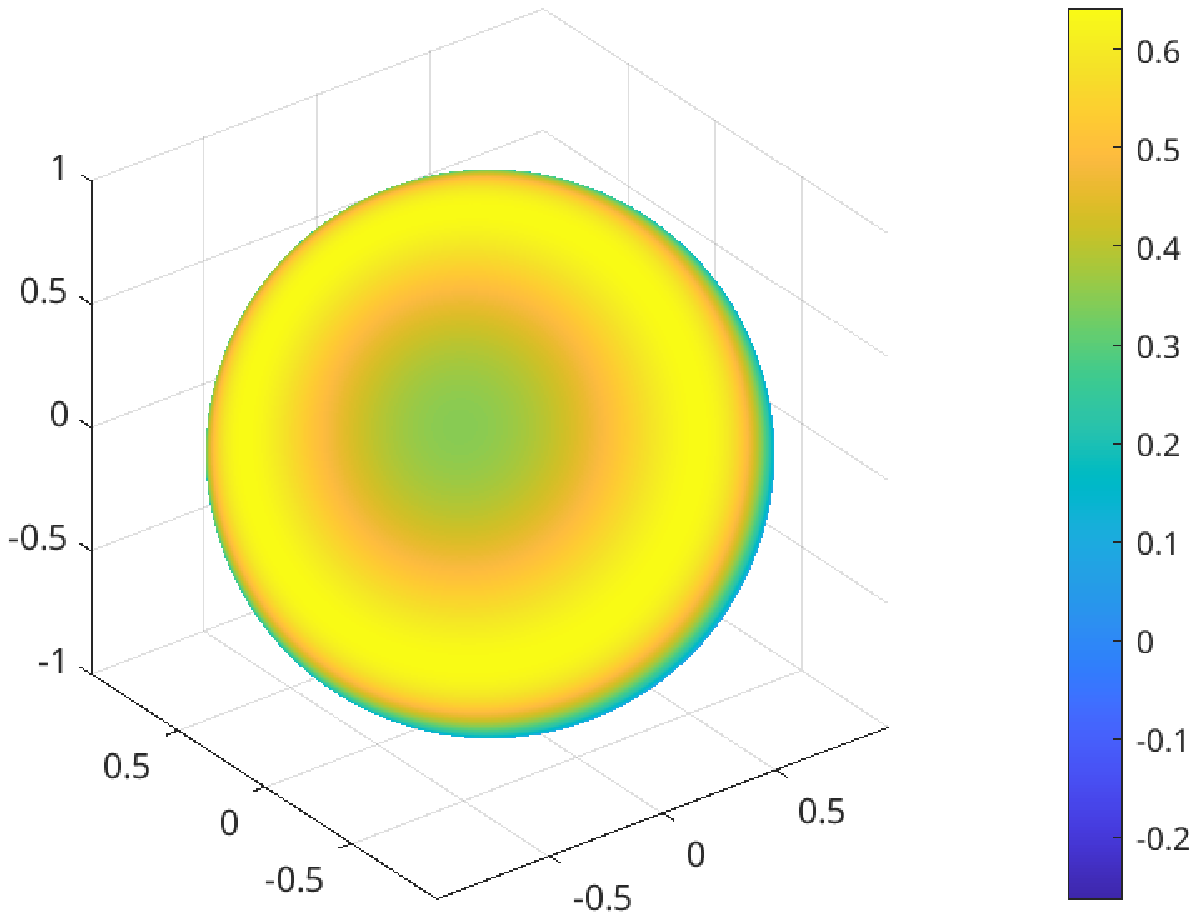}
}
\subfloat[$N=3$, $K=5$]{%
  \includegraphics[width=.33\textwidth]{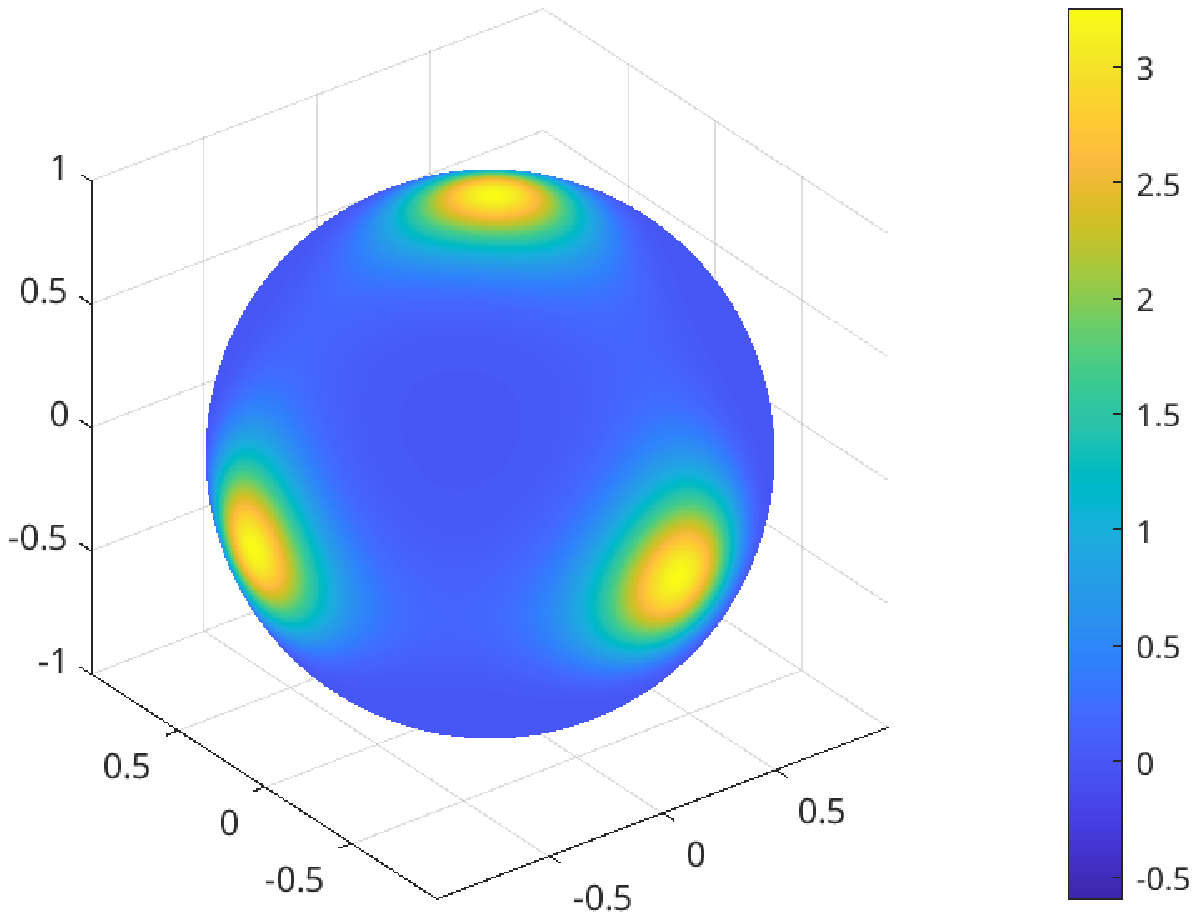}
} \\
\subfloat[$N=1$, $K=9$]{%
  \includegraphics[width=.33\textwidth]{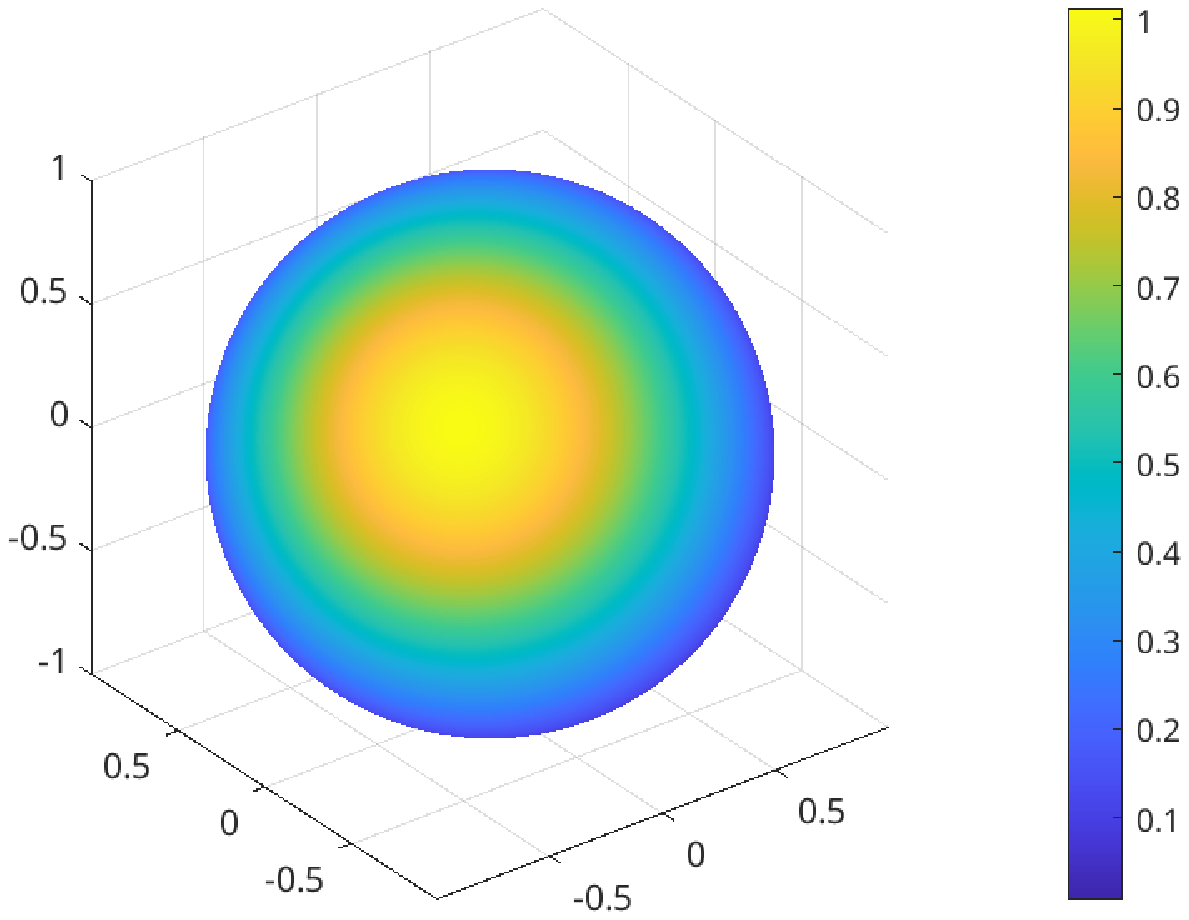}
}
\subfloat[$N=2$, $K=9$]{%
  \includegraphics[width=.33\textwidth]{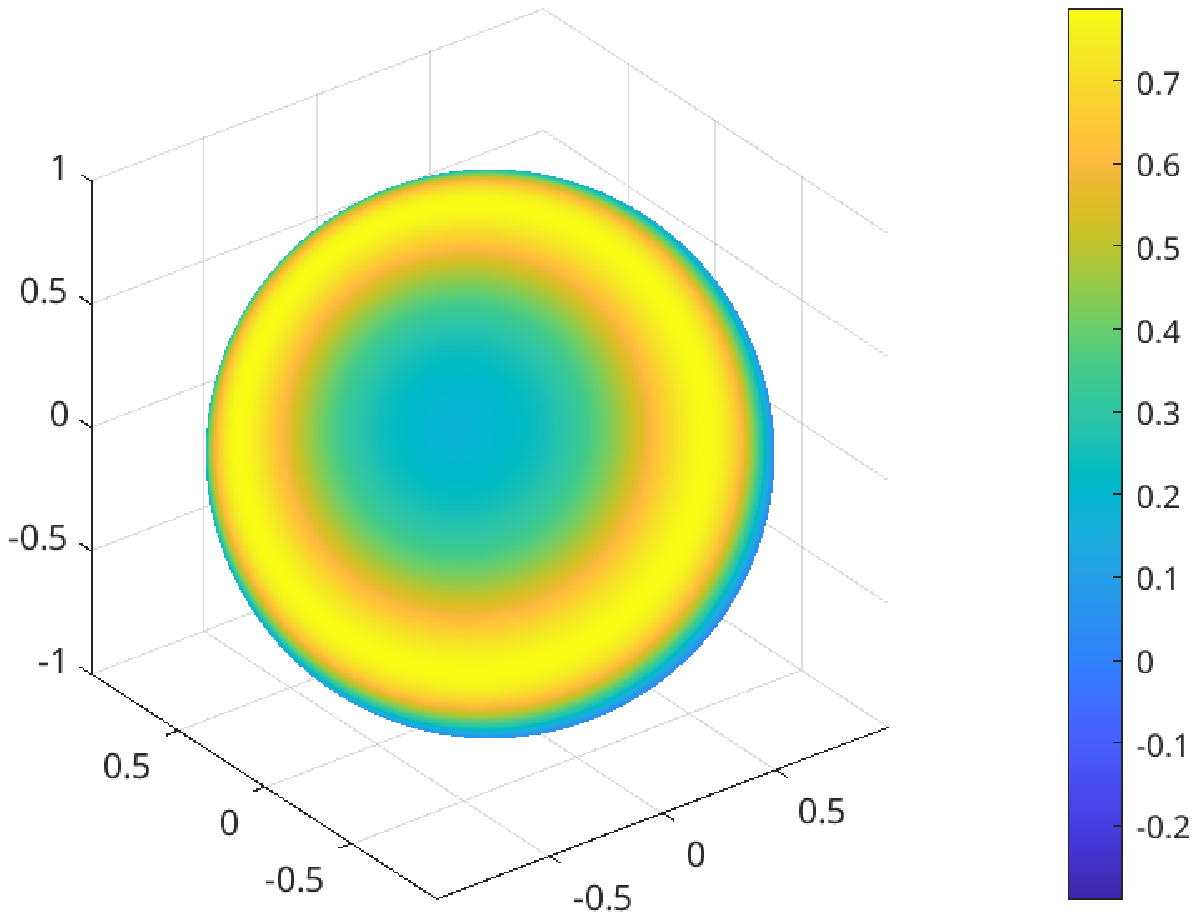}
}
\subfloat[$N=3$, $K=9$]{%
  \includegraphics[width=.33\textwidth]{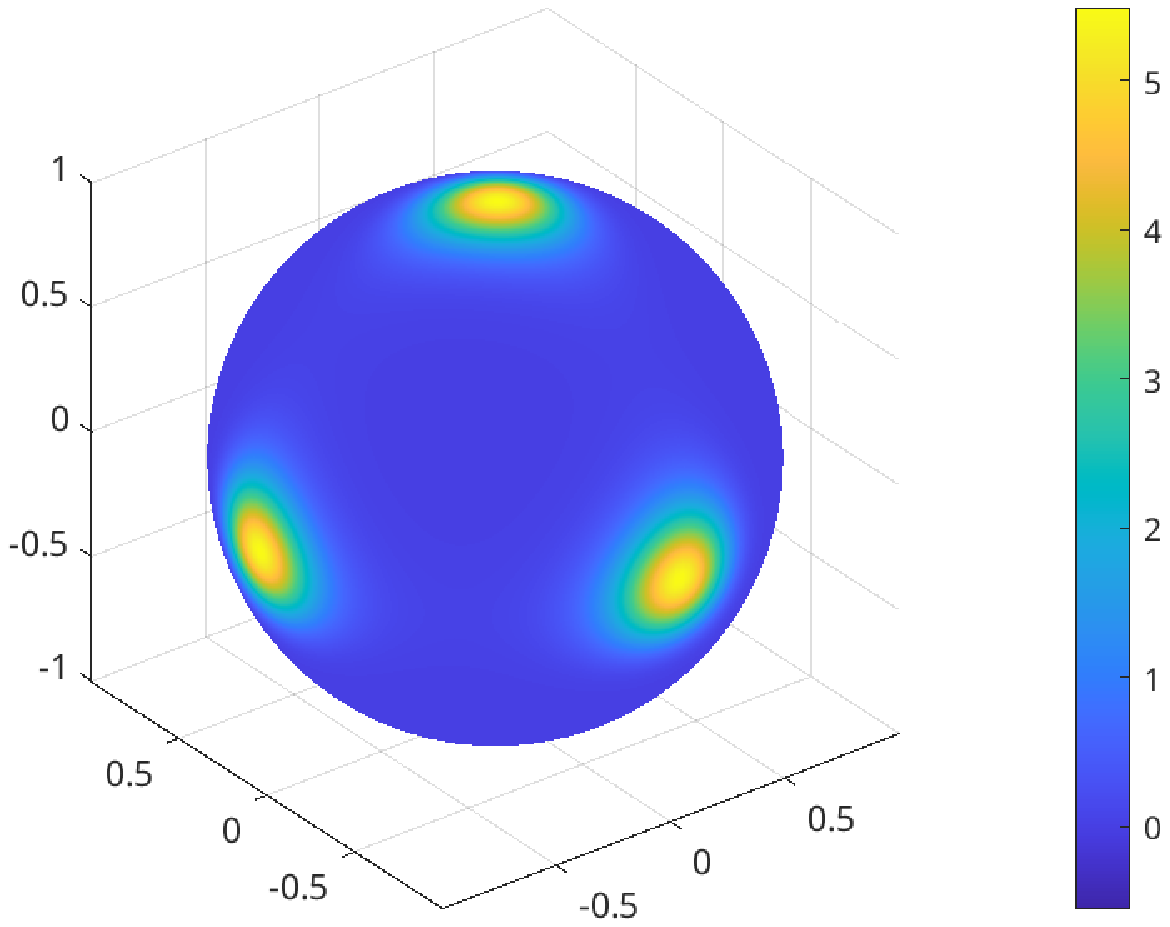}
}
\caption{Approximation to the sum of three Dirac delta functions using the $\beta_{N,K}$ model}
\label{fig:triple_Dirac}
\end{figure}

\subsubsection{Approximation of smooth functions}
In this section, we consider the approximation of the following smooth intensity function:
\begin{equation} \label{eq:smooth_test}
I(\Omega) = \sum_{k=1}^6 \exp(-5\|\Omega - \Omega_k\|^2),
\end{equation}
where
\begin{gather*}
  \Omega_1 = (1,0,0)^{\top}, \qquad \Omega_2 = (0,1,0)^{\top}, \qquad \Omega_3 = (0,0,1)^{\top}, \\
  \Omega_4 = (-1,0,0)^{\top}, \qquad \Omega_5 = (0,-1,0)^{\top}, \qquad \Omega_6 = (0,0,-1)^{\top}.
\end{gather*}
The original function and the approximations with the $\beta_{9,1}$ and $\beta_{5,3}$ models are given in Figure \ref{fig:smooth}, which shows that the $\beta_{5,3}$ model gives better approximation than the $P_5$ model. However, when $N$ further increases, the $P_N$ model will start overtaking. The results are shown in Figure \ref{fig:smooth_convergence}, where the $P_N$ model shows a considerably faster convergence rate than the other two models, although the $\beta_{N,3}$ model and the $\beta_{N,5}$ model also show spectral accuracy due to the logarithmic scale of the vertical axis. Meanwhile, it can be seen that the error of the $\beta_{N,5}$ model is slightly larger than the $\beta_{N,3}$ model. One possible reason of this phenomenon is that the $\beta_{N,K}$ model essentially approximates the function $[I(\Omega)]^{1/K}$ by a polynomial of degree $N$. When $K > 1$, the function $x^{1/K}$ has large derivatives for $x$ close to zero, so that the function $[I(\Omega)]^{1/K}$ may be difficult to approximate at places where the value of $I(\Omega)$ is small.

\begin{figure}
\centering
\subfloat[Original function]{%
  \includegraphics[width=.33\textwidth]{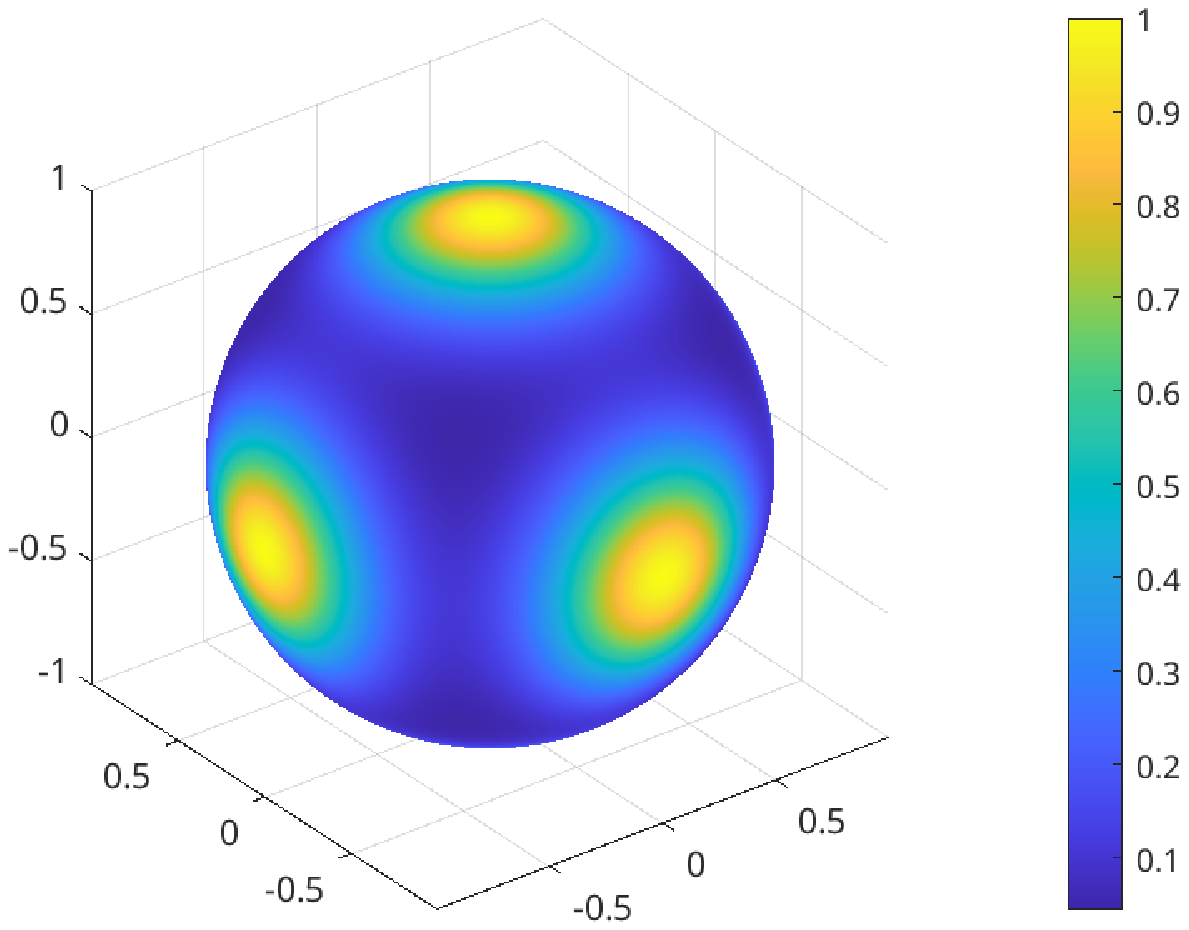}
}
\subfloat[$\beta_{5,1}$ ($P_5$) approximation]{%
  \includegraphics[width=.33\textwidth]{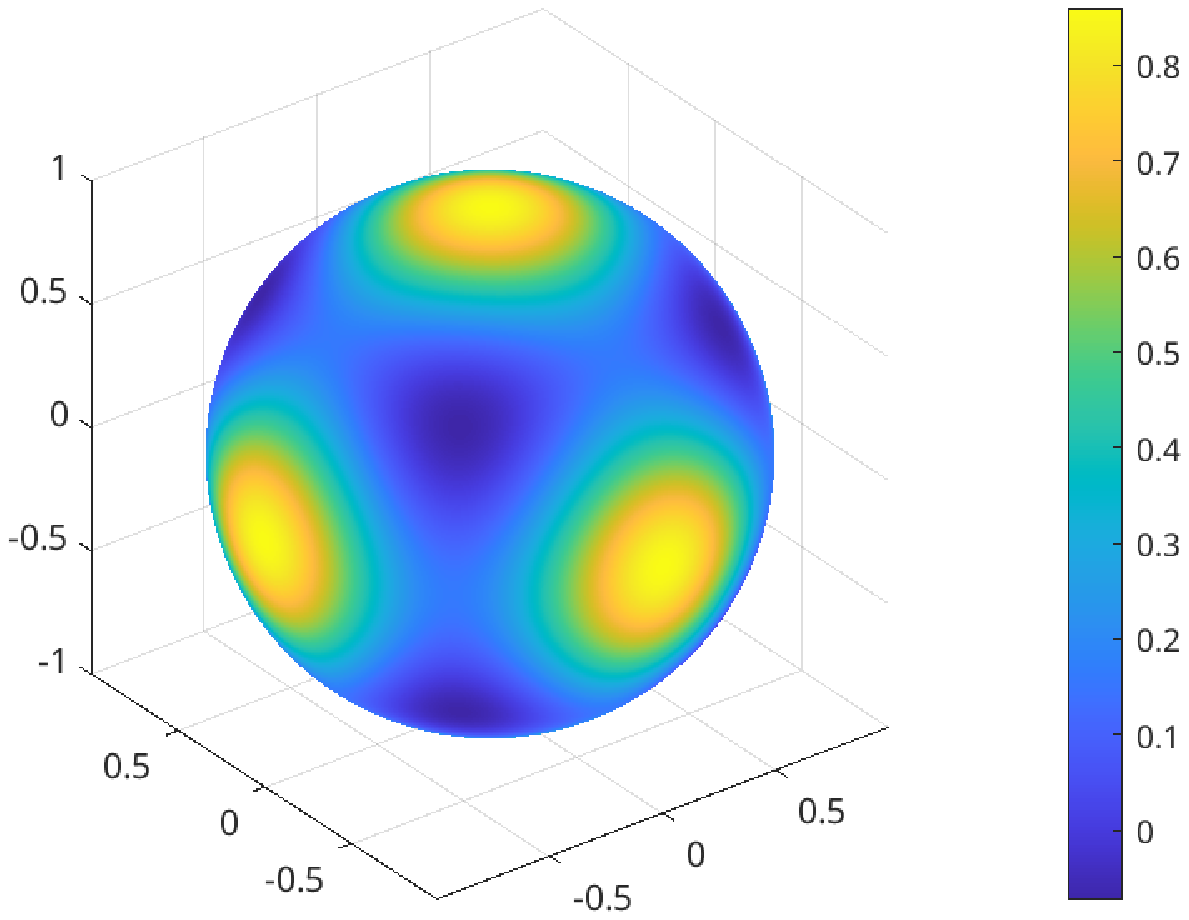}
}
\subfloat[$\beta_{5,3}$ approximation]{%
  \includegraphics[width=.33\textwidth]{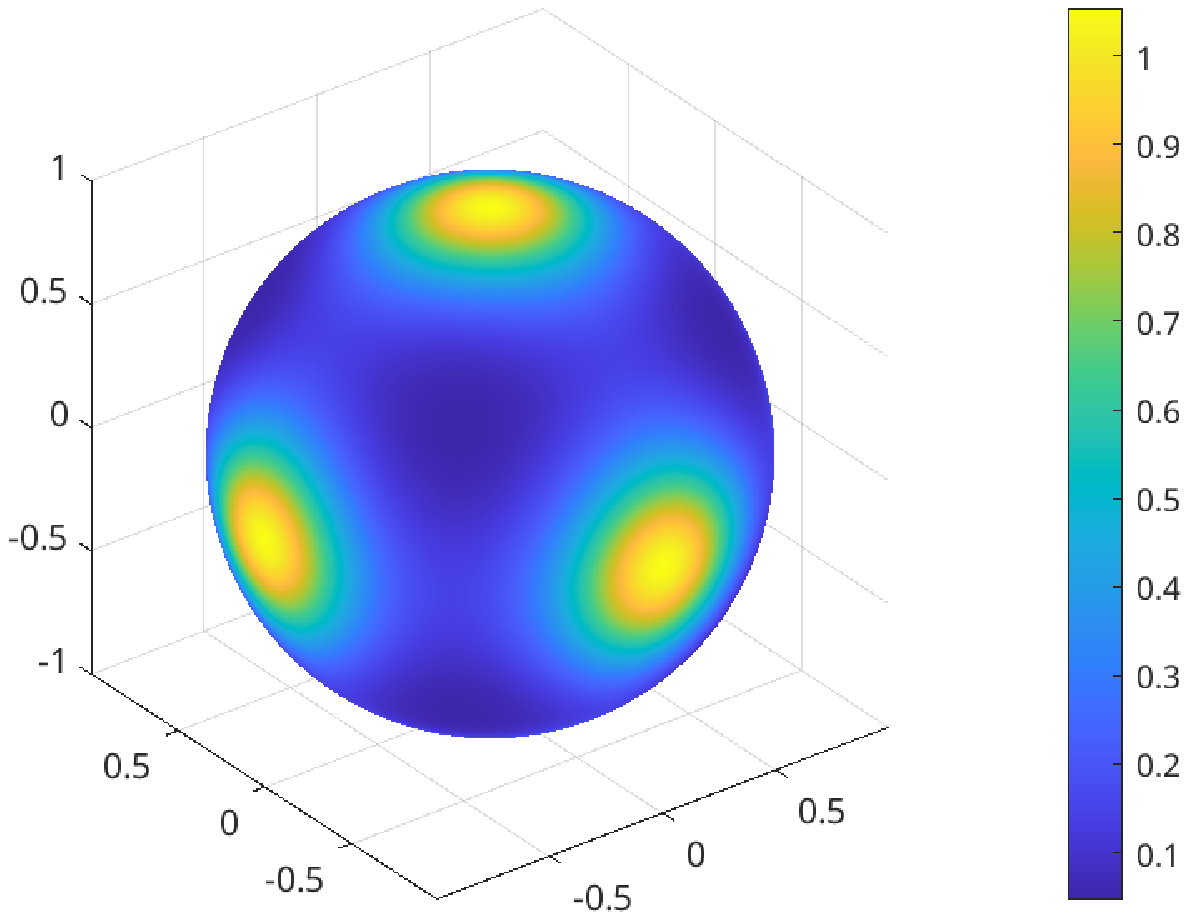}
}  
\caption{Approximation of the intensity function \eqref{eq:smooth_test}}
\label{fig:smooth}
\end{figure}

\begin{figure}[!ht]
\centering
\includegraphics[width=.6\textwidth]{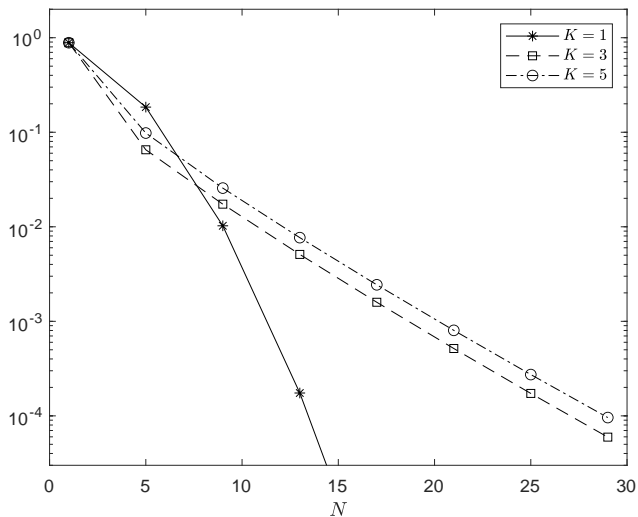}
\caption{$L^2$ error of the approximation to the intensity function \eqref{eq:smooth_test} with the $\beta_{N,K}$ model}
\label{fig:smooth_convergence}
\end{figure}

To verify this conjecture, we consider the approximation of a slightly different function
\begin{equation} \label{eq:smooth_test1}
I(\Omega) = \frac{1}{2} + \sum_{k=1}^6 \exp(-5\|\Omega - \Omega_k\|^2),
\end{equation}
which adds $1/2$ to the intensity function \eqref{eq:smooth_test}, so that the ratio of its maximum value to its minimum value is much smaller. The convergence result is plotted in Figure \ref{fig:smooth_convergence1}. It shows that the three models still have similar performance when $N = 13$, and the $\beta_{N,3}$ and $\beta_{N,5}$ models can reach a much smaller $L^2$ error compared with the previous example. Nevertheless, the $P_N$ model still shows much better results from $N = 17$ due to the large high-order derivatives of the function $x^{1/K}$.
\begin{figure}[!ht]
\centering
\includegraphics[width=.6\textwidth]{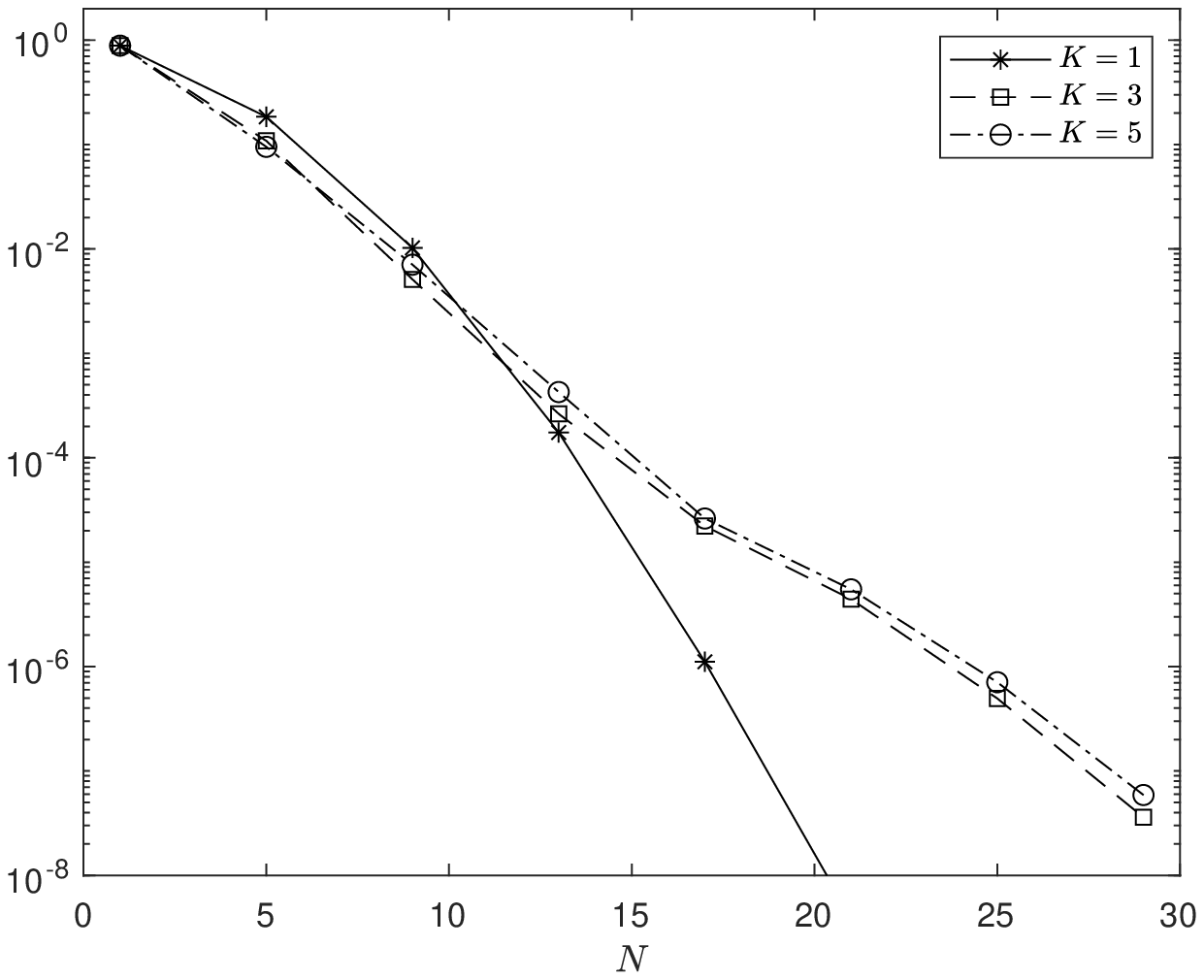}
\caption{$L^2$ error of the approximation to the intensity function \eqref{eq:smooth_test1} with the $\beta_{N,K}$ model}
\label{fig:smooth_convergence1}
\end{figure}

\subsection{Line source problem}
We now test the performance of the $\beta_{N,K}$ model on the benchmark line source problem, which is often used to demonstrate the ray effect in the discrete ordinates method  and test the capability of moment methods \cite{Laiu2016, Camminady2019}. Consider the two-dimensional spatial domain with the initial condition
\begin{equation} \label{eq:init}
I(x,y,\Omega,0) = \frac{1}{8\pi^2 \omega^2} \exp \left( -\frac{x^2 + y^2}{2\omega^2} \right),
\end{equation}
which simulates a point source at the origin. The problem is highly challenging due to the beam-like solutions in all directions. A variety of moment methods have been tested on this test problem in \cite{Garrett2013comparison}. Following \cite{Garrett2013comparison}, we choose $\omega = 0.03$ and set the scattering coefficient to be $\sigma = 1$. The Green's function of this PDE has been obtained in \cite{Ganapol1999homogeneous}, so that the solution of the initial value problem can be found by convoluting the Green's function and the initial data. The exact solution of the density (the integral of $I(x,y,\Omega,t)$ with respect to $\Omega$) for the initial condition \eqref{eq:init} at $t = 1$ is plotted in Figure \ref{fig:LineSource_exact2d}. Due to the radial symmetry of the initial condition, the solution for any $t$ also depends only on the radius $r = \sqrt{x^2 + y^2}$. The density as a function of $r$ is plotted in Figure \ref{fig:LineSource_exact1d}.
\begin{figure}[!ht]
\centering
\subfloat[2D plot]{%
  \label{fig:LineSource_exact2d}
  \includegraphics[width=.4\textwidth]{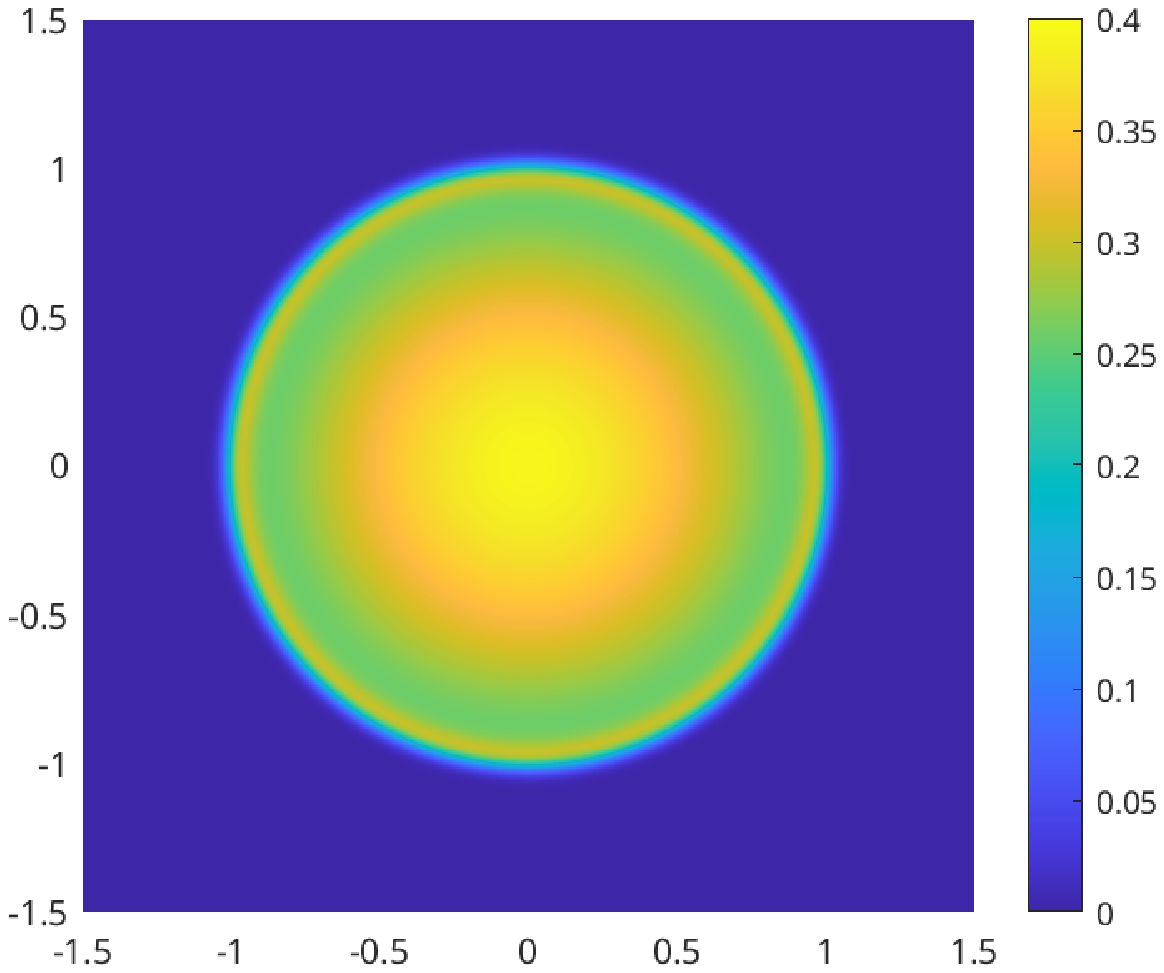}
}
\subfloat[Radial plot]{%
  \label{fig:LineSource_exact1d}
  \includegraphics[width=.4\textwidth]{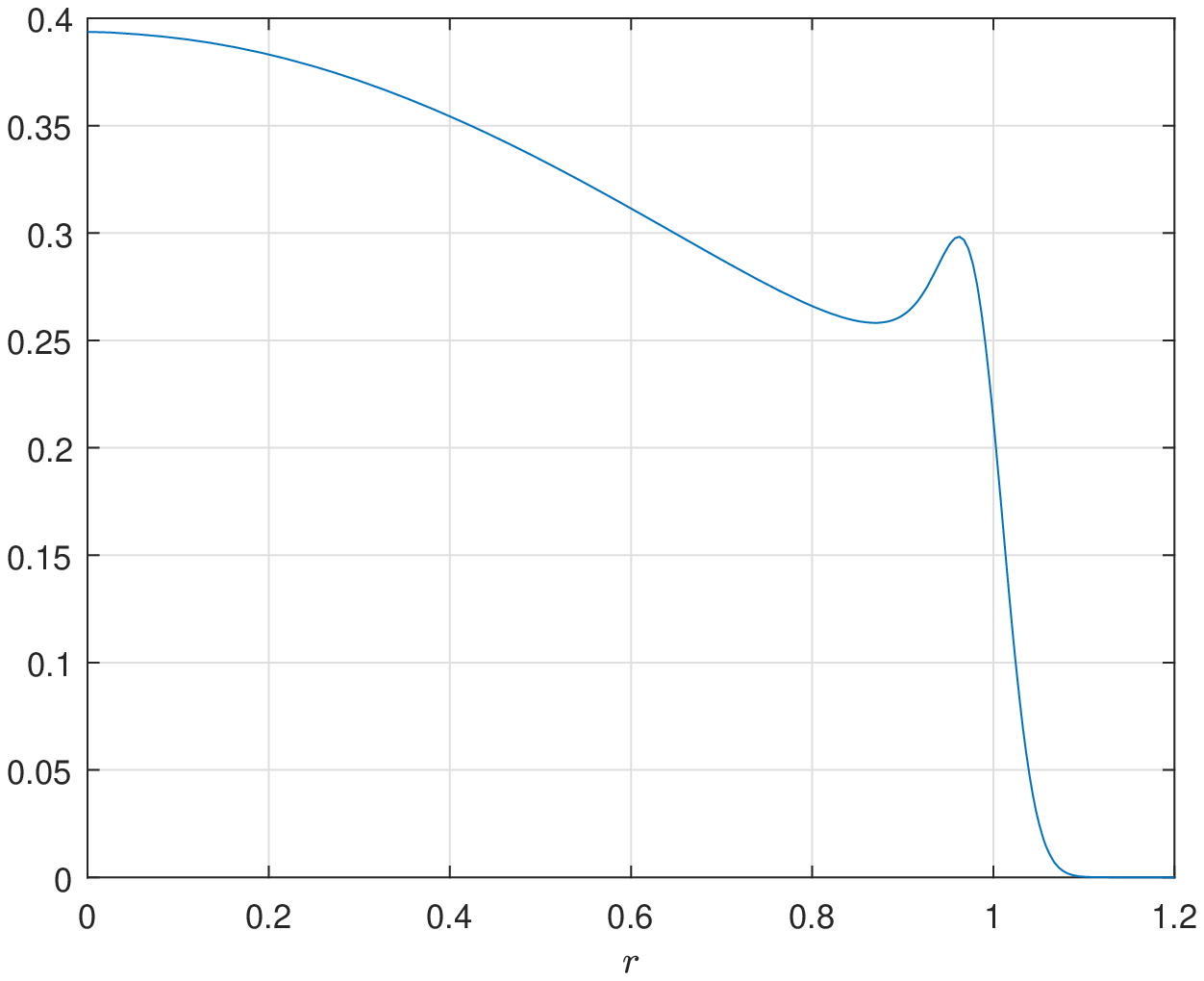}
}
\caption{Exact solution of the density for the line source problem at $t=1$.}
\label{fig:LineSource_exact}
\end{figure}

Since the intensity function decays exponentially, it suffices to set the computational domain to be $[-1.5, 1.5] \times [-1.5, 1.5]$. A uniform grid with $400 \times 400$ grid cells is used to discretize the spatial domain, and the boundary conditions are simulated using the ghost-cell method with all the moments set to be zero in the ghost cells. When solving the moment inversion problem, the Lagrange multiplier solved in the previous time step is used as the initial value of Newton's iteration. More details of the numerical method can be found in Section \ref{sec:num_meth}. We terminate the computation at $t = 1$.

The $P_N$ and $M_N$ results can both be found in the reference \cite{Garrett2013comparison}, where it shows that the $P_N$ method has strong oscillations at $N = 11$. As for the $M_N$ method, although the oscillation still exists, it is much milder. However, due to the inexact numerical integration in the implementation of the $M_N$ method, ray effect can still be observed in the 2D plots. Here we expect that the results of $\beta_{N,K}$ model should have oscillations with amplitude between the $P_N$ model and the $M_N$ model, and the ray effect can be completely eliminated since the polynomials can be integrated exactly on the sphere.

The results of the density for some $\beta_{5,K}$ models at $t = 1$ is plotted in Figure \ref{fig:LineSourceN5}. The radial symmetry is generally well preserved, although some numerical artifacts leading to slight asymmetry can still be observed due to the square grid cells. Note that when $K$ increases, the numerical solution does not converge to the exact solution. Instead, we expect convergence towards the $M_5$ model. The general wave structure for the $M_5$ model has already formed at $K = 7$, which does not change much at $K = 11$. Figure \ref{fig:LineSourceK5} shows some results for the $\beta_{N,5}$ models. This time we expect convergence to the exact solution as $N$ increases. When $N$ increases, the number of oscillations is larger, while their amplitudes get smaller. By comparing our results with the $P_N$ and $M_N$ models shown in \cite{Garrett2013comparison}, we see that the $\beta_{N,5}$ results are closer to $M_N$ than $P_N$, and the strong oscillations in the $P_N$ results have been remarkably suppressed.

\begin{figure}[!ht]
\centering
\subfloat[$N=5$, $K=3$]{%
  \includegraphics[width=.33\textwidth, bb=50 0 551 421, clip]{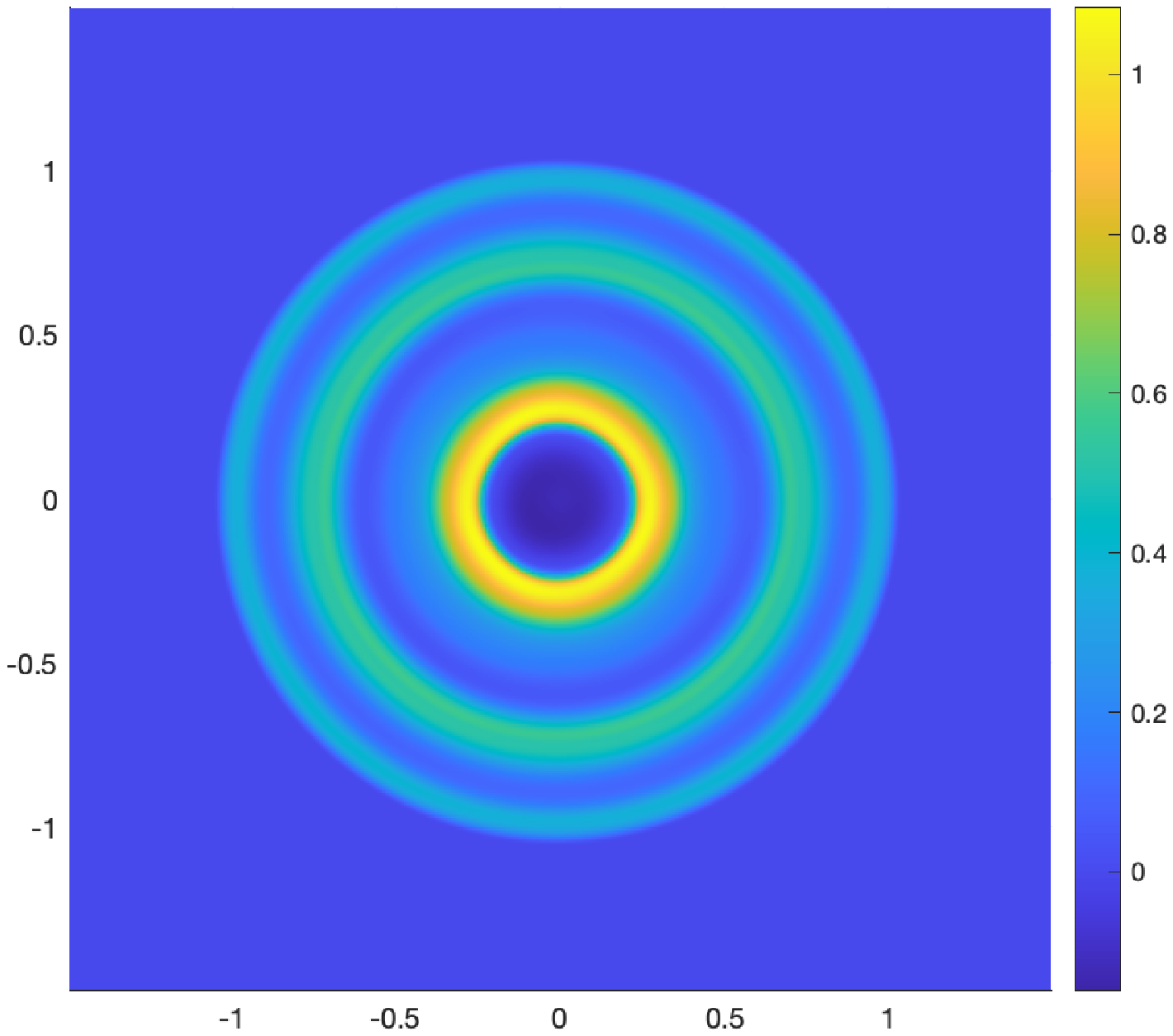}
}
\subfloat[$N=5$, $K=7$]{%
  \includegraphics[width=.33\textwidth, bb=50 0 551 421, clip]{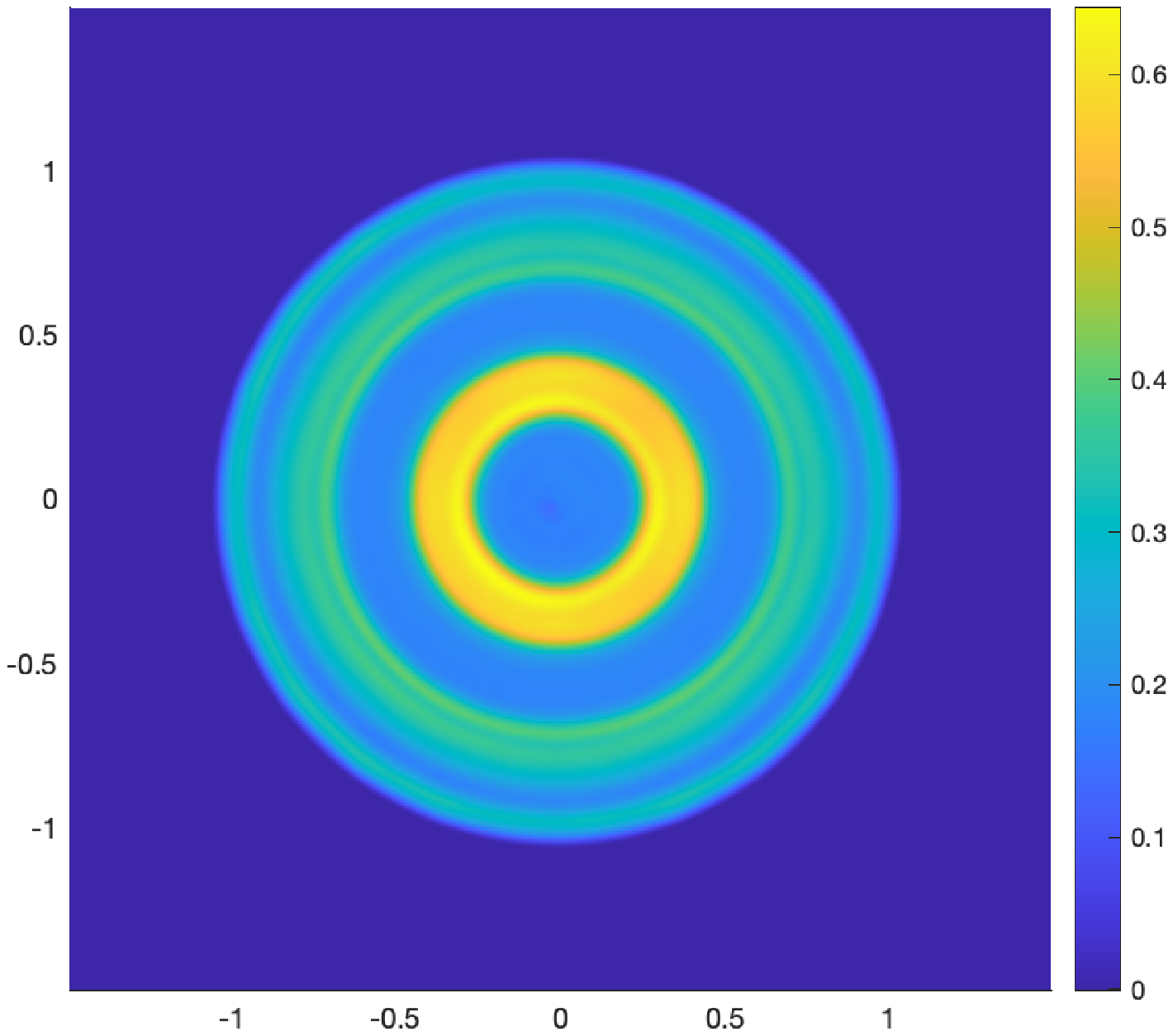}
}
\subfloat[$N=5$, $K=11$]{%
  \includegraphics[width=.33\textwidth, bb=50 0 551 421, clip]{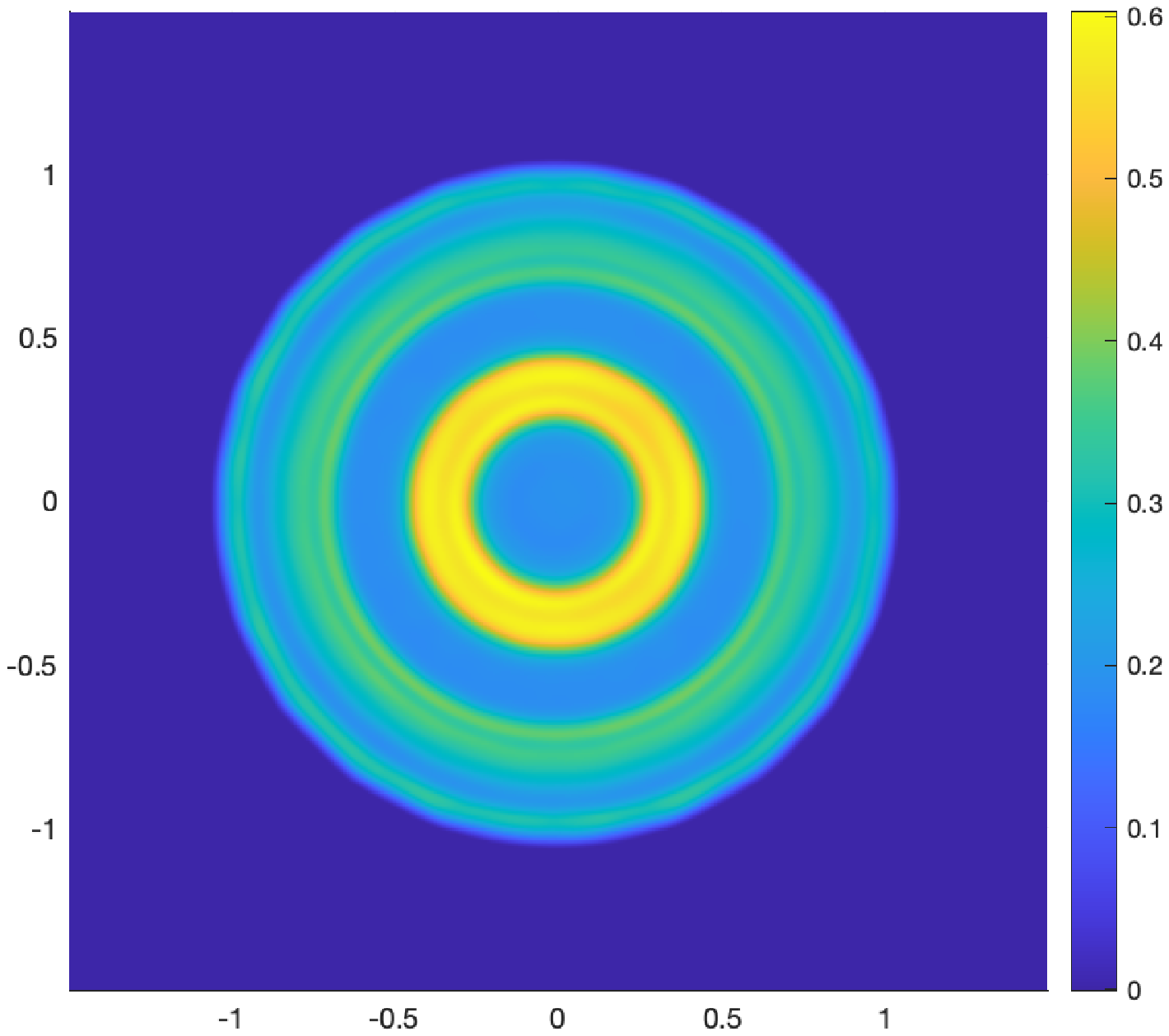}
}
\caption{Solutions of the line source problem for $\beta_{5,K}$ models.}
\label{fig:LineSourceN5}
\end{figure}

\begin{figure}[!ht]
\centering
\subfloat[$N=3$, $K=5$]{%
  \includegraphics[width=.33\textwidth, bb=50 0 551 421, clip]{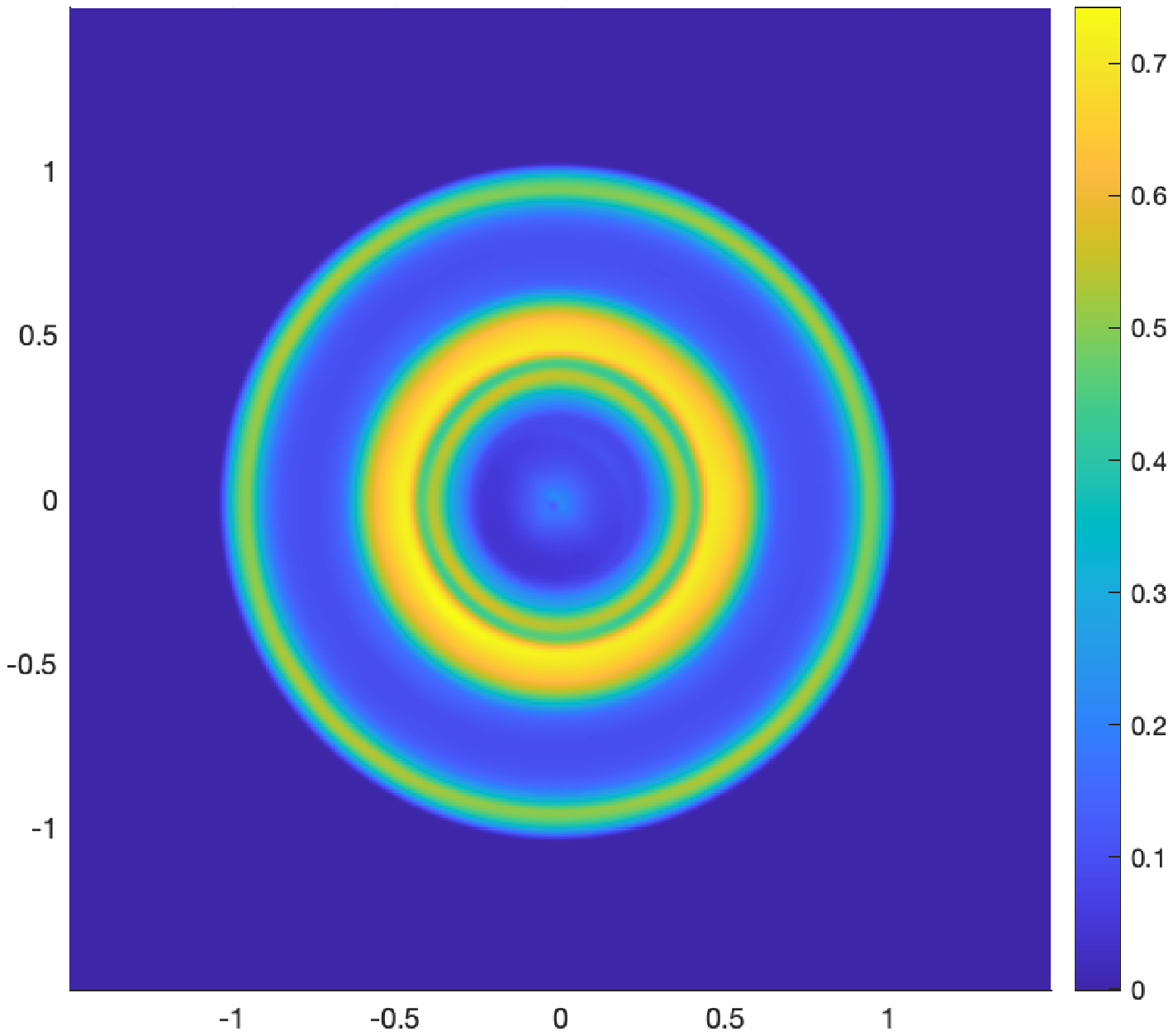}
}
\subfloat[$N=7$, $K=5$]{%
  \includegraphics[width=.33\textwidth, bb=50 0 551 421, clip]{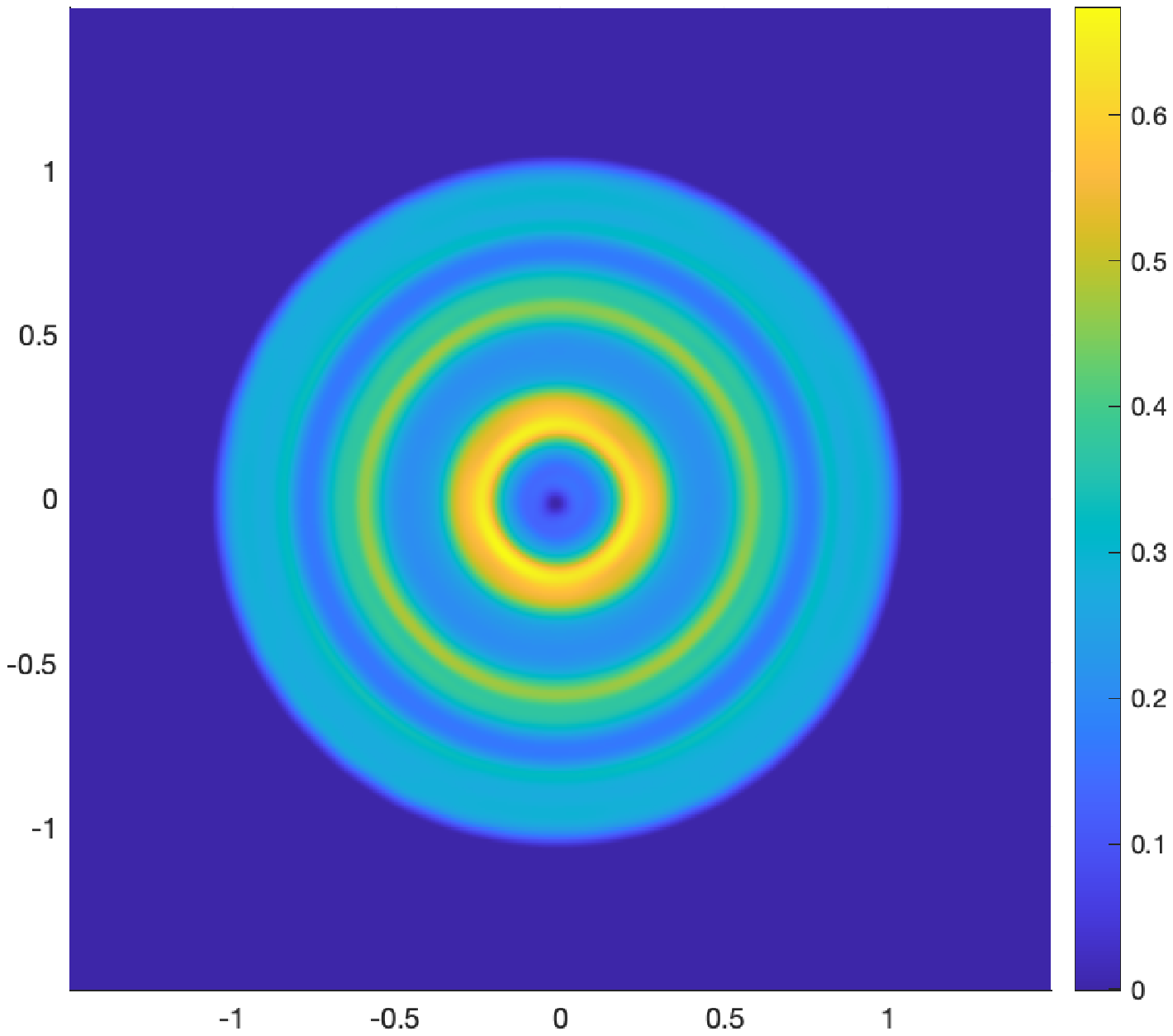}
}
\subfloat[$N=11$, $K=5$]{%
  \includegraphics[width=.33\textwidth, bb=50 0 551 421, clip]{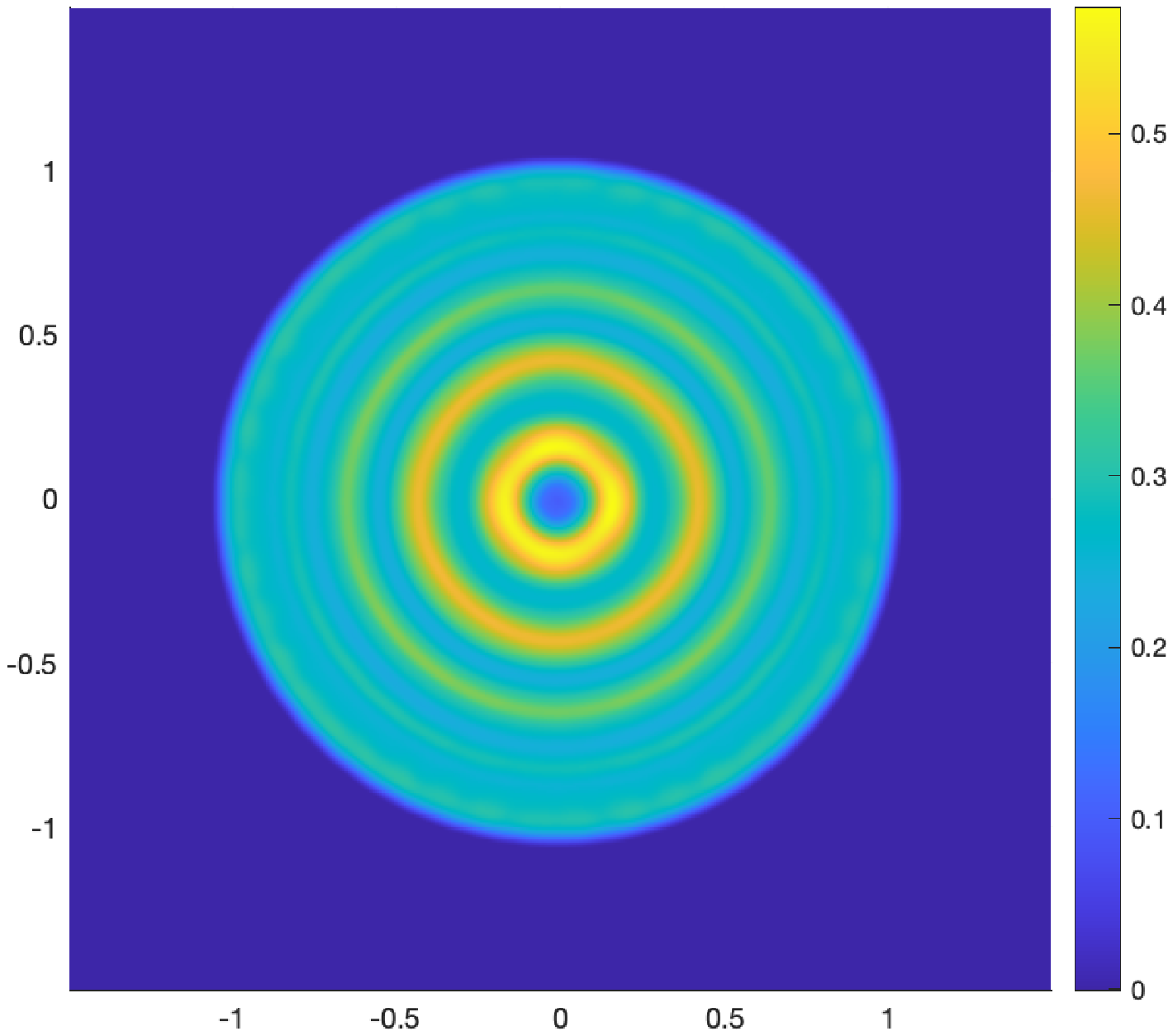}
}
\caption{Solutions of the line source problem for $\beta_{N,5}$ models.}
\label{fig:LineSourceK5}
\end{figure}

To get a clearer comparison between the numerical results and the exact solution, we also plot the numerical solutions and functions of the radial variable. The results for $\beta_{5,K}$ models are given in Figure \ref{fig:LineSourceN5_1d}, from which one can find a significant negative part near the origin for $K = 3$. Starting from $K = 5$, the positivity of the solution is well maintained, and the converging trend (towards the $M_5$ model instead of the exact solution) is obvious as $K$ increases. The $\beta_{N,5}$ solutions are given in Figure \ref{fig:LineSourceK5_1d}. Due to the different behaviors of odd and even $N$'s, we plot the results in two separate figures. One can find that the results for even $N$'s have a peak at the origin, whereas the origin is a valley for odd $N$'s.
\begin{figure}[!ht]
\centering
  \includegraphics[width=.5\textwidth]{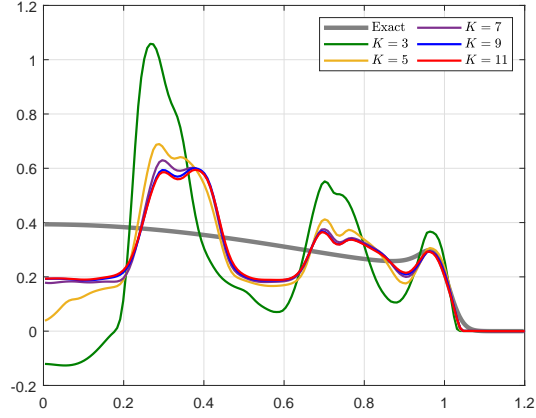}
\caption{Solutions of the line source problem for $\beta_{5,K}$ models}
\label{fig:LineSourceN5_1d}
\end{figure}

\begin{figure}[!ht]
\centering
\subfloat[Odd $N$]{%
  \includegraphics[width=.5\textwidth]{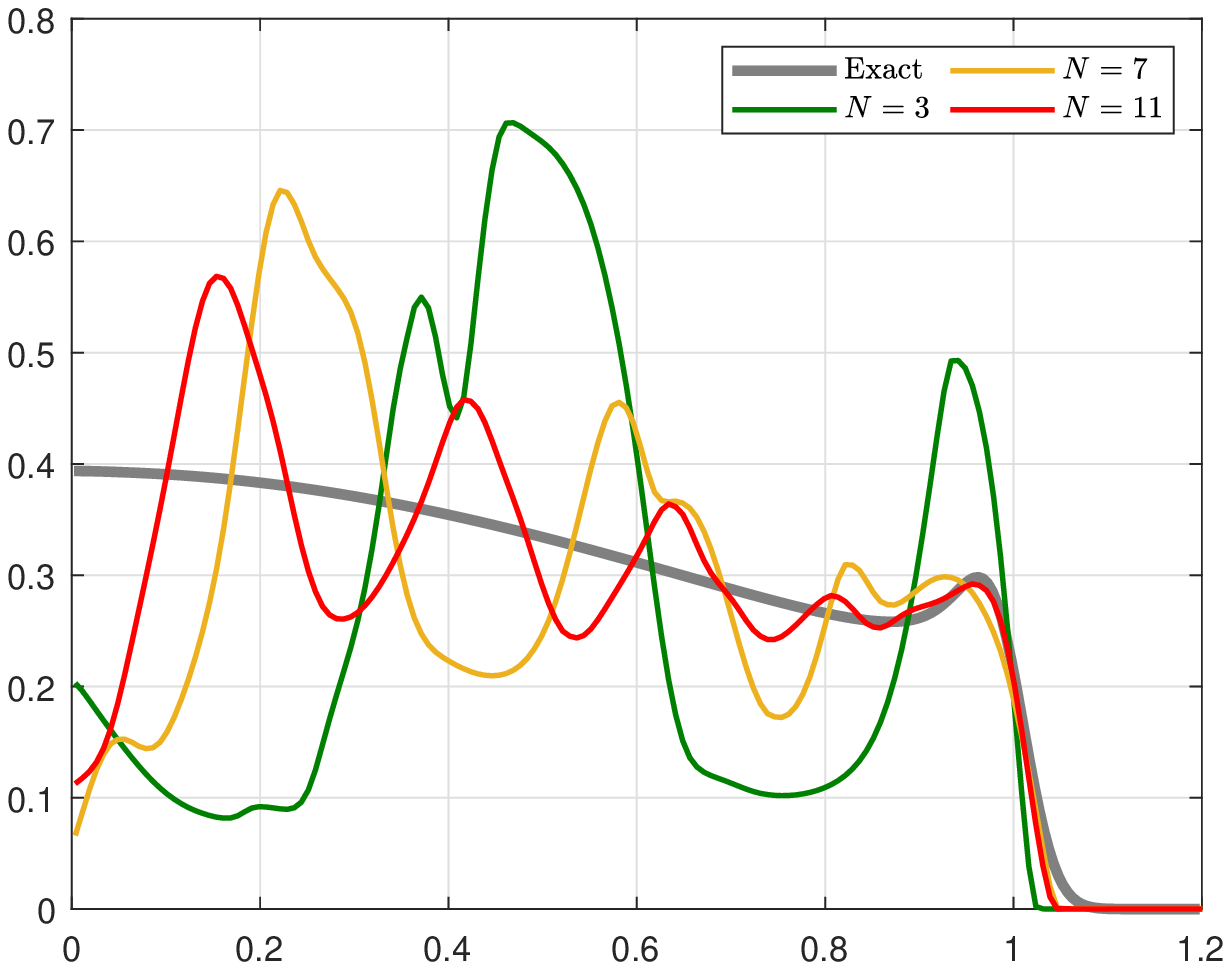}
}%
\subfloat[Even $N$]{%
  \includegraphics[width=.5\textwidth]{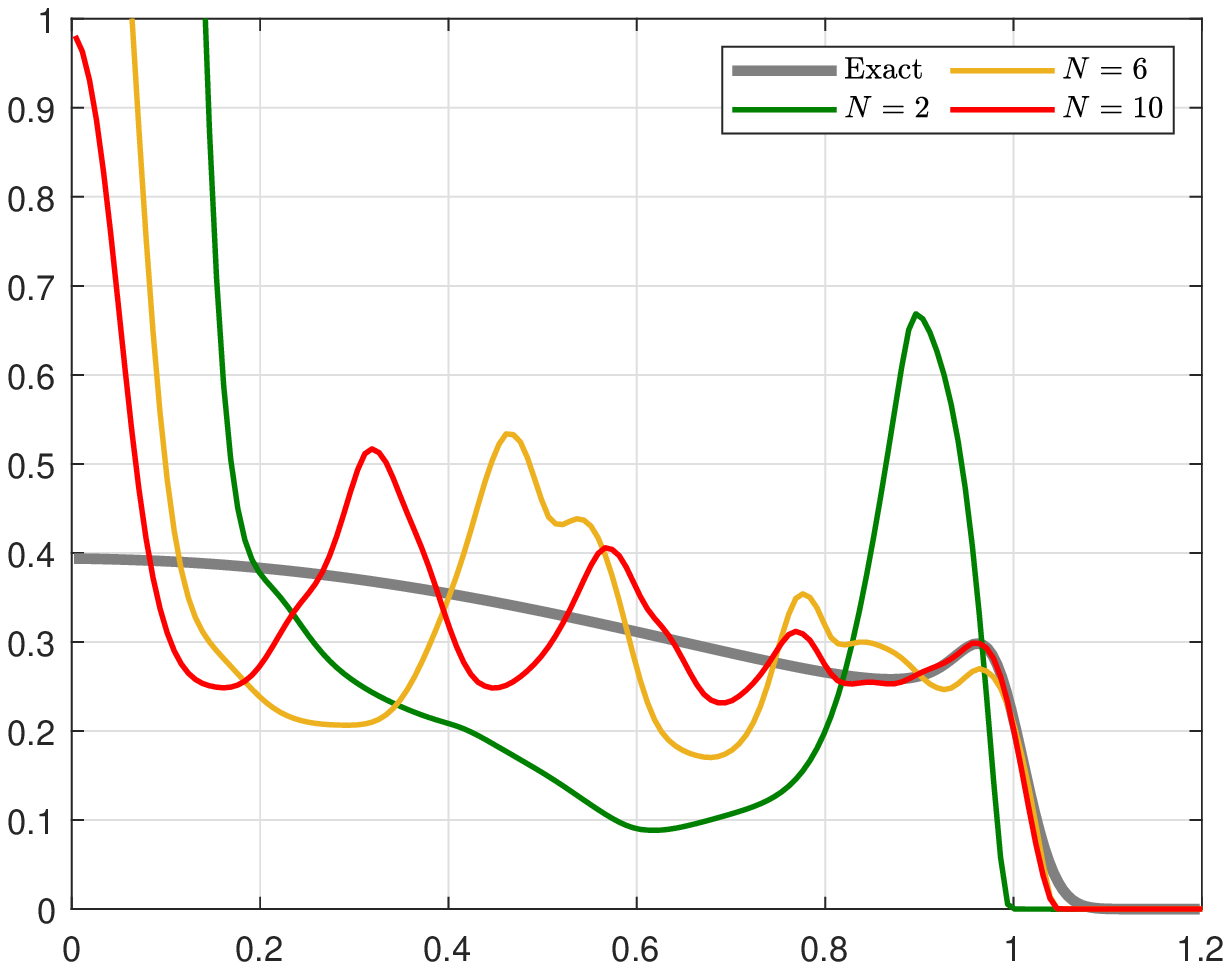}
}
\caption{Solutions of the line source problem for $\beta_{N,5}$ models.}
\label{fig:LineSourceK5_1d}
\end{figure}

By comparison with the results in \cite{Garrett2013comparison}, the $\beta_{N,K}$ models have satisfactory results among the unfiltered models, especially for small values of $K$ like $K = 5$. For this particular problem, further improvements may be made by adding filters, which will be studied in our future works.

\subsection{Two-beam interaction}
Our third numerical example is a test in the domain $[-1/2, 1/2] \times [-1/2, 1/2]$ with the following boundary conditions:
\begin{itemize}
  \item Left boundary condition: for $\Omega = (\Omega_1, \Omega_2, \Omega_3)^{\top}$ with $\Omega_1 > 0$,
  \begin{displaymath}
    I(-1/2, y, \Omega, t) = \left\{ \begin{array}{@{}ll}
      \delta(\Omega - e_x), & \text{if } y \in [-1/8, 1/8], \\
      0, & \text{otherwise},
    \end{array} \right.
  \end{displaymath}
  where $e_x = (1,0,0)^{\top}$.
  \item Bottom boundary condition: for $\Omega = (\Omega_1, \Omega_2, \Omega_3)^{\top}$ with $\Omega_2 > 0$,
  \begin{displaymath}
    I(x, -1/2, \Omega, t) = \left\{ \begin{array}{@{}ll}
      \delta(\Omega - e_y), & \text{if } x \in [-1/8, 1/8], \\
      0, & \text{otherwise},
    \end{array} \right.
  \end{displaymath}
  where $e_y = (0,1,0)^{\top}$.
  \item Right boundary condition: for $\Omega = (\Omega_1, \Omega_2, \Omega_3)^{\top}$ with $\Omega_1 < 0$,
  \begin{displaymath}
    I(1/2,y,\Omega,t) = 0.
  \end{displaymath}
  \item Top boundary condition: for $\Omega = (\Omega_1, \Omega_2, \Omega_3)^{\top}$ with $\Omega_2 < 0$,
  \begin{displaymath}
    I(x,1/2,\Omega,t) = 0.
  \end{displaymath}
\end{itemize}
These boundary conditions indicate beams with width $1/4$ injecting into the domain from the left and the bottom. The initial condition is a vacuum in the domain:
\begin{displaymath}
I(x,y,\Omega,0) = 0, \qquad \forall x \in [-1/2, 1/2] \text{ and } y \in [-1/2, 1/2].
\end{displaymath}
The spatial domain is discretized by a uniform grid with $400 \times 400$ cells. The numerical solver used in this test again follows the scheme described in Section \ref{sec:num_meth}.

We first study the two-beam problem with no scattering ($\sigma = 0$). In this case, the two beams will cross each other without interaction. According to the discussion in Section \ref{sec:Dirac}, the $M_2$ model can simulate this problem exactly since the exact solution includes only one-beam or two-beam intensity functions. In our experiments, we test the $\beta_{3,K}$ models, and present in Figure \ref{fig:two_beam_N3} the results at $t = 0.5$ (both beams reach middle of the domain), $t = 1$ (both beams reach the other side of the domain) and $t = 1.1$ (both beams fully penetrate the domain). At $t = 0.5$ (the first column of Figure \ref{fig:two_beam_N3}), one can already observe significant outspreading of the beams. Increasing the value of $K$ can help make the beams more concentrated, but such effect is not strong enough due to the slow convergence rate as we have seen in Figure \ref{fig:single_Dirac_error_K}. At $t = 1$ (the second column of Figure \ref{fig:two_beam_N3}), in the numerical solutions, there is still an obvious gap between the front of the radiation and the other side of the boundary, which implies that the maximum characteristic speed for $\beta_{N,K}$ models is less than $1$. The gap narrows for larger $K$, but a characteristic speed equal to $1$ can only be achieved in the limiting case, \textit{i.e.}, the $M_N$ model. The last column of Figure \ref{fig:two_beam_N3} also shows the slow improvement of the solution as $K$ increases.

\begin{figure}
  \centering
  \subfloat[$K=3$, $t = 0.5$]{%
    \includegraphics[width=.33\textwidth, bb=35 10 386 316, clip]{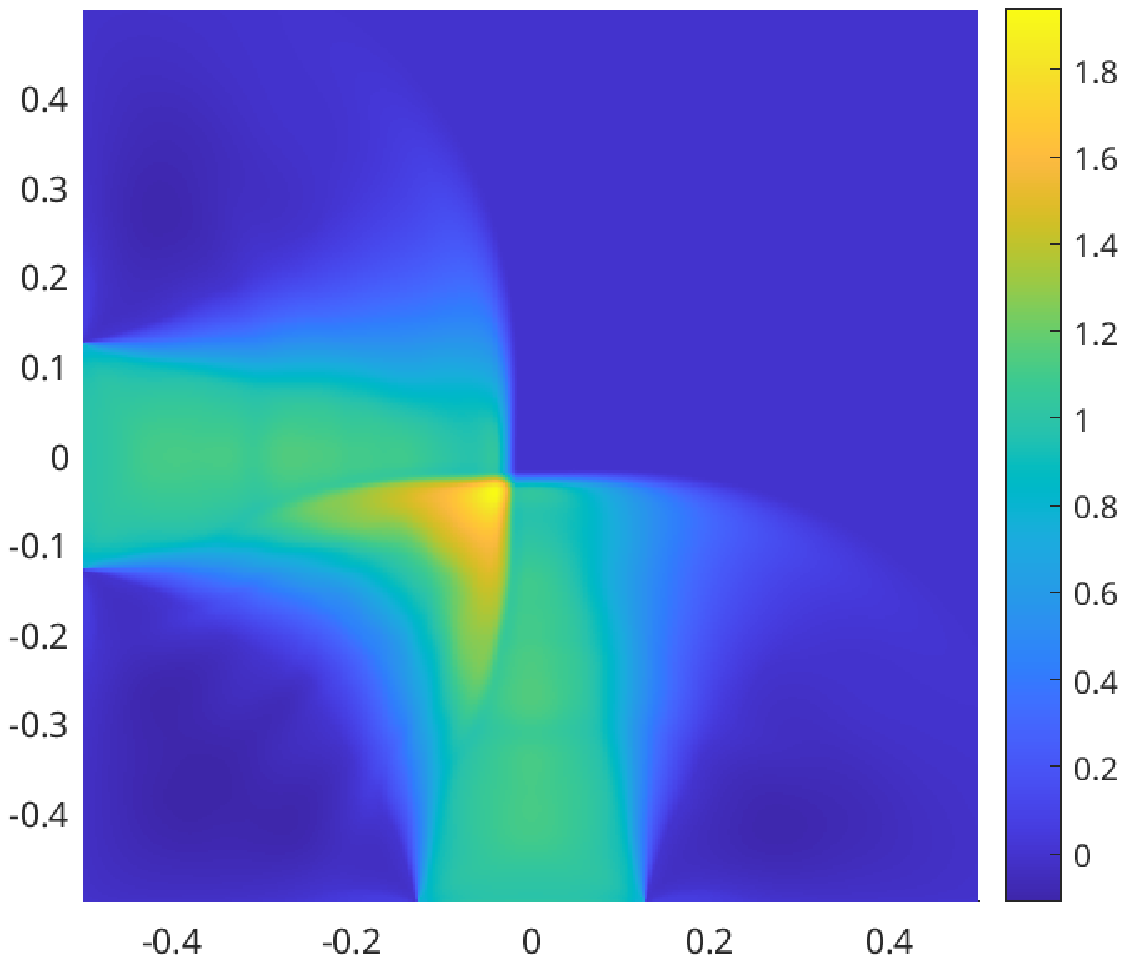}
  }
  \subfloat[$K=3$, $t = 1$]{%
    \includegraphics[width=.33\textwidth, bb=35 10 386 316, clip]{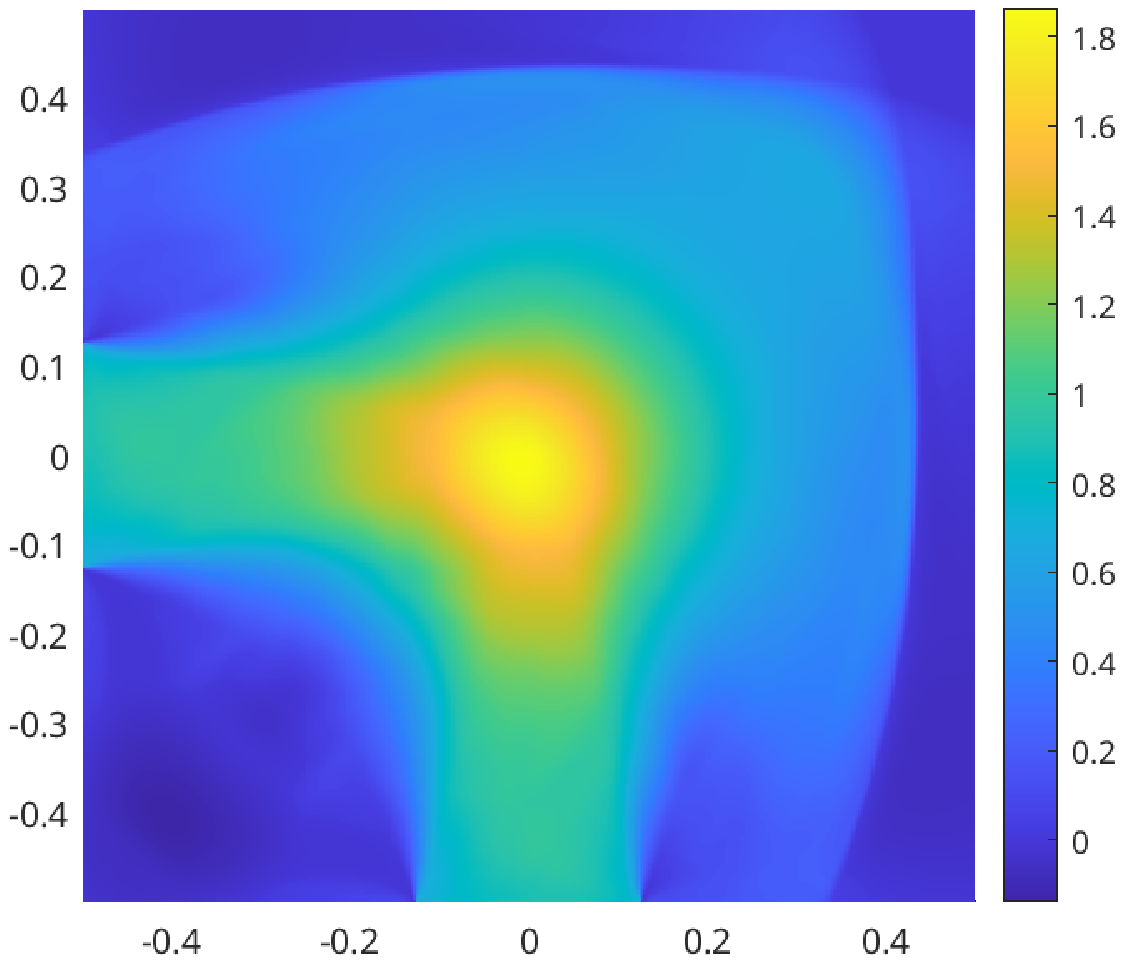}
  }
  \subfloat[$K=3$, $t = 1.1$]{%
    \includegraphics[width=.33\textwidth, bb=35 10 386 316, clip]{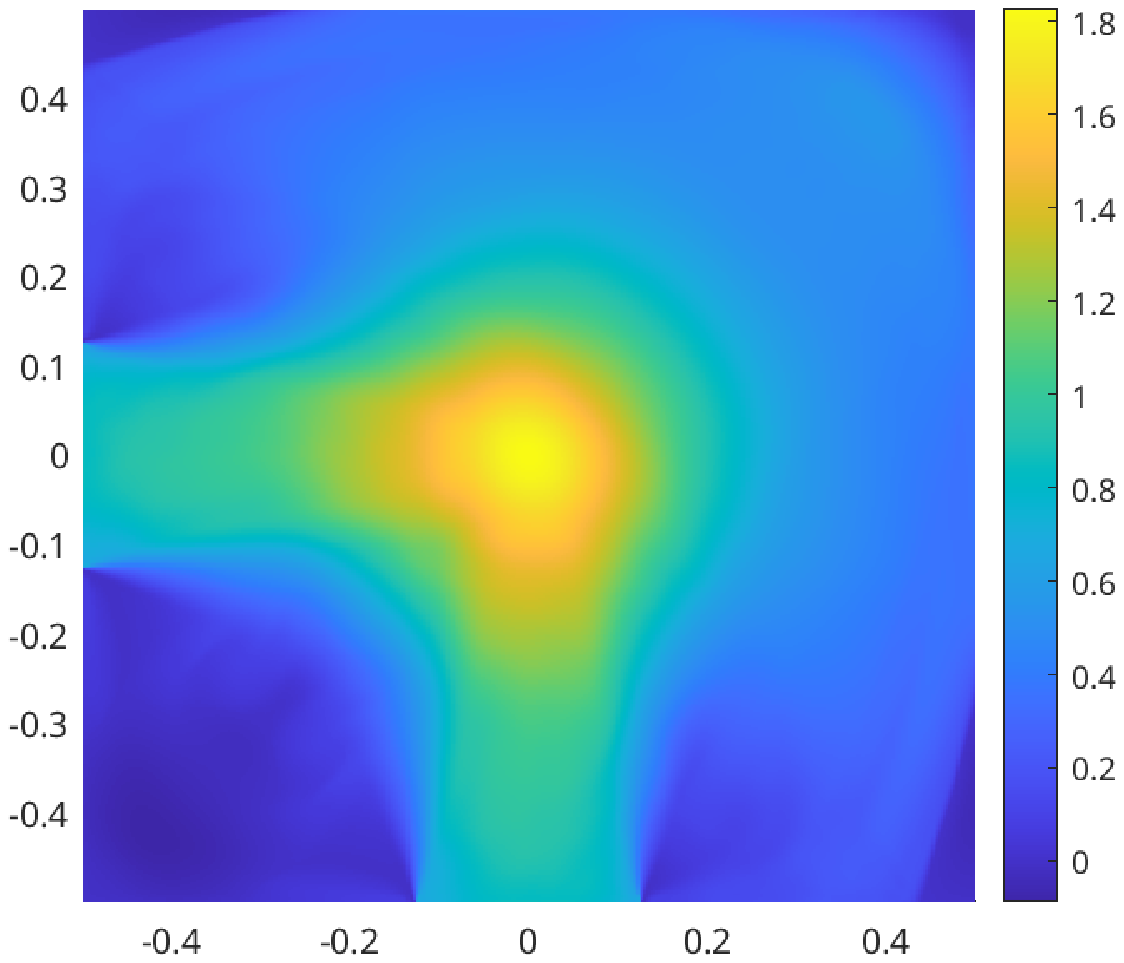}
  } \\
  \subfloat[$K=5$, $t = 0.5$]{%
    \includegraphics[width=.33\textwidth, bb=35 10 386 316, clip]{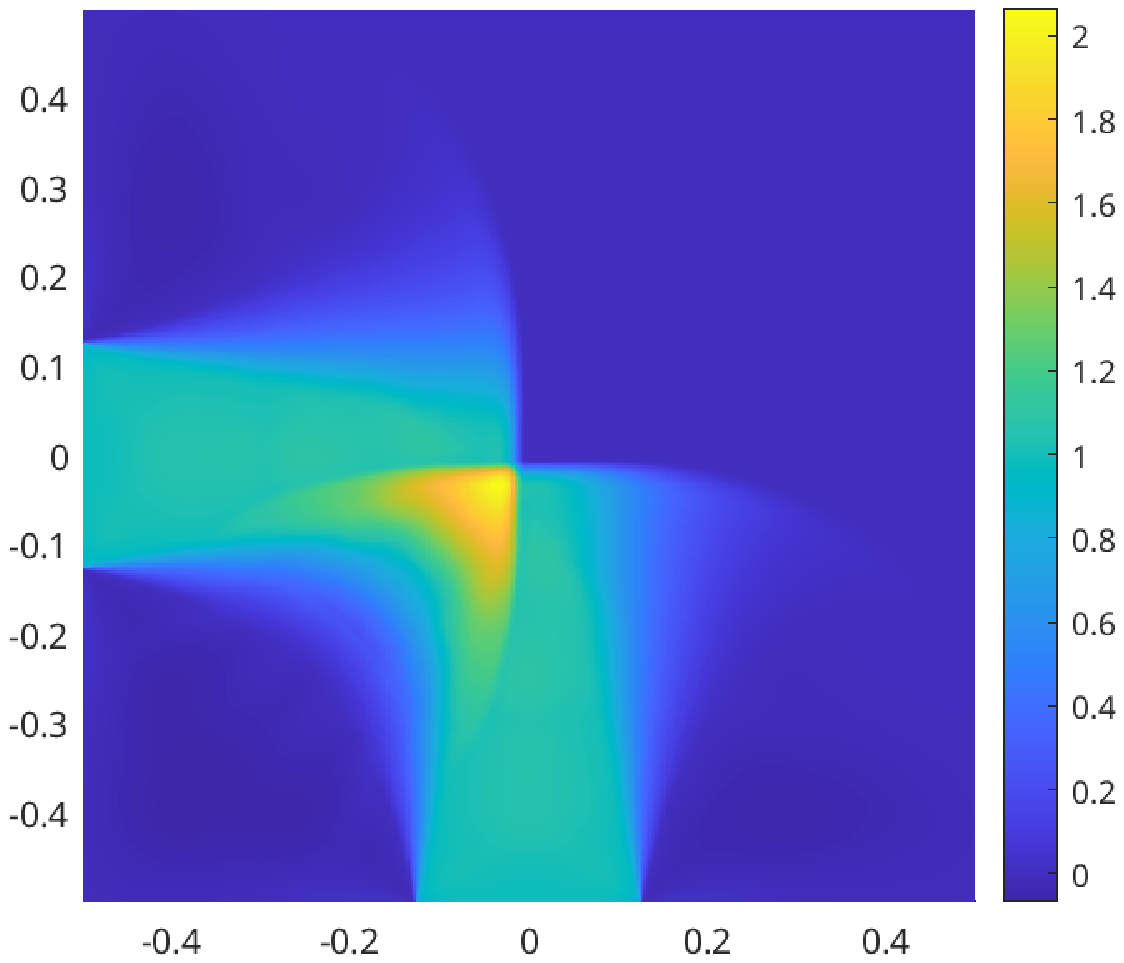}
  }
  \subfloat[$K=5$, $t = 1$]{%
    \includegraphics[width=.33\textwidth, bb=35 10 386 316, clip]{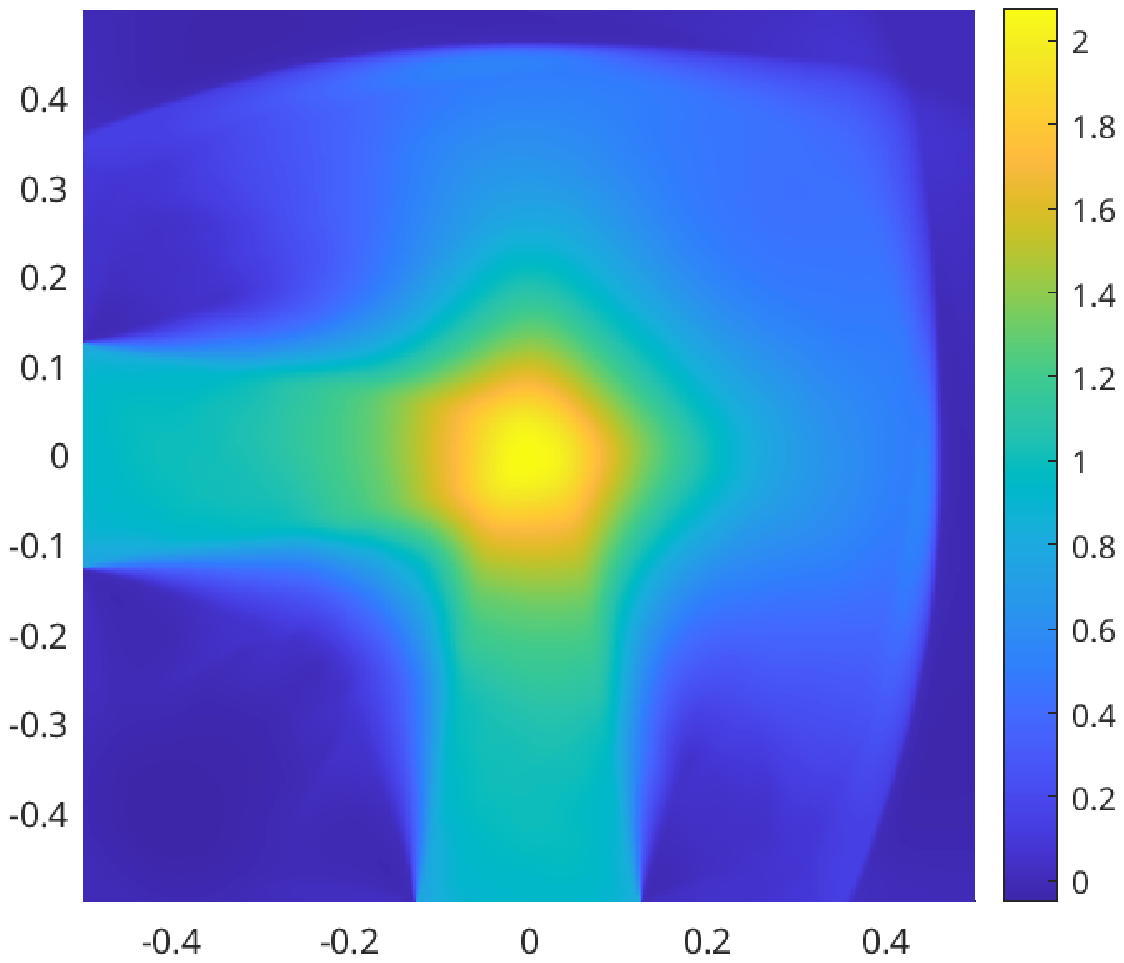}
  }
  \subfloat[$K=5$, $t = 1.1$]{%
    \includegraphics[width=.33\textwidth, bb=35 10 386 316, clip]{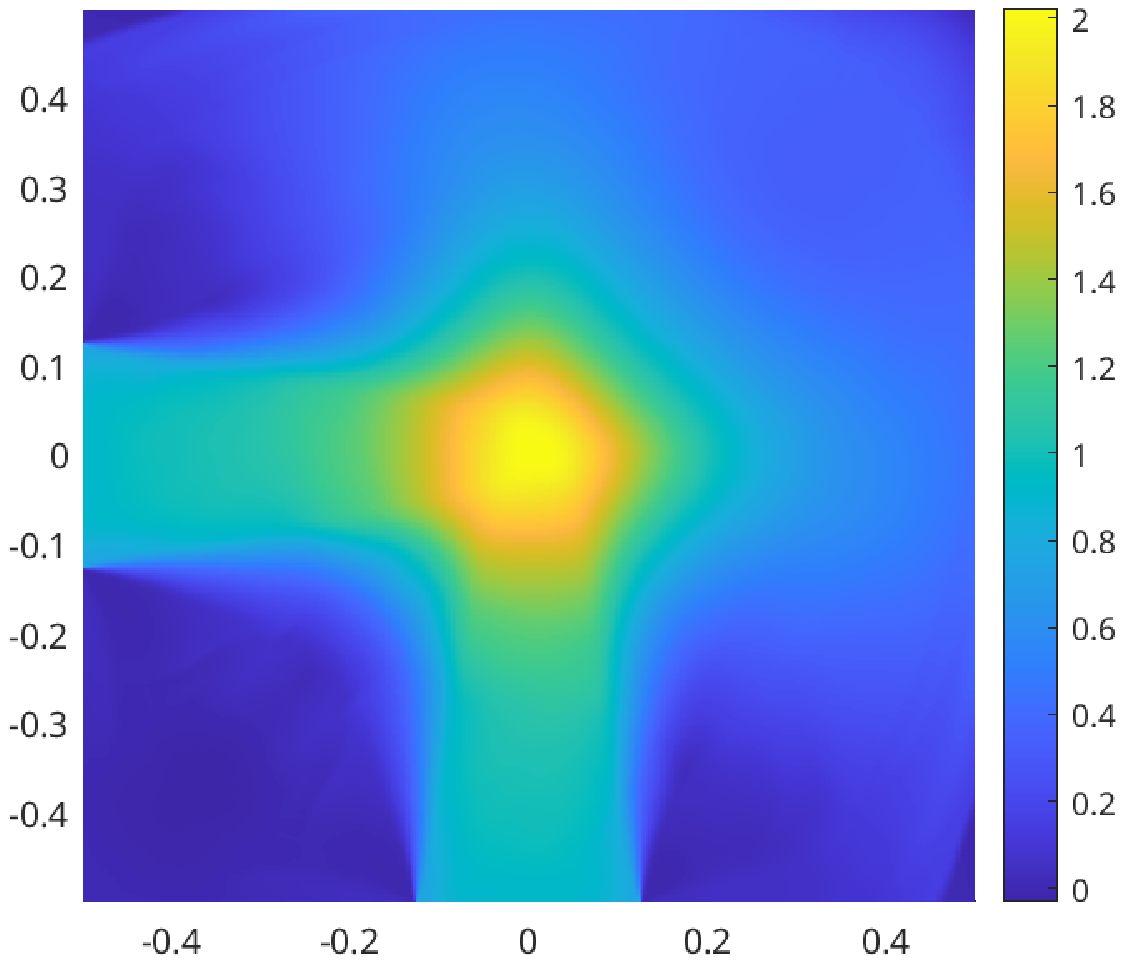}
  } \\
  \subfloat[$K=7$, $t = 0.5$]{%
    \includegraphics[width=.33\textwidth, bb=35 10 386 316, clip]{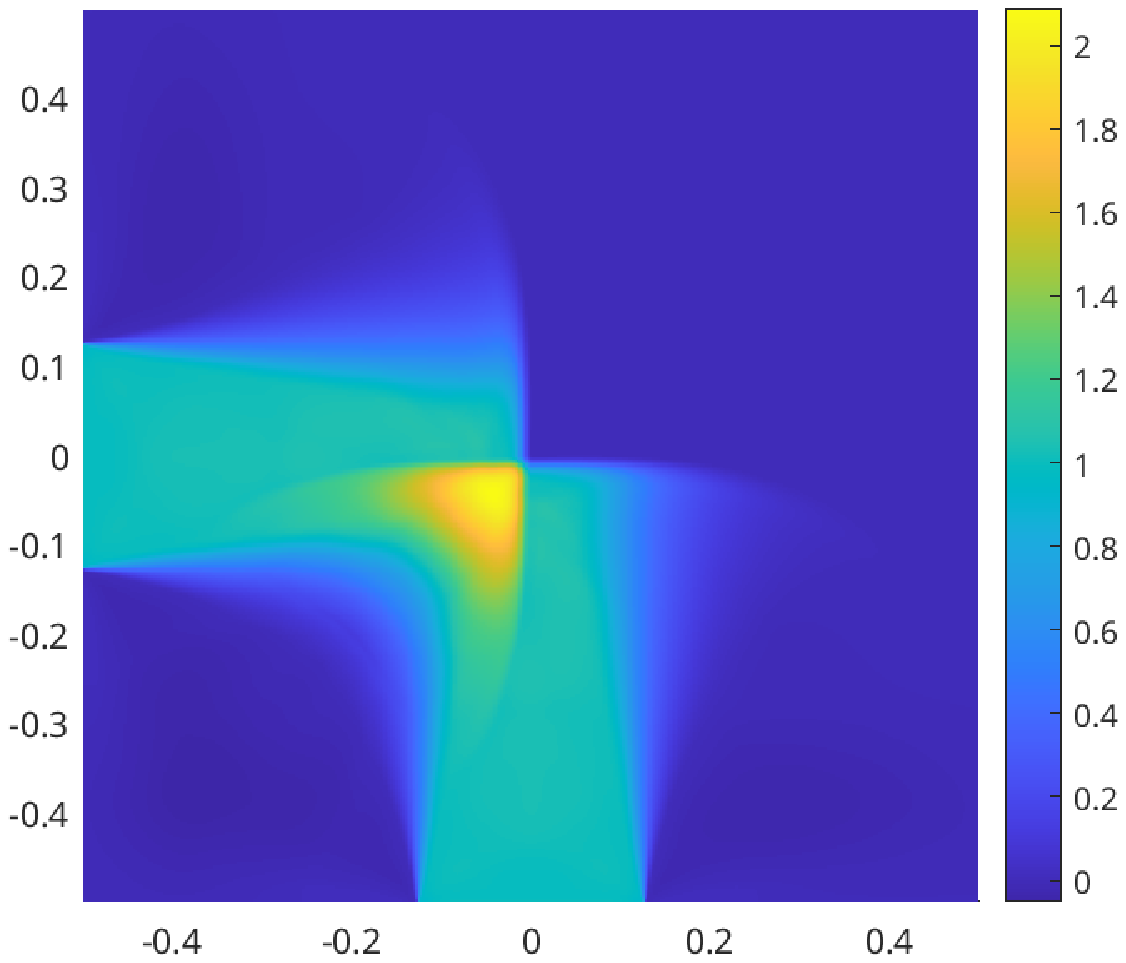}
  }
  \subfloat[$K=7$, $t = 1$]{%
    \includegraphics[width=.33\textwidth, bb=35 10 386 316, clip]{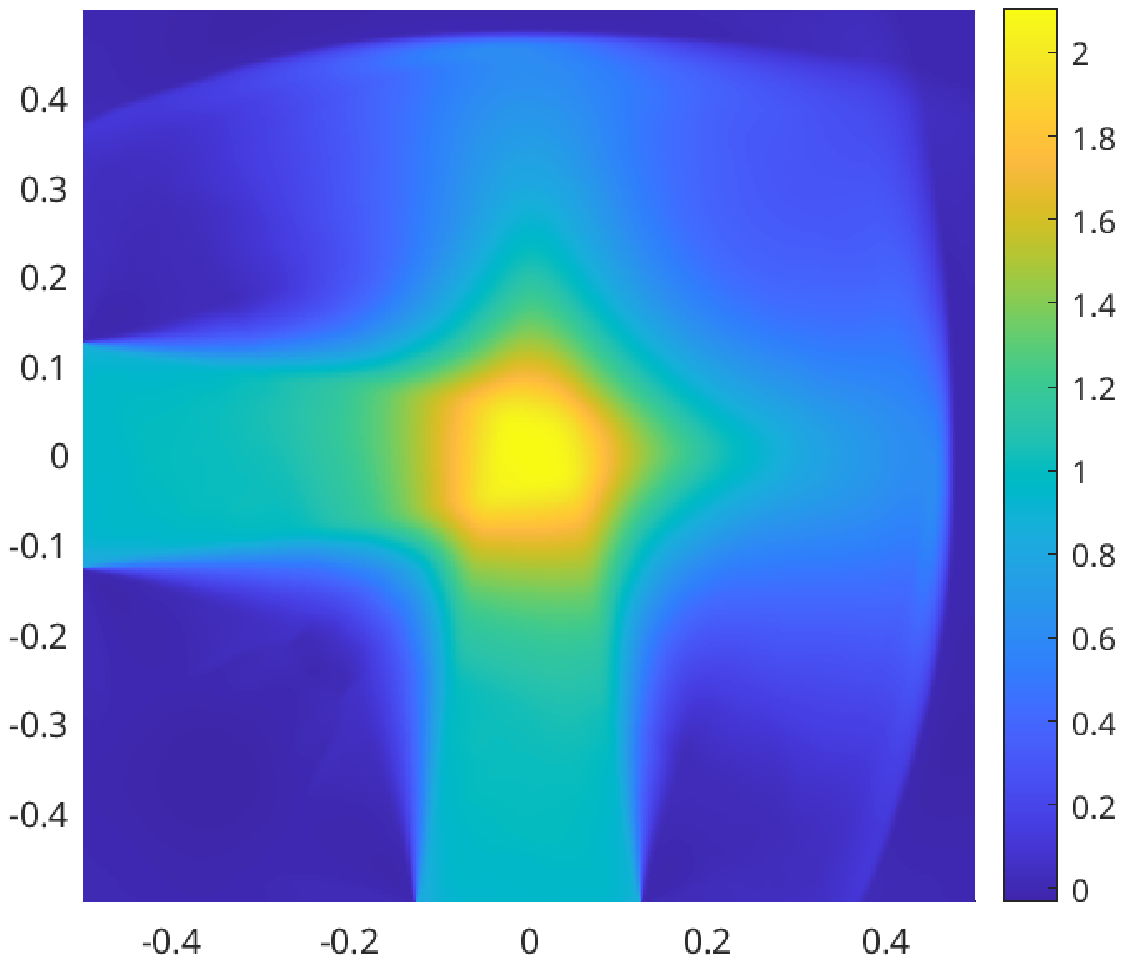}
  }
  \subfloat[$K=7$, $t = 1.1$]{%
    \includegraphics[width=.33\textwidth, bb=35 10 386 316, clip]{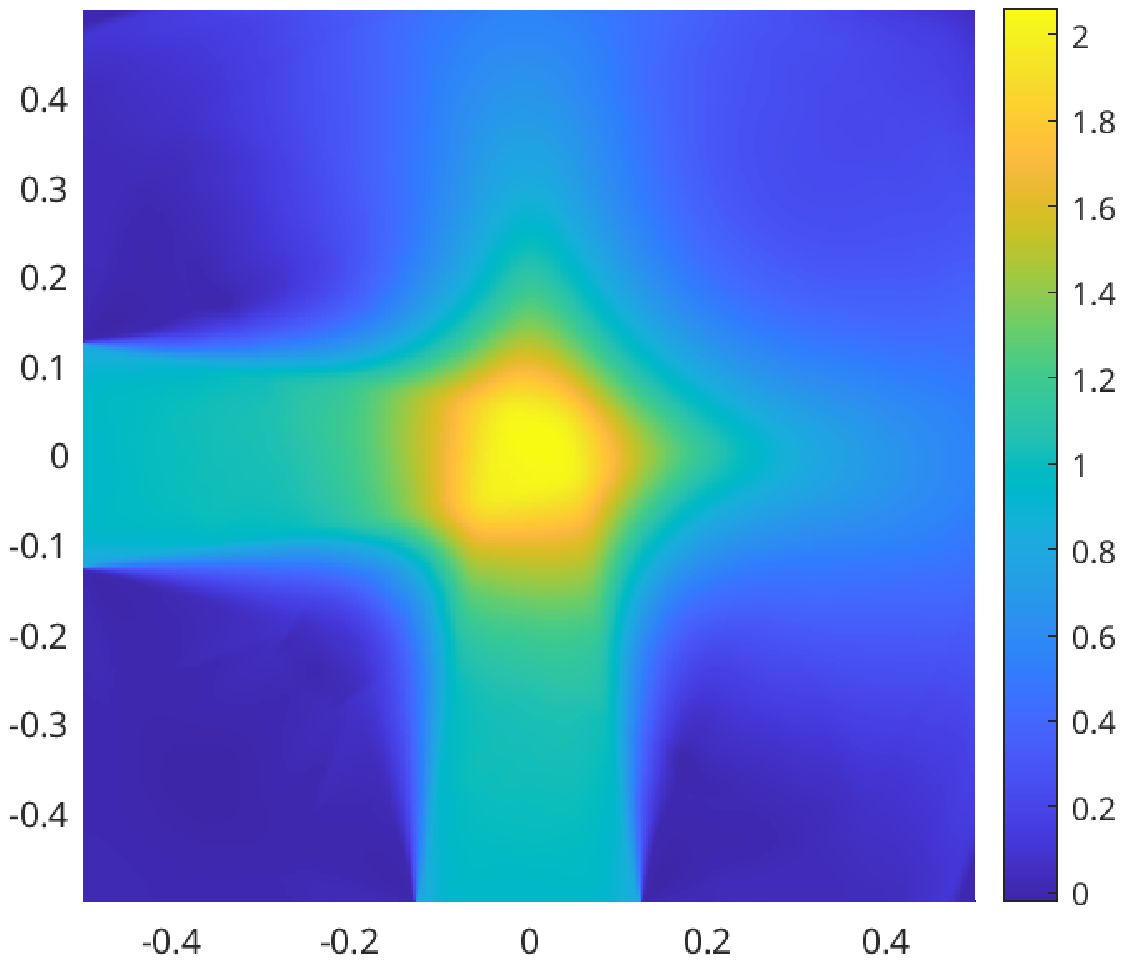}
  } \\
  \subfloat[$K=9$, $t = 0.5$]{%
    \includegraphics[width=.33\textwidth, bb=35 10 386 316, clip]{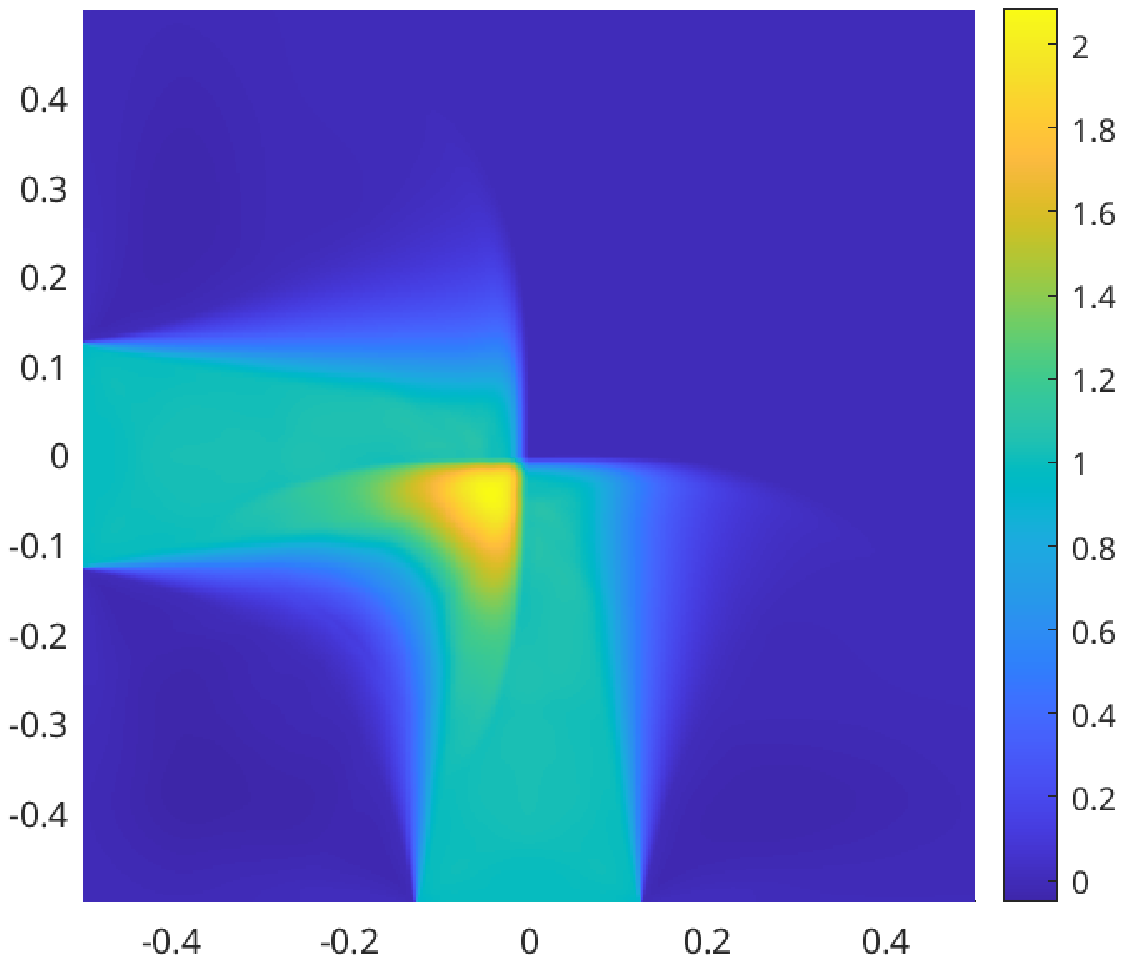}
  }
  \subfloat[$K=9$, $t = 1$]{%
    \includegraphics[width=.33\textwidth, bb=35 10 386 316, clip]{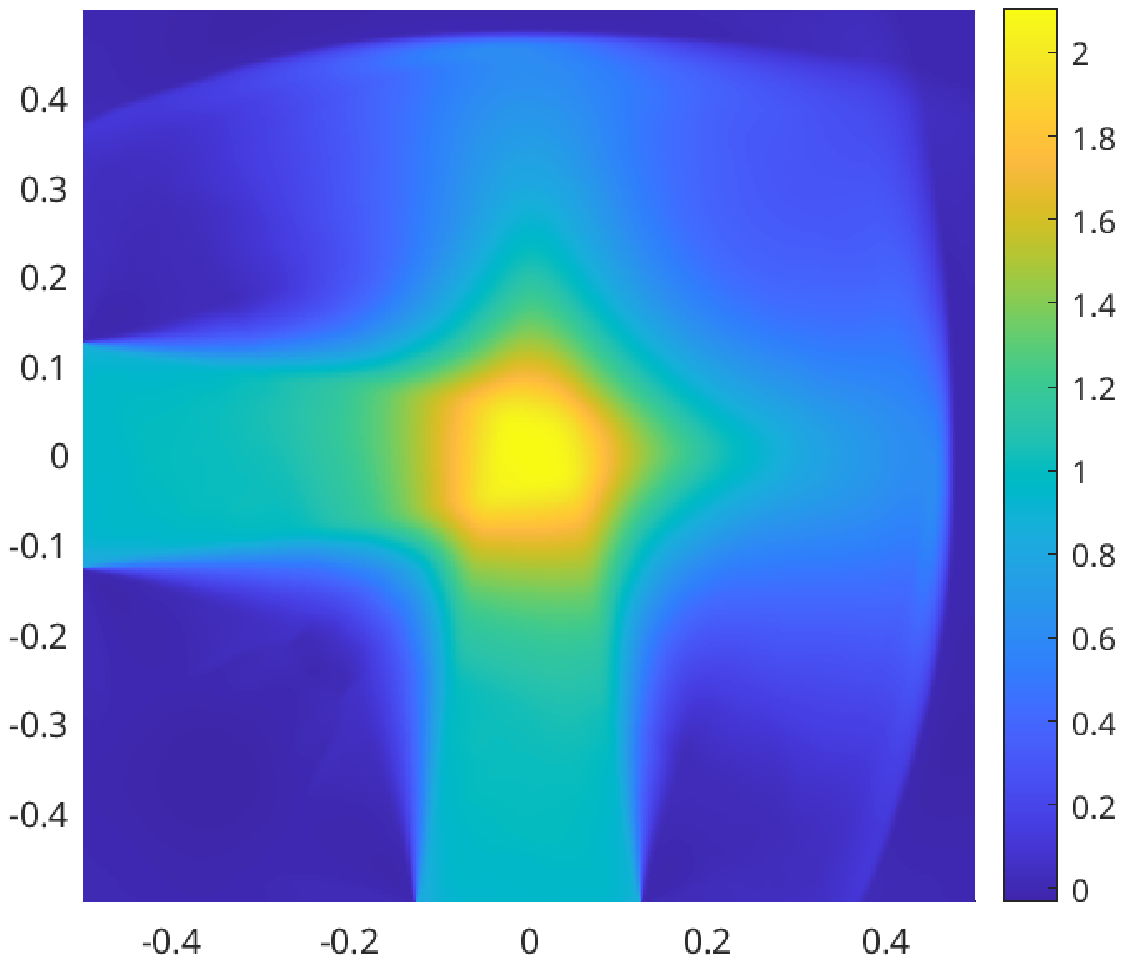}
  }
  \subfloat[$K=9$, $t = 1.1$]{%
    \includegraphics[width=.33\textwidth, bb=35 10 386 316, clip]{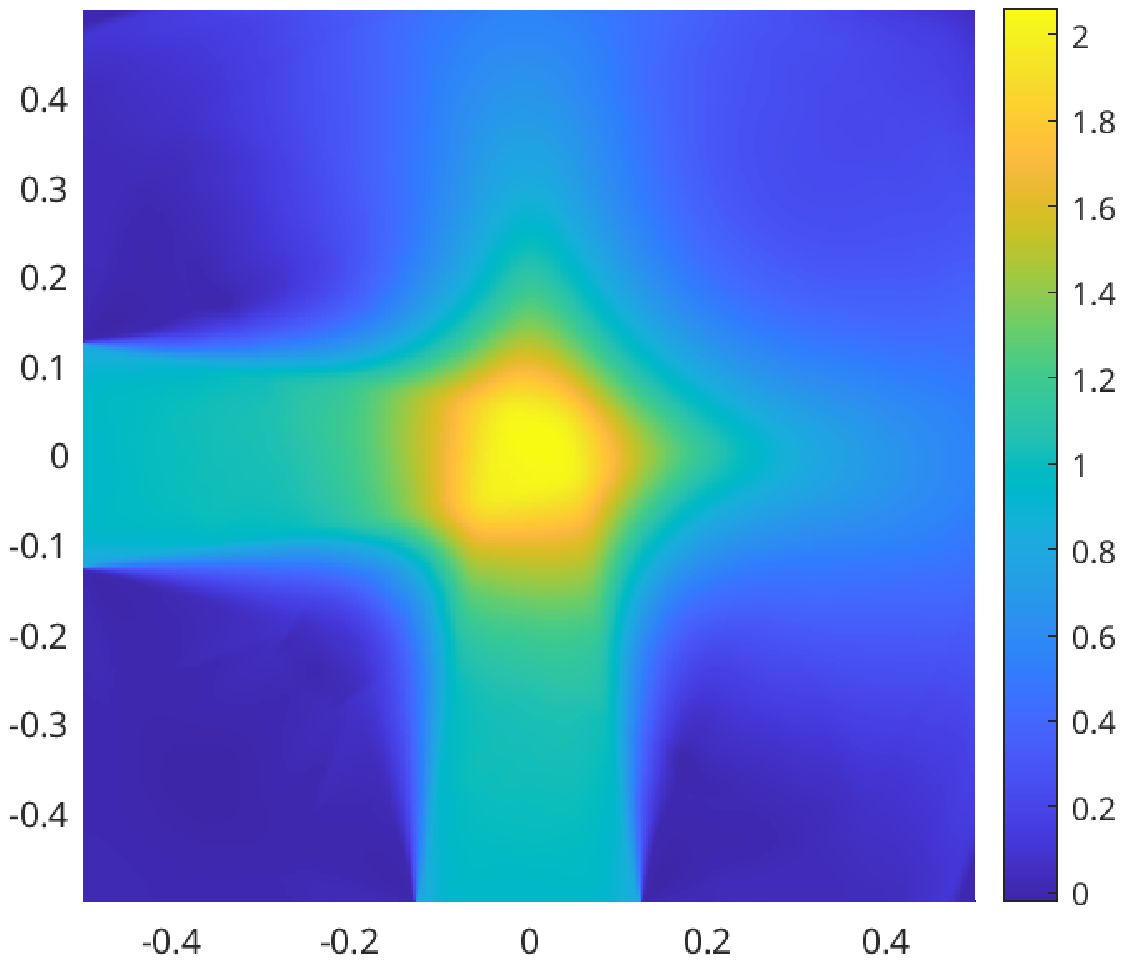}
  }
  \caption{Solution of the two-beam problem for the $\beta_{3,K}$ models}
  \label{fig:two_beam_N3}
\end{figure}

We now fix the value of $K$ and increase $N$. Some results are plotted in Figure \ref{fig:two_beam_K5}. It is clear that increasing $N$ is more effective than increasing $K$, due to its faster convergence towards the Dirac delta function (see Figure \ref{fig:single_Dirac_error_N}). Note that the computational cost also increases more quickly since the number of moments in the moment equations is $(N+1)^2$. At $t = 1$, for the $\beta_{7,5}$ and $\beta_{9,5}$ models, both beams have almost reached the other side of the boundary, and the crossing part is quite similar to a square. In most of our numerical results, the greatest value can be slightly larger than $2$ due to the overlapping of the radiations coming from different boundary points caused by the approximation of beams. But the positivity of the solution is generally well maintained.

\begin{figure}
  \centering
  \subfloat[$N=3$, $t = 0.5$]{%
    \includegraphics[width=.33\textwidth, bb=35 10 386 316, clip]{TwoBeam/k3N5_t0.5.eps}
  }
  \subfloat[$N=3$, $t = 1$]{%
    \includegraphics[width=.33\textwidth, bb=35 10 386 316, clip]{TwoBeam/k3N5_t1.eps}
  }
  \subfloat[$N=3$, $t = 1.1$]{%
    \includegraphics[width=.33\textwidth, bb=35 10 386 316, clip]{TwoBeam/k3N5_t1.1.eps}
  } \\
  \subfloat[$N=5$, $t = 0.5$]{%
    \includegraphics[width=.33\textwidth, bb=35 10 386 316, clip]{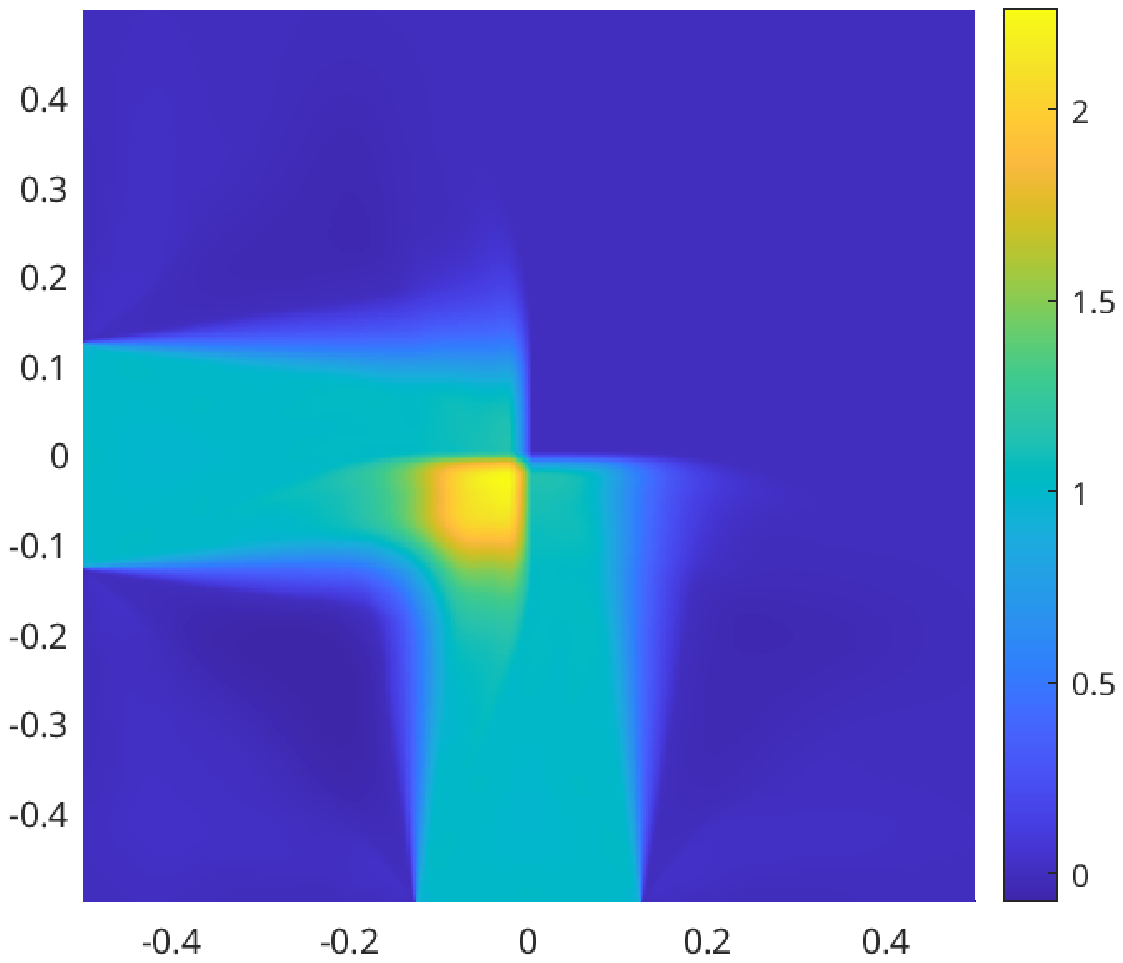}
  }
  \subfloat[$N=5$, $t = 1$]{%
    \includegraphics[width=.33\textwidth, bb=35 10 386 316, clip]{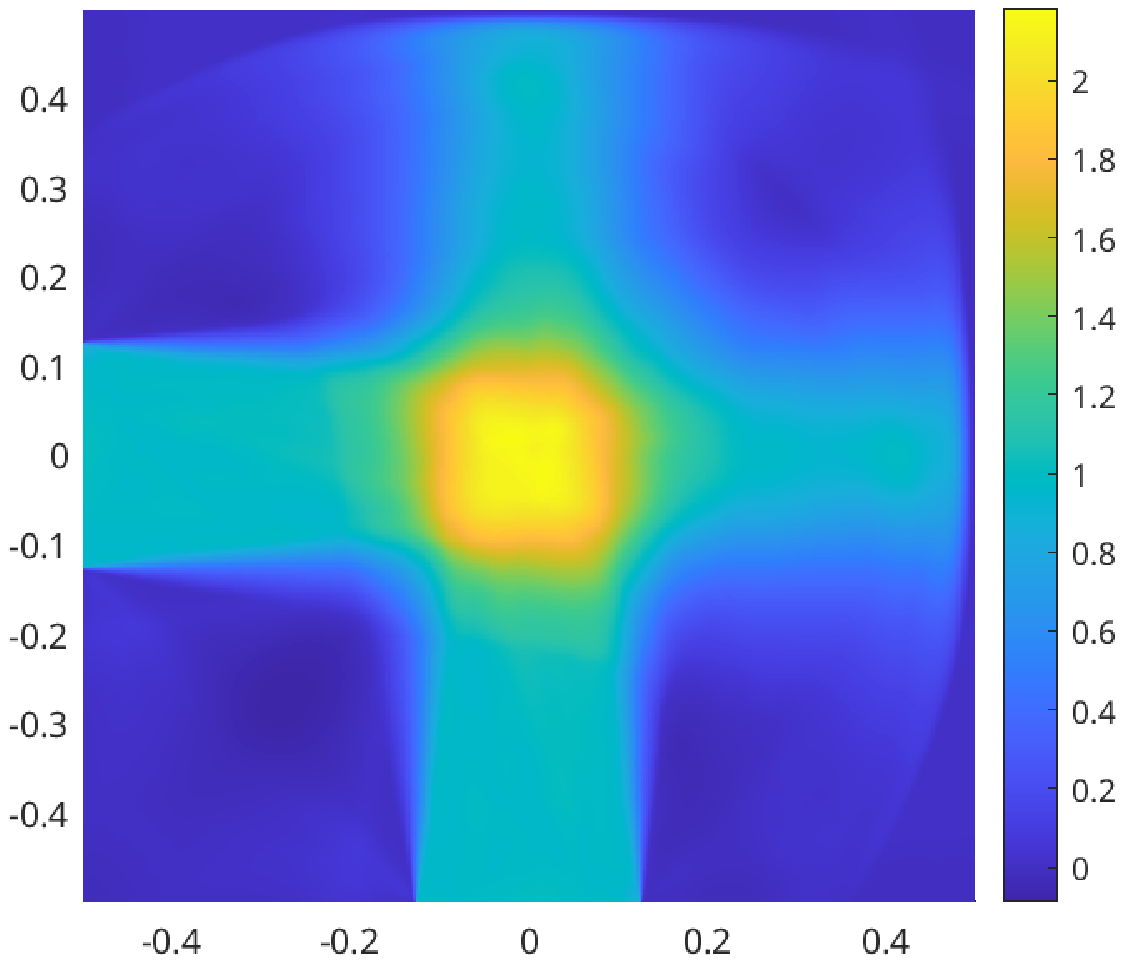}
  }
  \subfloat[$N=5$, $t = 1.1$]{%
    \includegraphics[width=.33\textwidth, bb=35 10 386 316, clip]{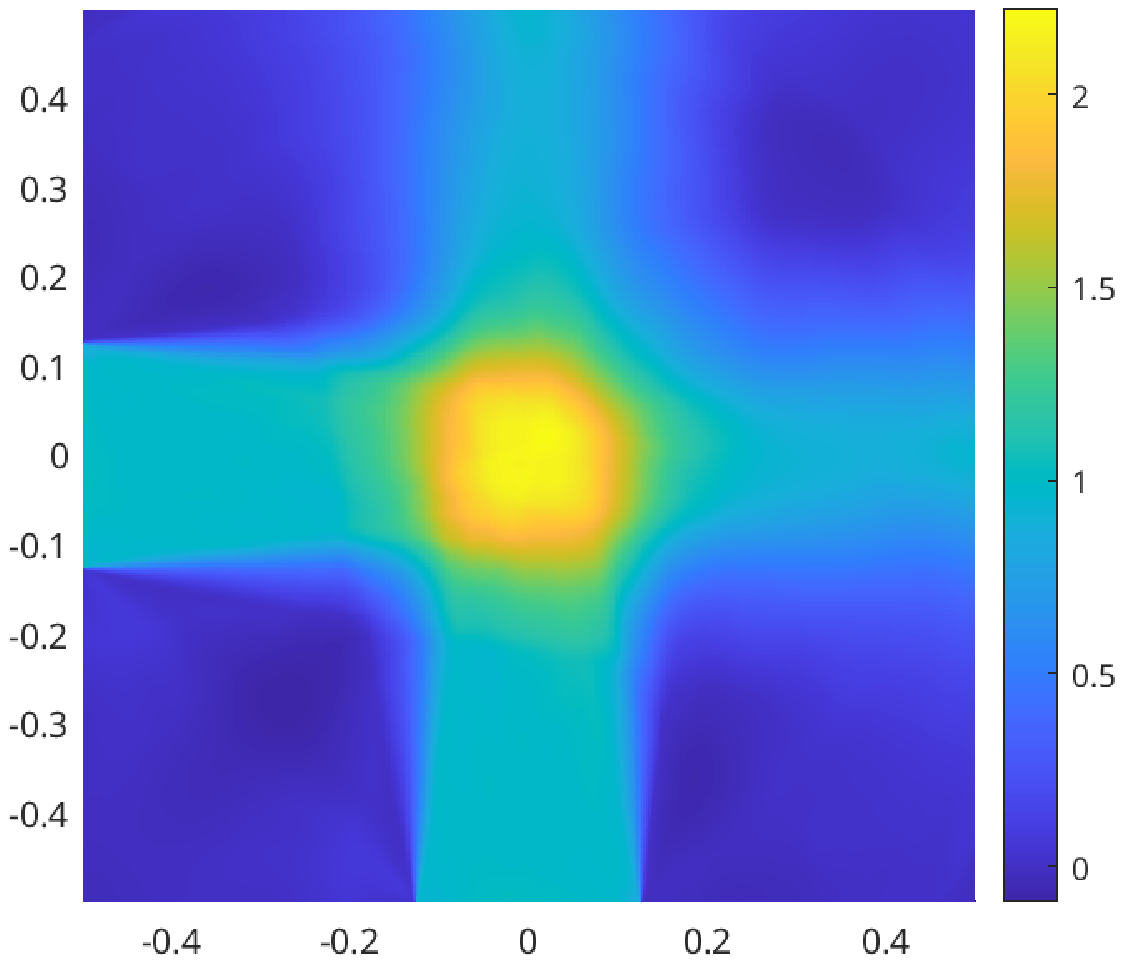}
  } \\
  \subfloat[$N=7$, $t = 0.5$]{%
    \includegraphics[width=.33\textwidth, bb=35 10 386 316, clip]{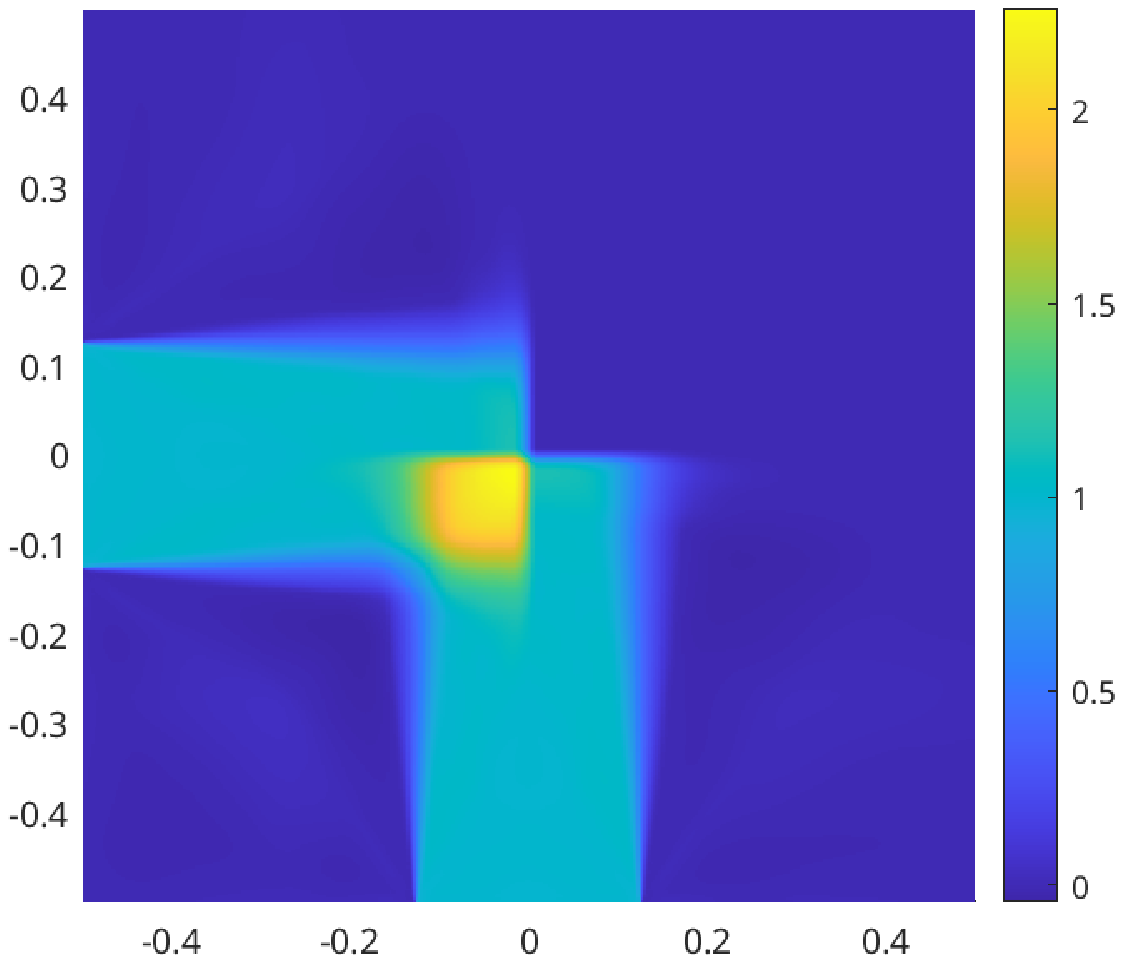}
  }
  \subfloat[$N=7$, $t = 1$]{%
    \includegraphics[width=.33\textwidth, bb=35 10 386 316, clip]{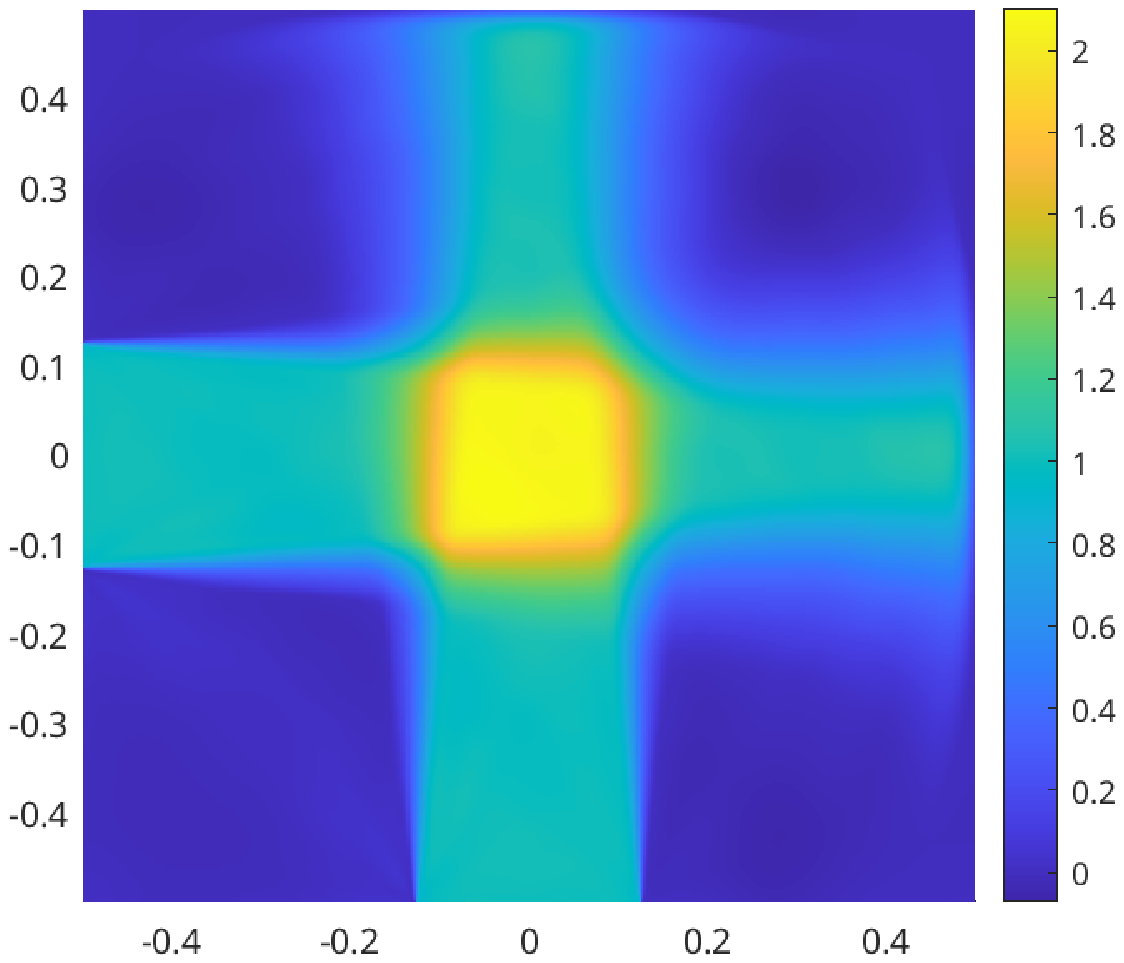}
  }
  \subfloat[$N=7$, $t = 1.1$]{%
    \includegraphics[width=.33\textwidth, bb=35 10 386 316, clip]{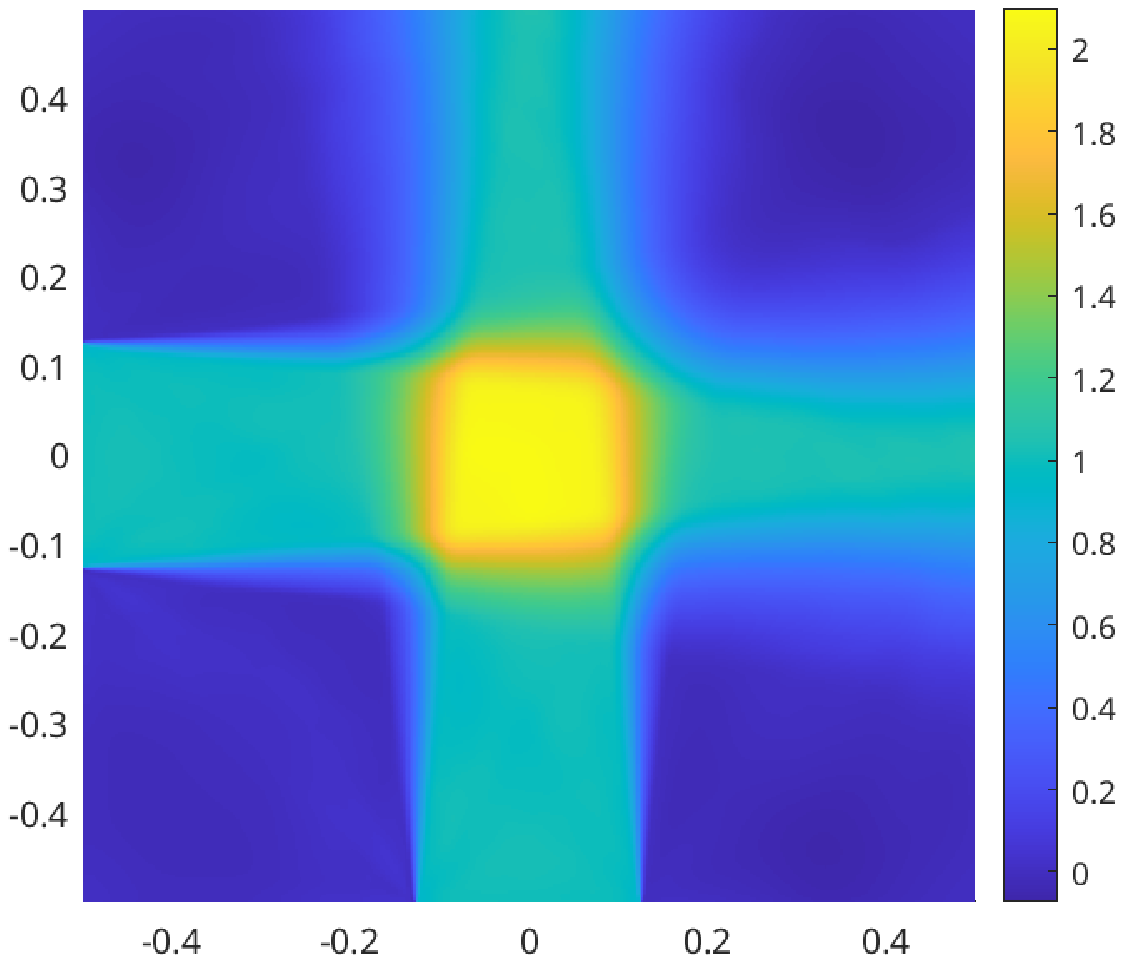}
  } \\
  \subfloat[$N=9$, $t = 0.5$]{%
    \includegraphics[width=.33\textwidth, bb=35 10 386 316, clip]{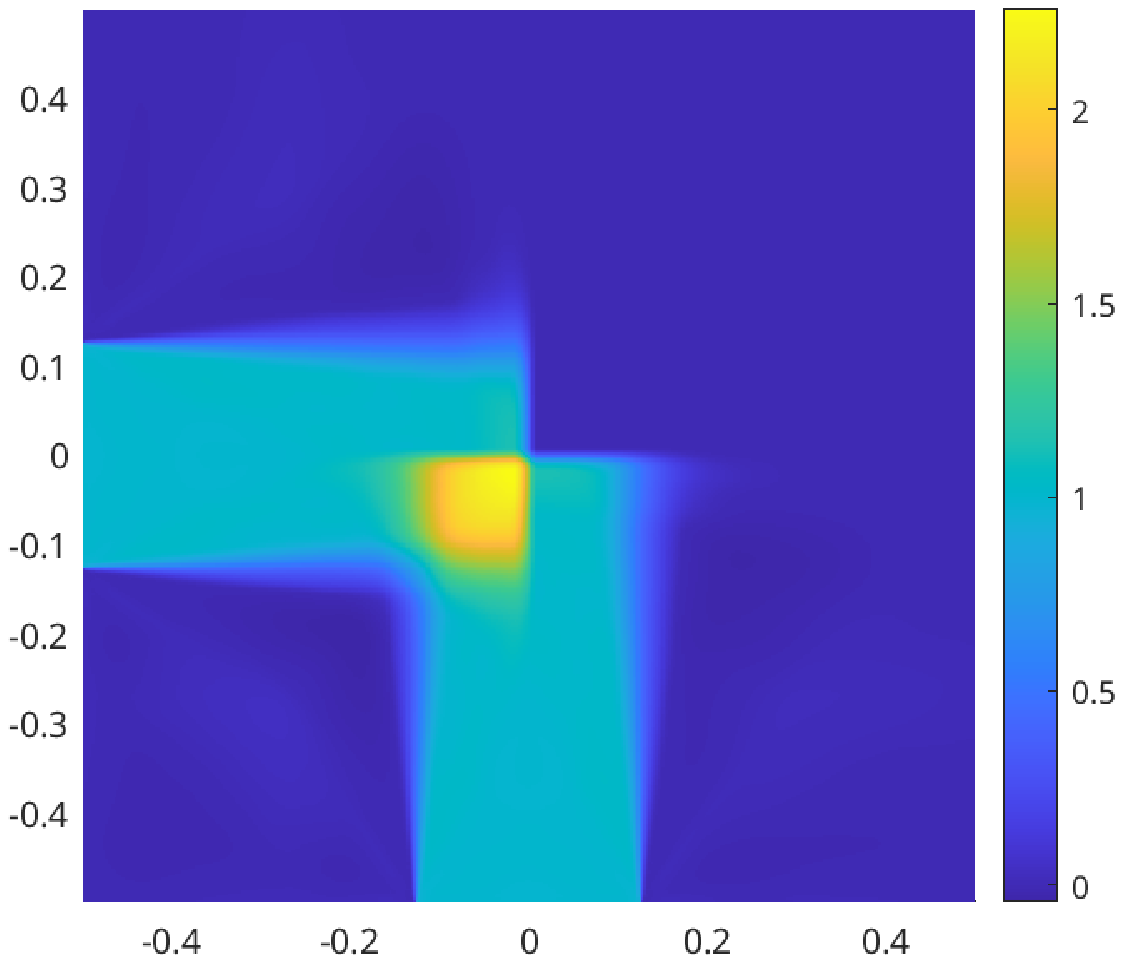}
  }
  \subfloat[$N=9$, $t = 1$]{%
    \includegraphics[width=.33\textwidth, bb=35 10 386 316, clip]{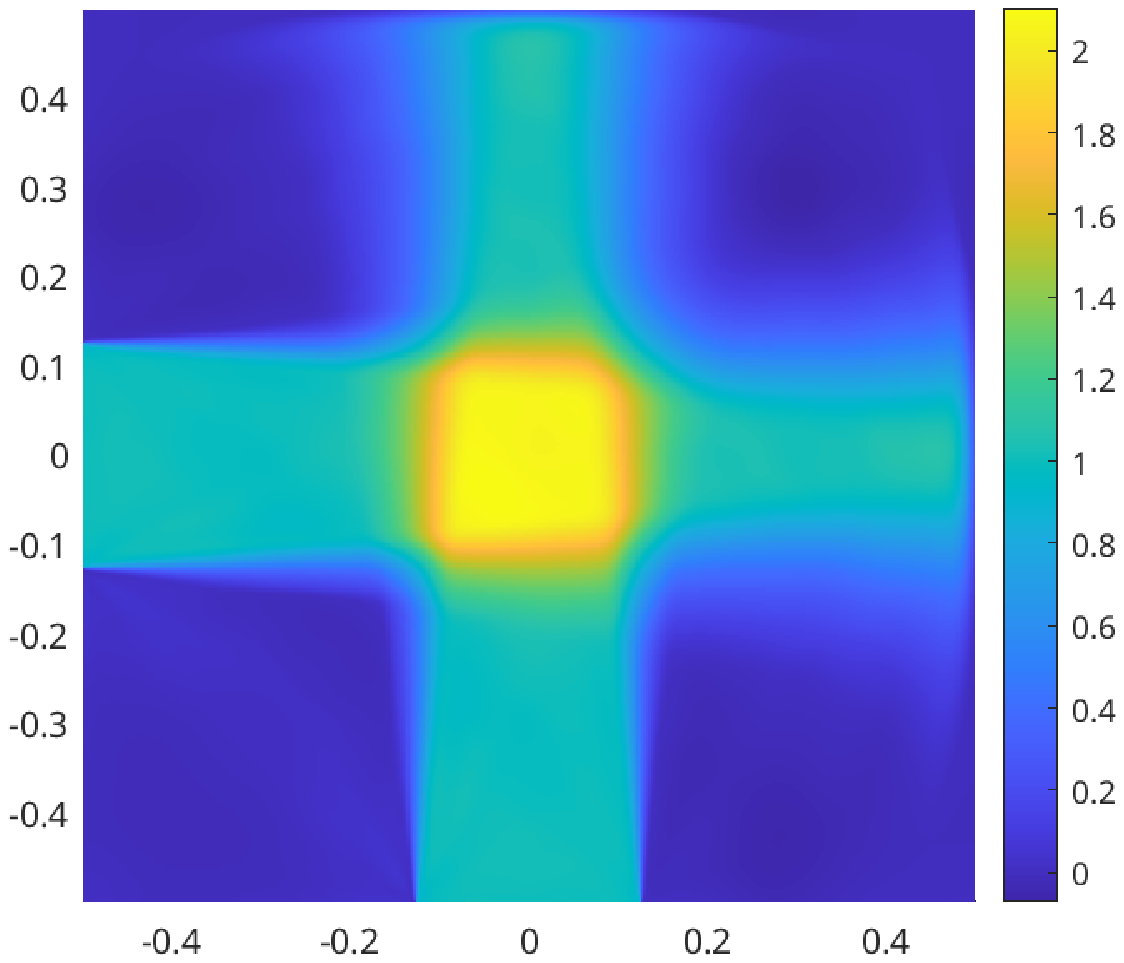}
  }
  \subfloat[$N=9$, $t = 1.1$]{%
    \includegraphics[width=.33\textwidth, bb=35 10 386 316, clip]{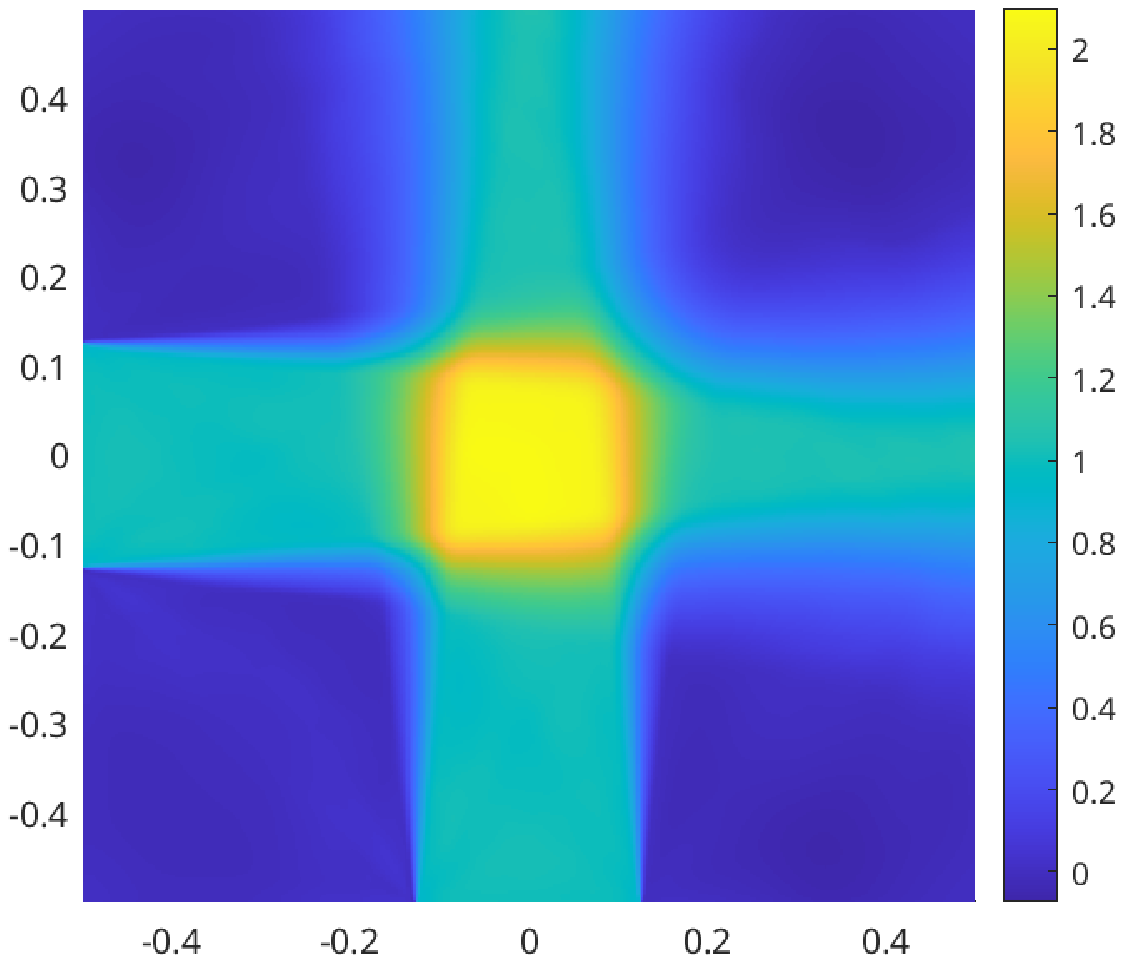}
  }
  \caption{Solution of the two-beam problem for the $\beta_{N,5}$ models}
  \label{fig:two_beam_K5}
\end{figure}

Our last test considers the same problem but with scattering coefficient $\sigma = 5$. Due to the scattering, the $M_N$ model can no longer exactly describe the solution of the radiative transfer equation, and therefore we only study the $\beta_{N,5}$ models. The results at $t = 0.3, 0.6$ and $1.0$ can be found in Figure \ref{fig:two_beam_scattering}. Because of the scattering, there are radiations pointing towards the source, causing the maximum value of the radiation density to exceed one even before the crossing, and the two beams are well mixed when they interact with each other. In this numerical example, all the $\beta_{N,5}$ models give qualitatively correct results, showing the effectiveness of the moment methods. Increasing $N$ still provides sharper solutions, especially for earlier times such as $t = 0.3$.
\begin{figure}
  \centering
  \subfloat[$N=3$, $t = 0.3$]{%
    \includegraphics[width=.33\textwidth, bb=35 10 386 316, clip]{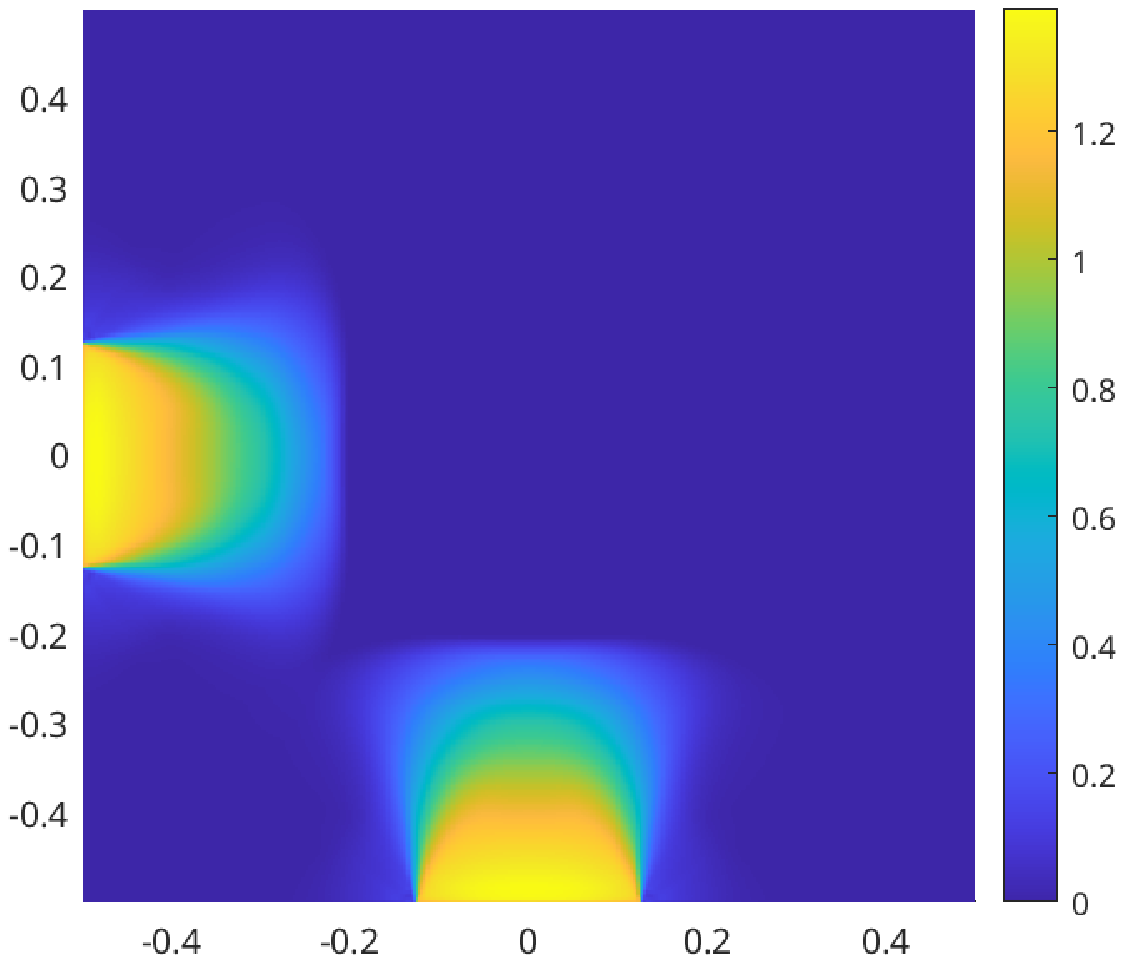}
  }
  \subfloat[$N=3$, $t = 0.6$]{%
    \includegraphics[width=.33\textwidth, bb=35 10 386 316, clip]{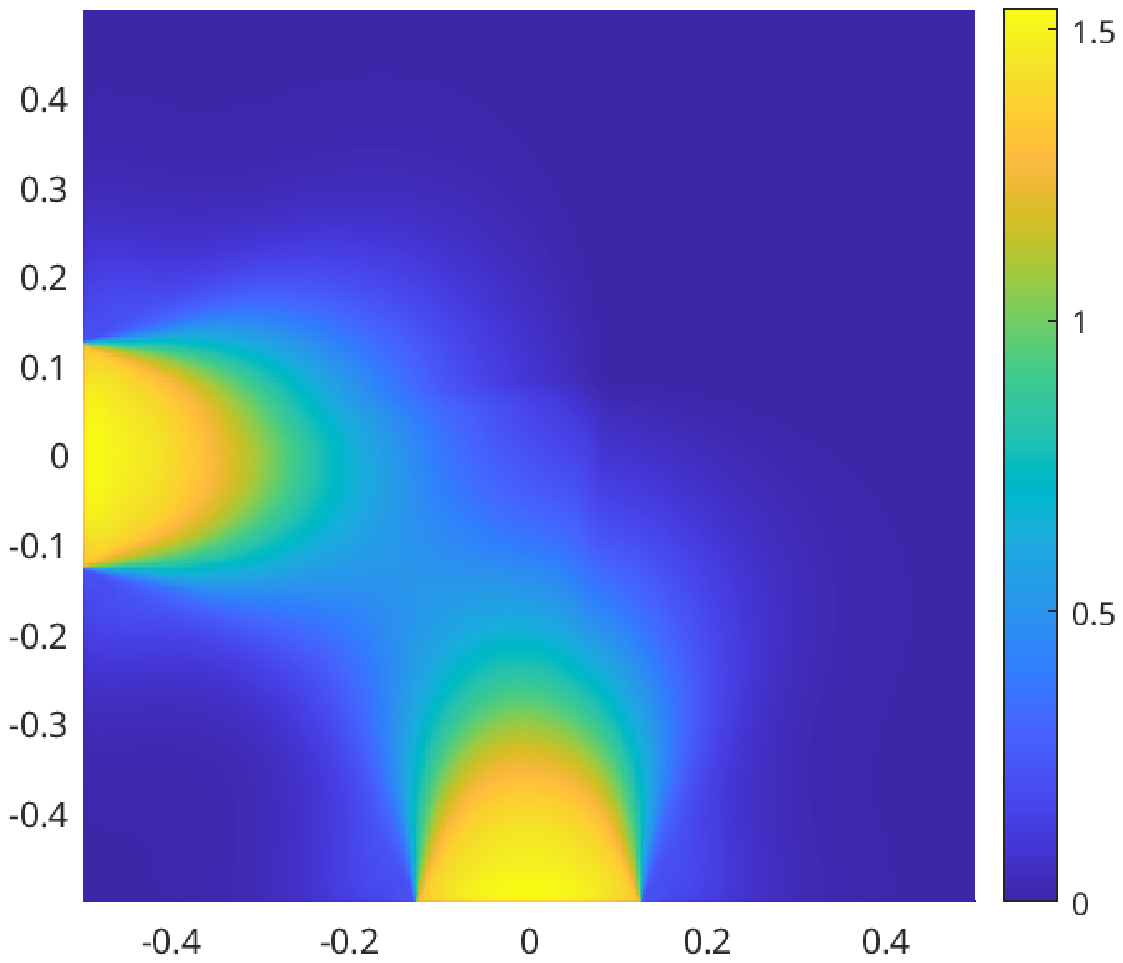}
  }
  \subfloat[$N=3$, $t = 1$]{%
    \includegraphics[width=.33\textwidth, bb=35 10 386 316, clip]{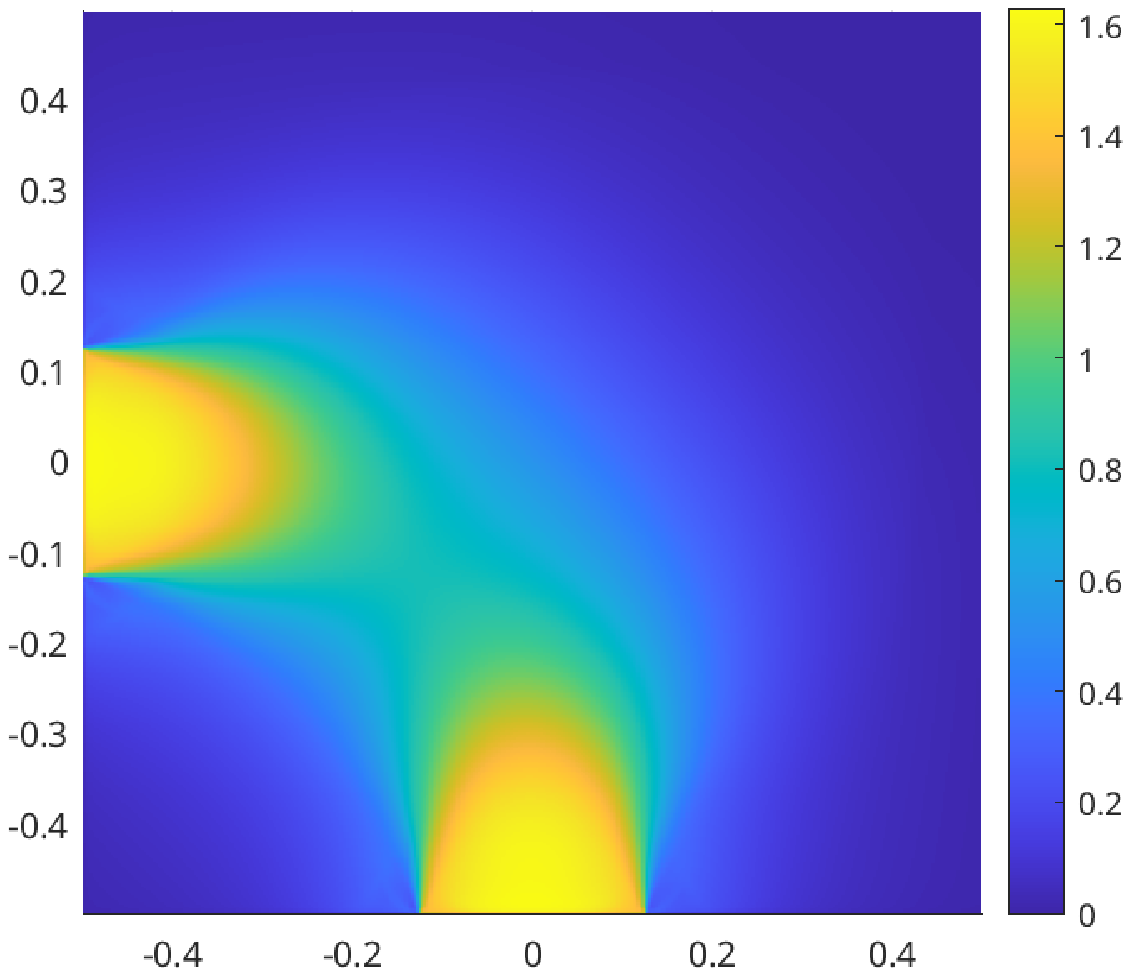}
  } \\
  \subfloat[$N=5$, $t = 0.3$]{%
    \includegraphics[width=.33\textwidth, bb=35 10 386 316, clip]{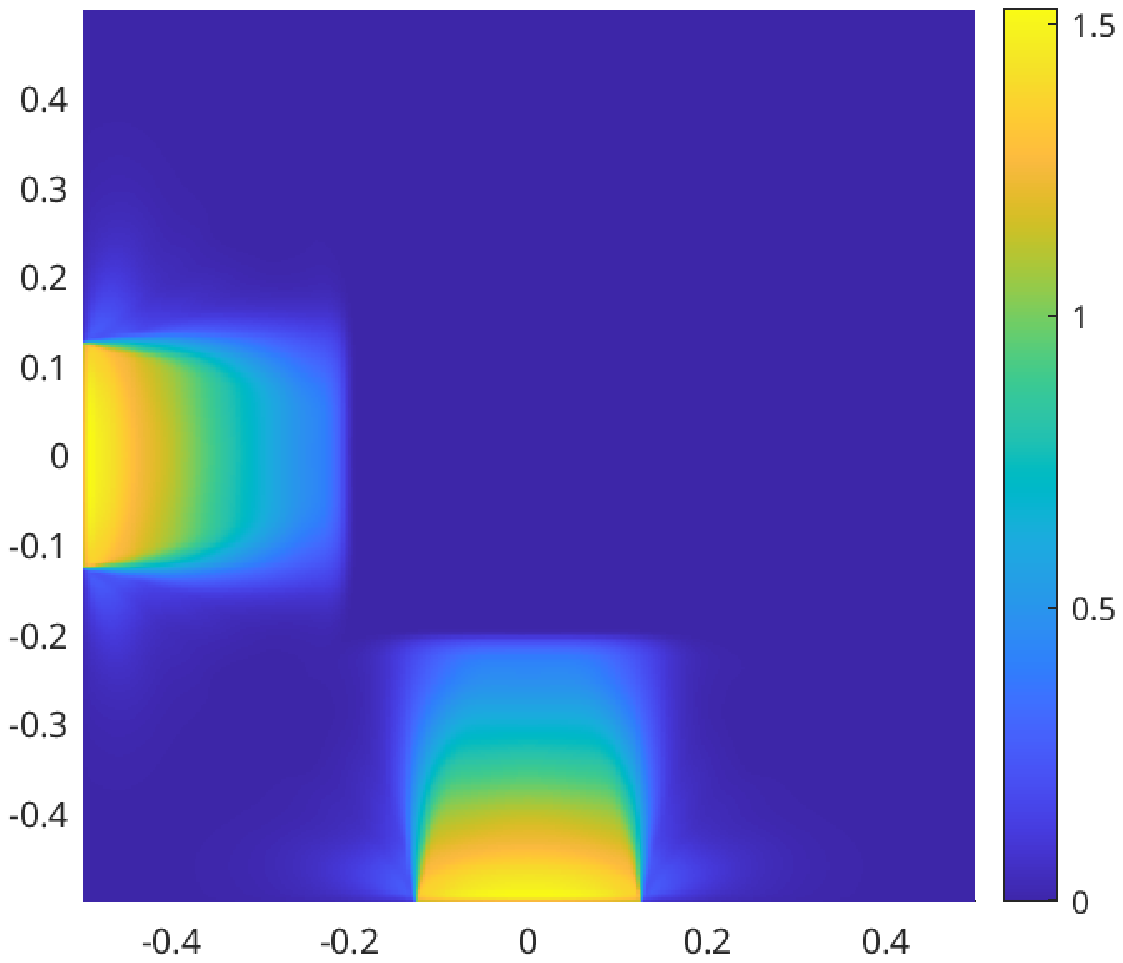}
  }
  \subfloat[$N=5$, $t = 0.6$]{%
    \includegraphics[width=.33\textwidth, bb=35 10 386 316, clip]{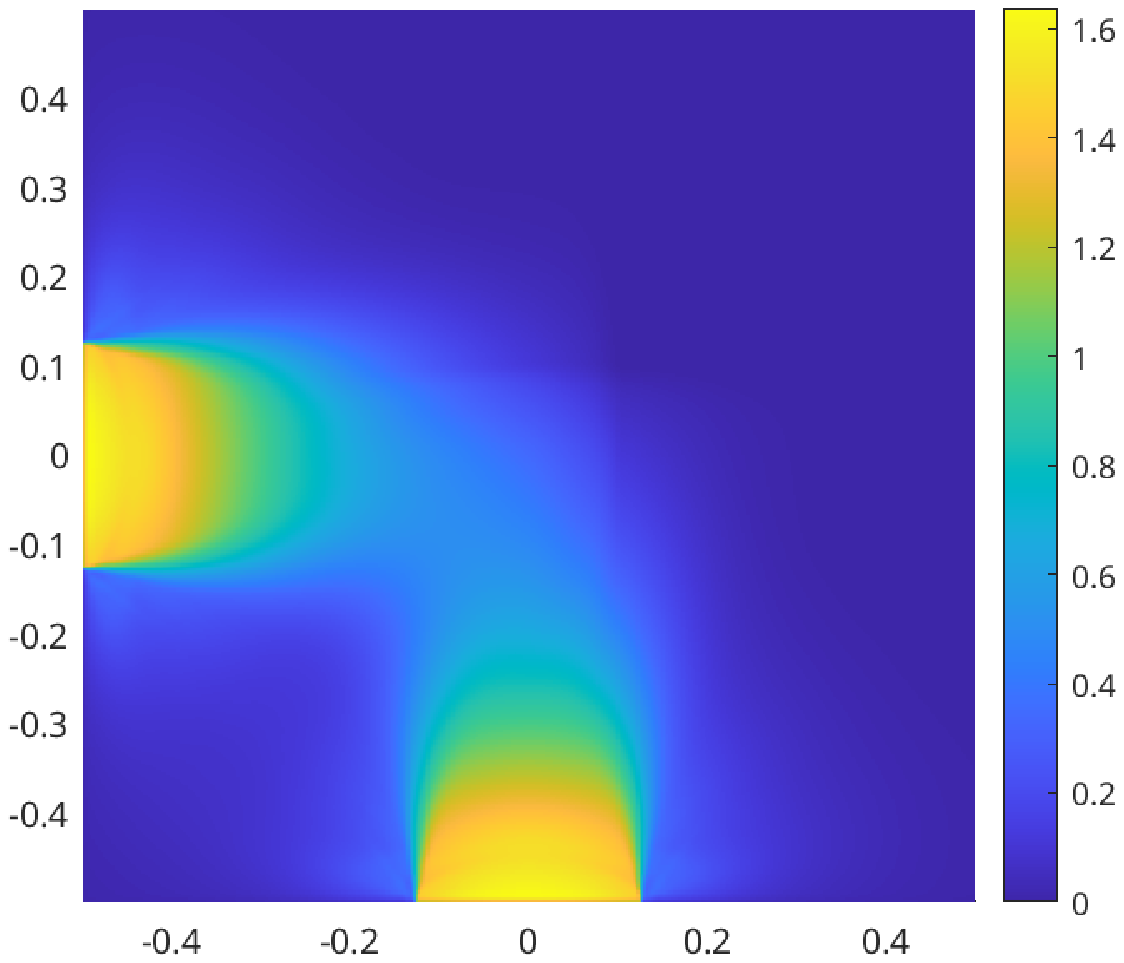}
  }
  \subfloat[$N=5$, $t = 1$]{%
    \includegraphics[width=.33\textwidth, bb=35 10 386 316, clip]{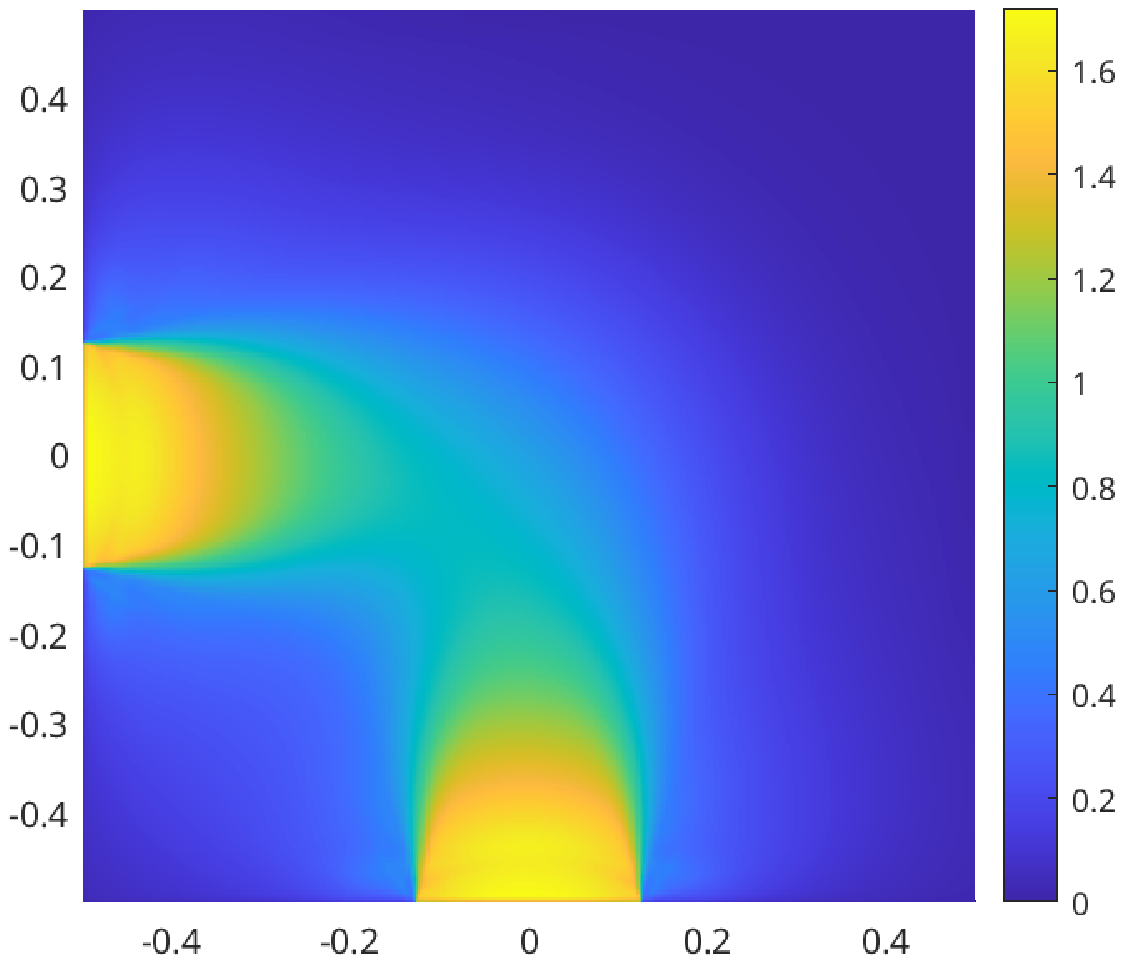}
  } \\
  \subfloat[$N=7$, $t = 0.3$]{%
    \includegraphics[width=.33\textwidth, bb=35 10 386 316, clip]{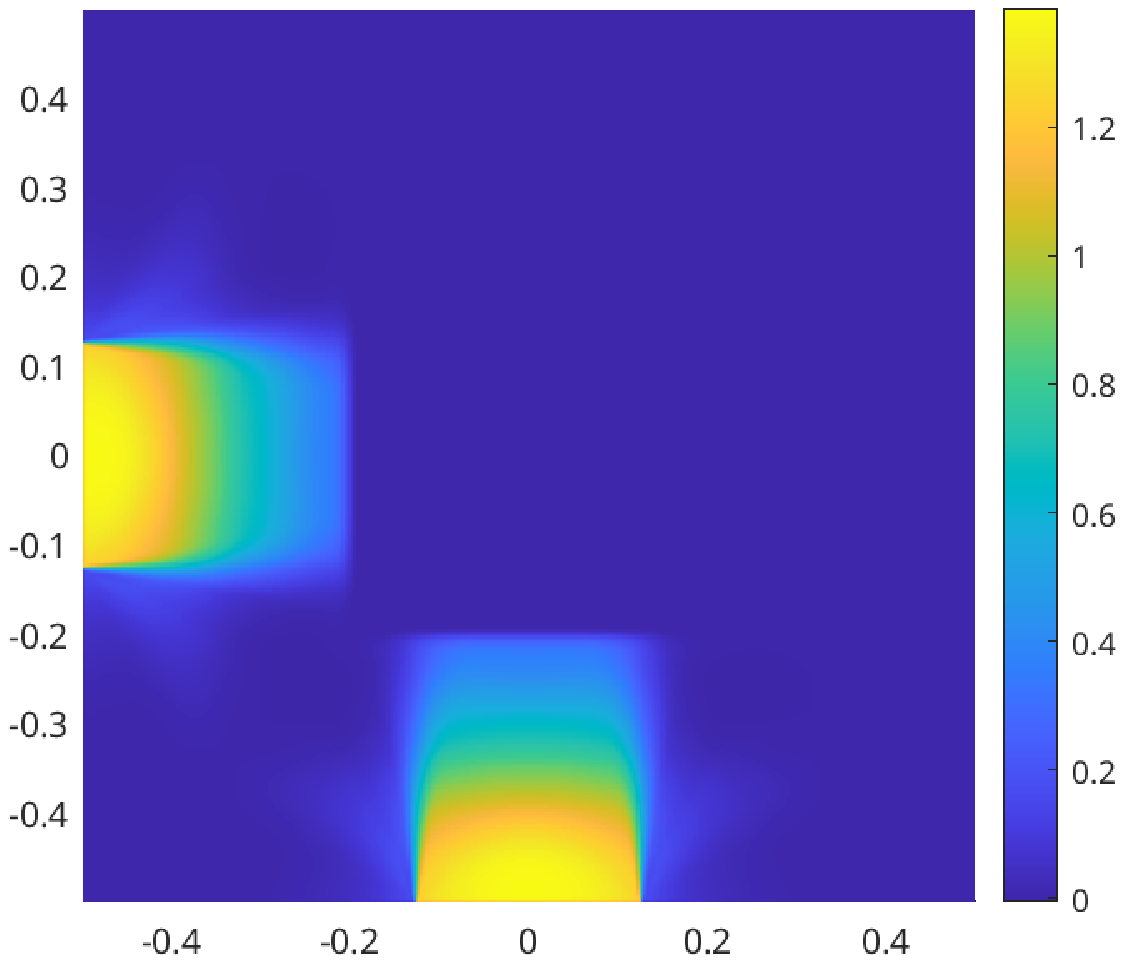}
  }
  \subfloat[$N=7$, $t = 0.6$]{%
    \includegraphics[width=.33\textwidth, bb=35 10 386 316, clip]{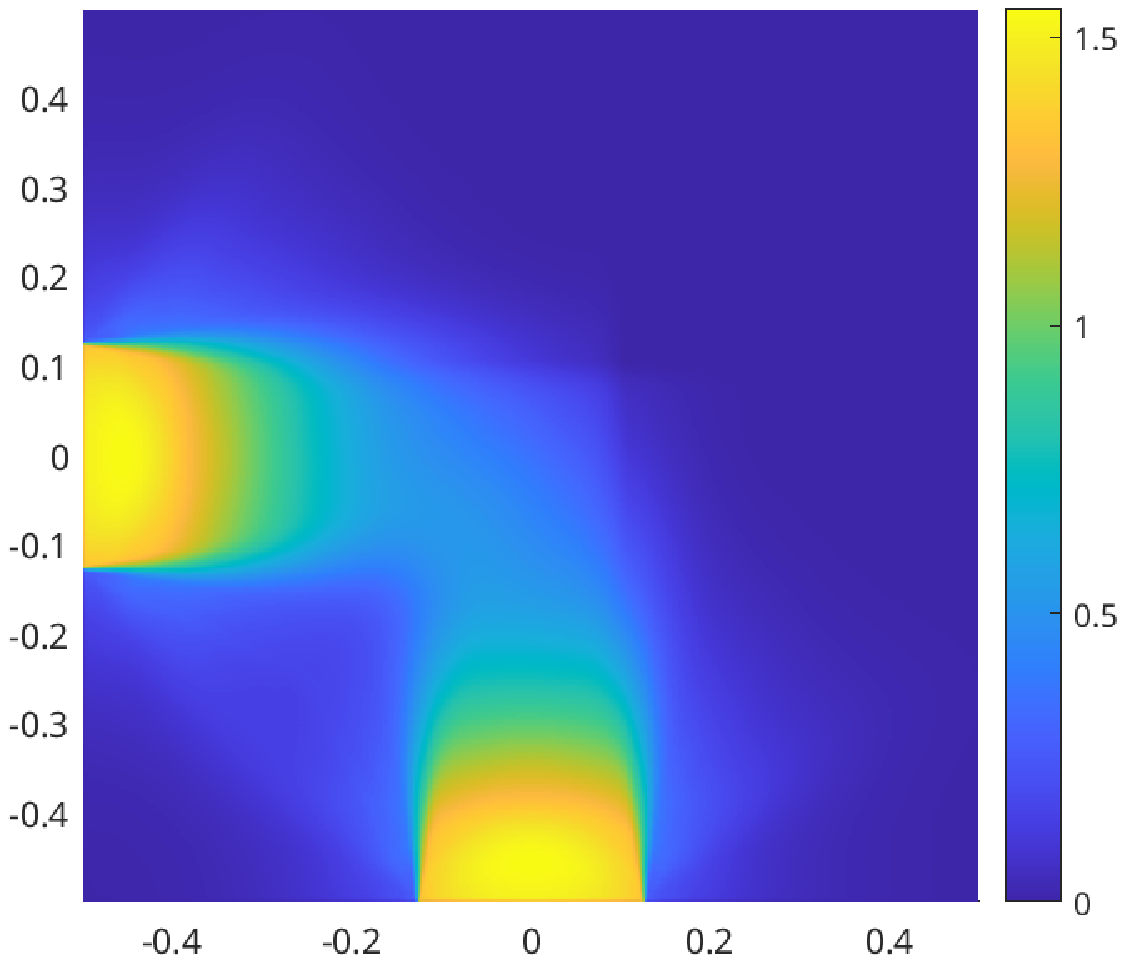}
  }
  \subfloat[$N=7$, $t = 1$]{%
    \includegraphics[width=.33\textwidth, bb=35 10 386 316, clip]{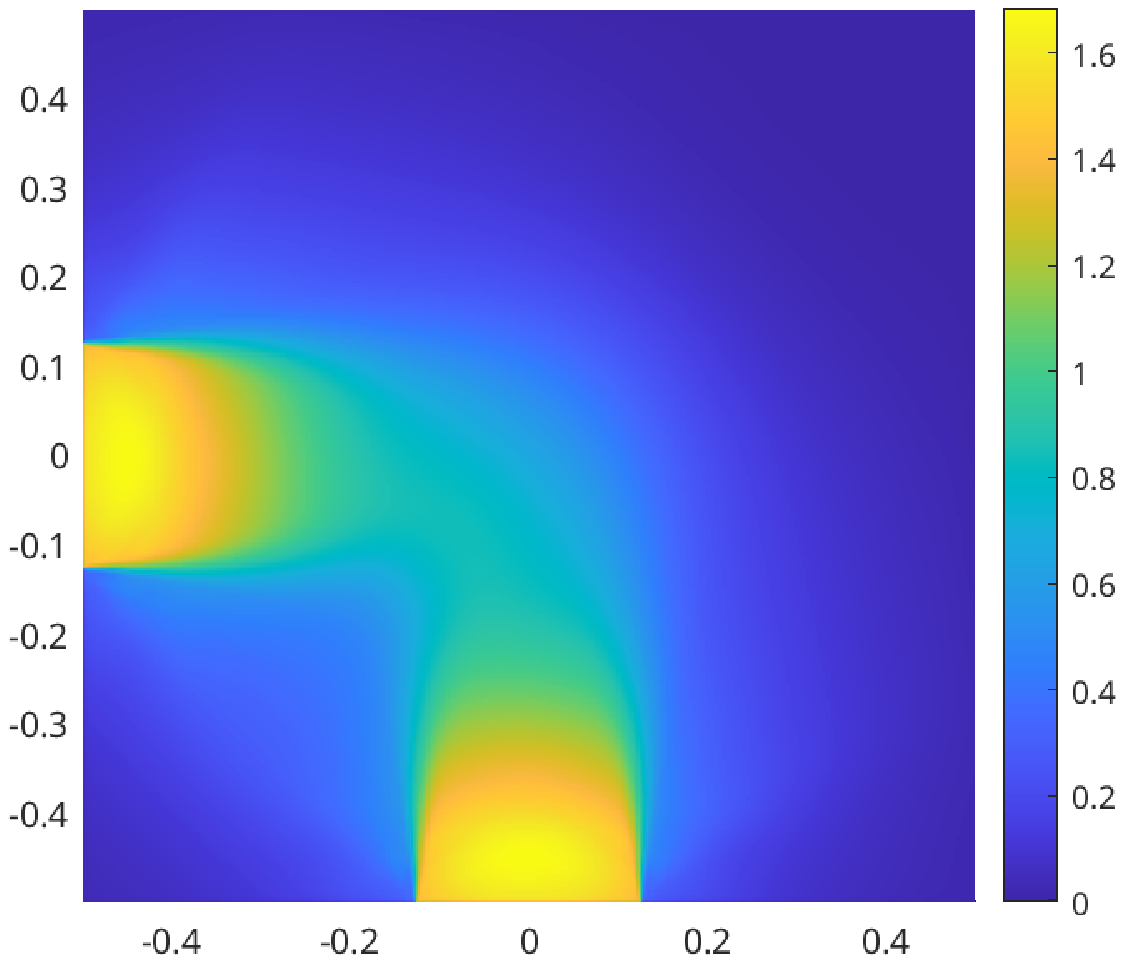}
  } \\
  \subfloat[$N=9$, $t = 0.3$]{%
    \includegraphics[width=.33\textwidth, bb=35 10 386 316, clip]{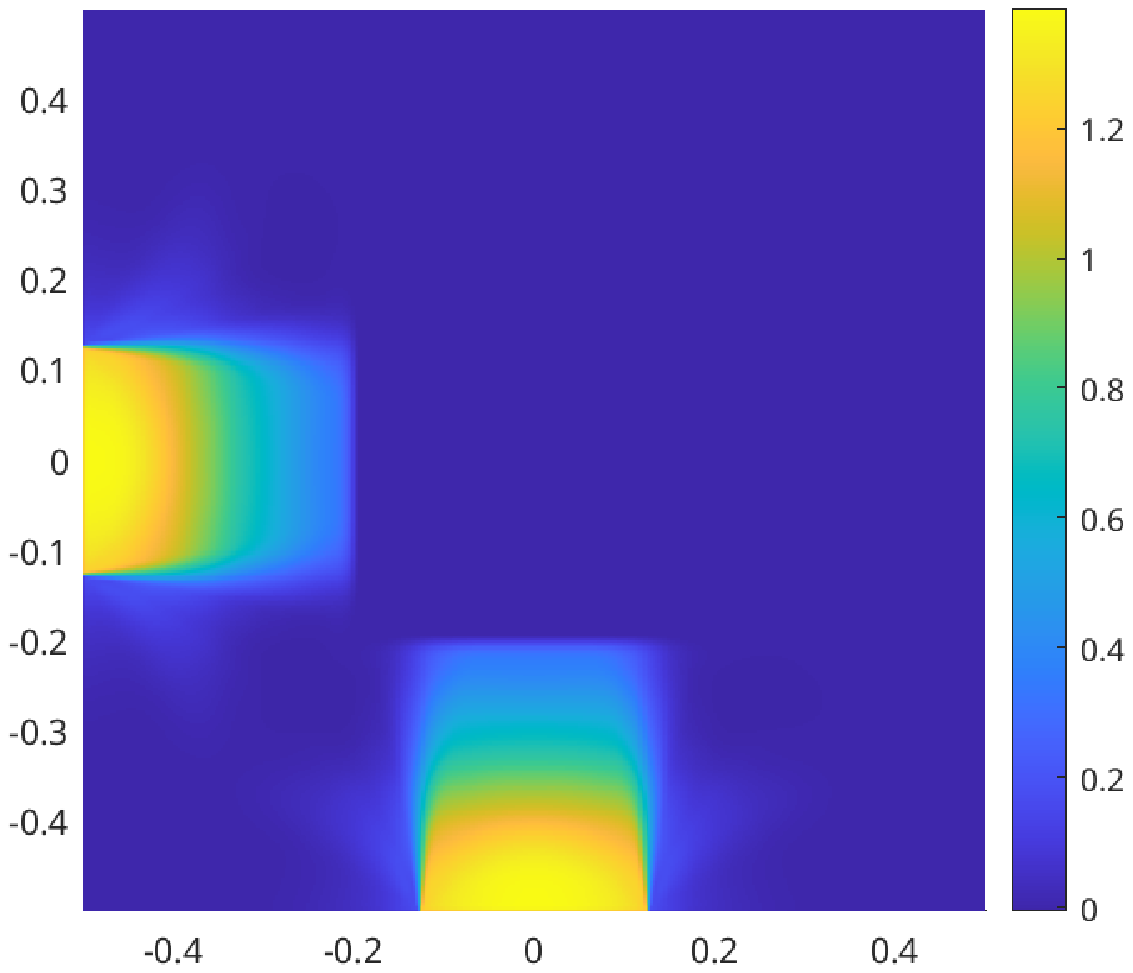}
  }
  \subfloat[$N=9$, $t = 0.6$]{%
    \includegraphics[width=.33\textwidth, bb=35 10 386 316, clip]{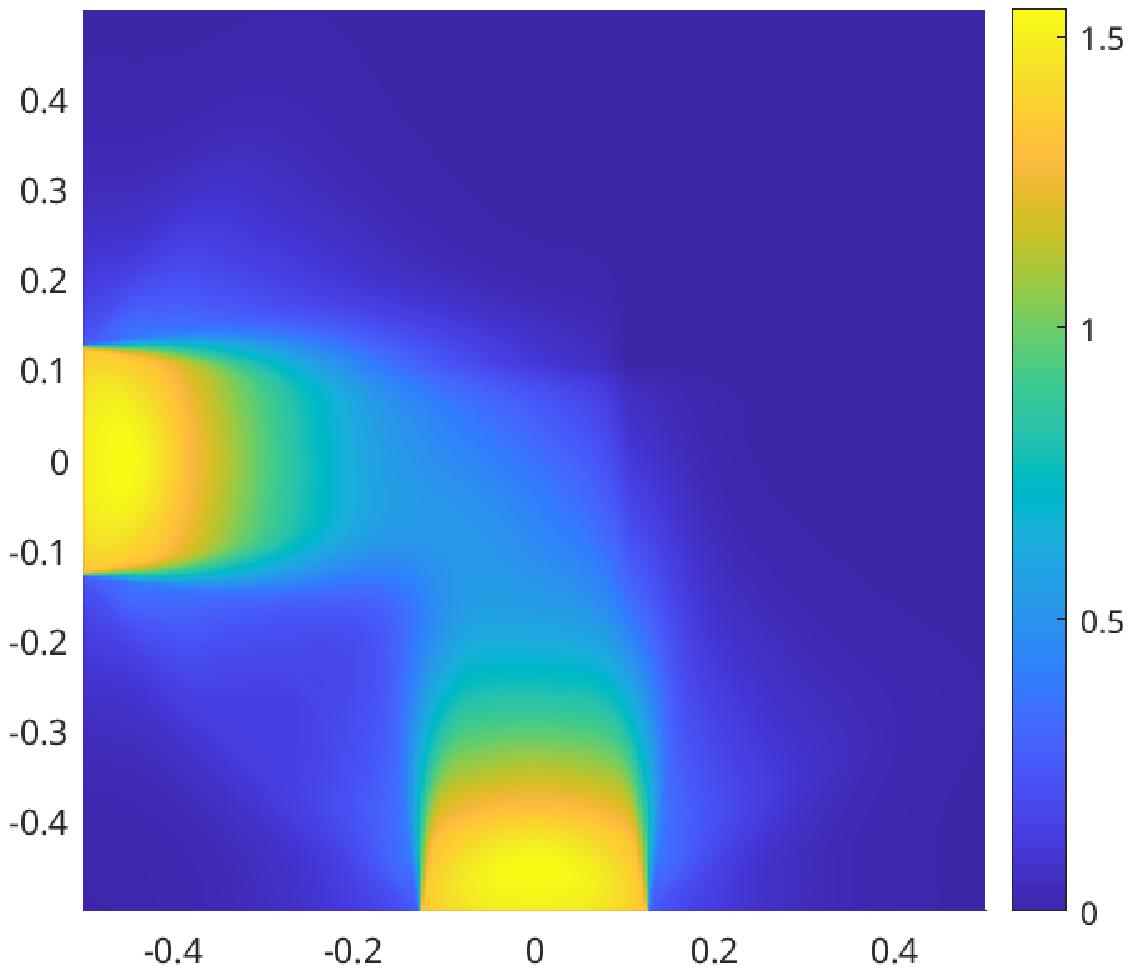}
  }
  \subfloat[$N=9$, $t = 1$]{%
    \includegraphics[width=.33\textwidth, bb=35 10 386 316, clip]{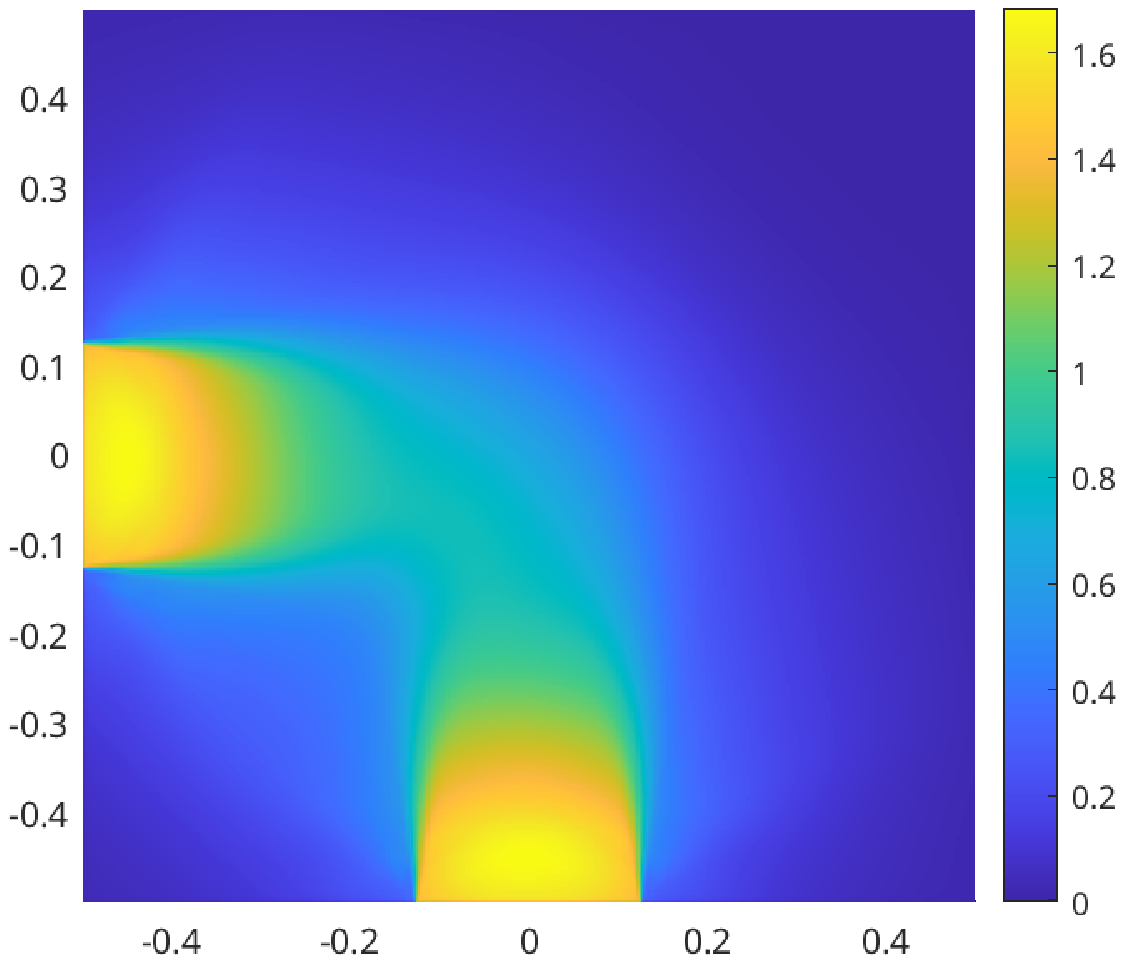}
  }
  \caption{Solution of the two-beam problem with scattering for the $\beta_{N,5}$ models}
  \label{fig:two_beam_scattering}
\end{figure}

\section{Conclusion and future works}
We have derived entropic moment equations for the radiative transfer equation using $\varphi$-divergences to define the entropy. The new moment equations, named as $\beta_{N,K}$ models, can be considered as interpolations between the $P_N$ models ($K = 1$) and the $M_N$ models ($K = \infty$). This new class of models preserves nearly all the fundamental properties of the radiative transfer equation, including conservation laws, rotational invariance and entropy dissipation. Compared with $P_N$ models, the $\beta_{N,K}$ models with $K > 1$ can better capture singular intensity functions such as beams, and the corresponding moment inversion problems of the $\beta_{N,K}$ models are easier to solve compared with the $M_N$ models. Our work has added a variety of possibilities to the family of moment equations for the radiative transfer equation, allowing considerably more options in applications.

Further improvements of the $\beta_{N,K}$ models are to be studied in our future works. In the current paper, all integrals in the moment inversion problem are computed exactly, requiring a large number of quadrature points when $N$ and $K$ are large. Better integration rules may be applied to reduce the computational cost. Additionally, to better simulate the line source problem, one can consider adding filters to the $\beta_{N,K}$ models to get smoother results.

\section*{Acknowledgements}
Zhenning Cai was supported by the Academic Research Fund of the Ministry of Education of Singapore under grant No. A-0004592-00-00. 
  
\bibliographystyle{amsplain}
\bibliography{PhiRTE}

\providecommand{\bysame}{\leavevmode\hbox to3em{\hrulefill}\thinspace}
\providecommand{\MR}{\relax\ifhmode\unskip\space\fi MR }
\providecommand{\MRhref}[2]{%
  \href{http://www.ams.org/mathscinet-getitem?mr=#1}{#2}
}
\providecommand{\href}[2]{#2}
\begin{thebibliography}{10}

\bibitem{Abdelmalik_thesis}
M.~Abdelmalik, \emph{Adaptive algorithms for optimal multiscale model
  hierarchies of the {Boltzmann} equation: {Galerkin} methods for kinetic
  theory}, Ph.D. thesis, TU Eindhoven, 2017.

\bibitem{abdelmalik}
M.~Abdelmalik and H.~van Brummelen, \emph{Moment closure approximations of the
  boltzmann equation based on $\varphi$-divergences}, J. Stat. Phys. (2016),
  no.~164, 77--104.

\bibitem{Hauck3}
G.~W. Alldredge, C.~D. Hauck, D.~P. O'Leary, and A.~L. Tits, \emph{Adaptive
  change of basis in entropy-based moment closures for linear kinetic
  equations}, J. Comp. Phys. \textbf{74} (2014), no.~4, 489--508.

\bibitem{Hauck2}
G.~W. Alldredge, C.~D. Hauck, and A.~L. Tits, \emph{High-order entropy-based
  closures for linear transport in slab geometry {II}: {A} comutational study
  of the optimization problem}, SIAM J. Sci. Comput. \textbf{34} (2012), no.~4,
  361--391.

\bibitem{Barcelo2021Fourier}
J.~A. Barcel{\'o}, M.~Folch-Gabayet, T.~Luque, S.~Pérez-Esteva, and M.~C.
  Vilela, \emph{The {F}ourier extension operator of distributions in {S}obolev
  spaces of the sphere and the {H}elmholtz equation}, Proceedings of the Royal
  Society of Edinburgh: Section A Mathematics \textbf{151} (2021), no.~6,
  1768–1789.

\bibitem{astro2}
H.~Bloch, P.~Tremblin, M.~González, T.~Padioleau, and E.~Audit, \emph{A
  high-performance and portable asymptotic preserving radiation hydrodynamics
  code with the {$M_1$} model}, Astronomy \& astrophysics \textbf{646} (2021),
  1--17.

\bibitem{Brouwer1912}
L.~E.~J. Brouwer, \emph{Beweis der invarianz des {$n$}-dimensionalen gebiets},
  Mathematische Annalen \textbf{71} (1911), 305--313.

\bibitem{Camminady2019}
T.~Camminady, M.~Frank, K.~K{\"u}pper, and J.~Kusch, \emph{Ray effect
  mitigation for the discrete ordinates method through quadrature rotation}, J.
  Comput. Phys. \textbf{382} (2019), 105--123.

\bibitem{spectral_meth2}
C.~Canuto, M.~Y. Hussaini, A.~Quarteroni, and T.~A. Zang, \emph{Spectral
  methods: Fundamentals in single domains}, Springer-Verlag, 2006.

\bibitem{csiszar}
I.~Csisz{\'a}r, \emph{A class of measures of informativity of observation
  channels}, Periodica Mathematica Hungarica \textbf{2} (1972), 191--213.

\bibitem{Dautray-Lions}
R.~Dautray and J.-L. Lions, \emph{Mathematical analysis and numerical methods
  for science and technology: Volume 6, evolution problems {II}}, Springer,
  2000.

\bibitem{dreyer1987}
W.~Dreyer, \emph{Maximisation of the entropy in non-equilibrium}, J. Phys. A:
  Math. Gen. \textbf{20} (1987), 6505--6517.

\bibitem{dubroca_feugeas}
B.~Dubroca and J.-L. Feugeas, \emph{Entropic moment closure hierarchy for the
  radiative transfer equation}, C. R. Acad. Sci. Paris Ser. I \textbf{329}
  (1999), 915--920.

\bibitem{duclous_thesis}
R.~Duclous, \emph{Modelling and numerical simulation of the multi-scale kinetic
  electron transport}, Ph.D. thesis, Universit\'e de Bordeaux, 2009.

\bibitem{Ganapol1999homogeneous}
Barray~D Ganapol, Randall~S Baker, Jon~A Dahl, and Raymond~E Alcouffe,
  \emph{Homogeneous infinite media time-dependent analytical benchmarks}, Tech.
  report, Los Alamos National Laboratory, 2001.

\bibitem{Garrett2013comparison}
C.~Krœistopher Garrett and Cory~D. Hauck, \emph{A comparison of moment
  closures for linear kinetic transport equations: The line source benchmark},
  Transport Theory and Statistical Physics \textbf{42} (2013), no.~6--7,
  203--235.

\bibitem{M1_relativist}
T.~Hanawa and E.~Audit, \emph{Reformulation of the {$M_1$} model of radiative
  transfer}, J. Quant. Spectros. Radiat. Transfer (2014), no.~145, 9--16.

\bibitem{Hauck}
C.~D. Hauck, \emph{High-order entropy-based closures for linear transport in
  slab geometry}, Commun. Math. Sci \textbf{9} (2011), no.~1, 187--205.

\bibitem{hauck_lev_tits}
C.~D. Hauck, C.~D. Levermore, and A.~L. Tits, \emph{Convex duality and
  entropy-based moment closures: {C}haracterizing degenerate densities}, SIAM
  J. Control Optim. \textbf{47} (2008), no.~4, 1977--2015.

\bibitem{spectral_meth1}
J.~S. Hesthaven, S.~Gottlieb, and D.~Gottlieb, \emph{Spectral methods for
  time-dependent problems}, Cambridge, 2009.

\bibitem{junk}
M.~Junk, \emph{Maximum entropy for reduced moment problems}, Math. Mod. Meth.
  in Appl. Sci. \textbf{10} (1998), no.~1001--1028, 2000.

\bibitem{Kawashima}
S.~Kawashima and W.-A. Yong, \emph{Dissipative structure and entropy for
  hyperbolic systems of balance laws}, Arch. Rational Mech. Anal. \textbf{174}
  (2004), 345--364.

\bibitem{Kershaw}
D.~Kershaw, \emph{Flux limiting nature's own way}, Tech. report, Lawrence
  Livermore Laboratory, 1976.

\bibitem{kuepper_thesis}
K.~K{\"u}pper, \emph{Models, numerical methods, and uncertainty quantification
  for radiation therapy}, Ph.D. thesis, RWTH Aachen University, 2016.

\bibitem{Laiu2016}
M.~Paul Laiu, C.~D. Hauck, R.~G. McClarren, D.~P. O'Leary, and A.~L. Tits,
  \emph{Positive filtered {$P_N$} moment closures for linear kinetic
  equations}, SIAM J. Numer. Anal. \textbf{54} (2016), no.~6, 3214--3238.

\bibitem{Lasserre_book}
J.-B. Lasserre, \emph{Moment, positive polynomials, and their applications},
  Imperial college press, 2009.

\bibitem{Lebedev1976quadratures}
V.I. Lebedev, \emph{Quadratures on a sphere}, USSR Comput. Math. Math. Phys.
  \textbf{16} (1976), no.~2, 10--24.

\bibitem{Lebedev1999quadrature}
V.I. Lebedev and D.N. Laikov, \emph{A quadrature formula for the sphere of the
  131st algebraic order of accuracy}, Doklady Mathematics \textbf{59} (1999),
  no.~3, 477--481.

\bibitem{levermore_eddington}
C.~D. Leveremore, \emph{Relating {Eddington} factors to flux limiters}, J.
  Quant. Spectros. Radiat. Transfer \textbf{31} (1984), 149--160.

\bibitem{levermore}
C.~D. Levermore, \emph{Moment closure hierarchies for kinetic theories}, J.
  Stat. Phys. \textbf{83} (1996), no.~5--6, 1021--1065.

\bibitem{Li_B2}
R.~Li and W.~Li, \emph{{3D $B_2$} model for radiative transfer equation}, Int.
  J. Numer. Anal. Modeling \textbf{17} (2020), no.~1, 118--150.

\bibitem{astro1}
E.~Audit M.~González and P.~Huynh, \emph{Heracles: a three-dimensional
  radiation hydrodynamics code}, Astronomy \& astrophysics \textbf{64} (2007),
  429--435.

\bibitem{mihalas_book}
D.~Mihalas and B.~R.~W. Mihalas, \emph{Foundations of radiation hydrodynamics},
  Oxford, 1983.

\bibitem{Minerbo}
G.~N. Minerbo, \emph{Maximum entropy {Eddington} factors}, J. Quant. Spectros.
  Radiat. Transfer \textbf{20} (1978), 541--545.

\bibitem{Monreal_thesis}
P.~Monreal, \emph{Moment realizability and {Kershaw} closures in radiative
  transfer}, Ph.D. thesis, RWTH Aachen University, 2012.

\bibitem{page_thesis}
J.~Page, \emph{Développement et validation de l'application de la force de
  lorentz dans le modèle aux moments entropiques m1. Étude de l'effet du
  champ magnétique sur le dépôt de dose en radiothérapie externe}, Ph.D.
  thesis, Universit\'e de Bordeaux, 2018.

\bibitem{pichard_thesis}
T.~Pichard, \emph{Mathematical modeling of dose deposition in photontherapy and
  protontherapy}, Ph.D. thesis, Universit\'e de Bordeaux, 2016.

\bibitem{PiN}
T.~Pichard, \emph{A moment closure based on a projection on the boundary of the
  realizability domain: 1d case}, Kin. rel. models (2020), no.~13, 1243--1280.

\bibitem{pichard_M2}
T.~Pichard, G.~W. Alldredge, S.~Brull, B.~Dubroca, and M.~Frank, \emph{An
  approximation of the {$M_2$} closure: application to radiotherapy dose
  simulation}, J. Sci. Comput. (2017), no.~71, 71--108.

\bibitem{pomraning_book}
G.~C. Pomraning, \emph{Equations of radiation hydrodynamics}, Pergamon, 1973.

\bibitem{groth_M2}
J.~A.~R. Sarr and C.~P.~T. Groth, \emph{A second-order maximum-entropy inspired
  interpolative closure for radiative heat transfer in gray participating
  media}, J. Quant. Spectros. Radiat. Transfer (2020), 107238.

\bibitem{Schmuedgen_book}
K.~Schmuedgen, \emph{The moment problem}, Springer, 2017.

\bibitem{Schneider_KN}
F.~Schneider, \emph{Kershaw closures for linear transport equations in slab
  geometry i: model derivation}, J. Comput. Phys. (2016), no.~322, 905--919.

\bibitem{schneider}
J.~Schneider, \emph{Entropic approximation in kinetic theory}, ESAIM: M2AN
  \textbf{38} (2004), no.~3, 541--561.

\end{thebibliography}

\end{document}